\documentclass[11pt]{article}
\usepackage{amssymb,amsmath,amsthm}

\newtheorem{theorem}{Theorem}[section]
\newtheorem{lemma}{Lemma}[section]

\theoremstyle{definition}

\begin{document}
\large

{\centering

{\LARGE \bf On absolute continuity of the spectrum of a d-dimensional
periodic magnetic Dirac operator}
\vskip 1.0cm

{\Large L.I.~Danilov }
\vskip 0.5cm

{\large Physical-Technical Institute,}\\
{\large 426000, Izhevsk, Kirov street, 132, Russia}\\
{\large e-mail: danilov@otf.pti.udm.ru}

\vskip 2.0cm
{\Large Abstract}
\vskip 0.5cm

}

In this paper, for $d\geqslant 3$, we prove the absolute continuity of the
spectrum of a $d$-dimensional periodic Dirac operator with some discontinuous
magnetic and electric potentials. In particular, for $d=3$, electric potentials
from Zygmund classes $L^3\ln ^{1+\delta }L(K)$, $\delta >0$, and also ones
with Coulomb singularities, with constraints on charges depending on the
magnetic potential, are admitted (here $K$ is the fundamental domain of the
period lattice).

\section*{Introduction and main results}
Let ${\mathcal M}_M$, $M\in {\mathbb N}$, be the linear space of complex 
$(M\times M)$-matrices, let ${\mathcal S}_M$ be the set of Her\-mi\-tian matrices from 
${\mathcal M}_M\, $, and let the matrices $\widehat \alpha _j\in {\mathcal S}_M\, $, 
$j=1,\dots ,d$ ($d\geqslant 2$), satisfy the commutation relations 
$
\widehat \alpha _j\widehat \alpha _l+\widehat \alpha _l\widehat \alpha _j=
2\delta _{jl}\widehat I,
$
where $\widehat I\in {\mathcal M}_M$ is the identity matrix and $\delta _{jl}$ is the 
Kronecker delta. Denote
$$
{\mathcal S}^{(s)}_M=\{ \widehat L\in {\mathcal S}_M:\widehat L\widehat \alpha _j=
(-1)^s\widehat \alpha _j\widehat L \text{\ for\ all\ } j=1,\dots ,d\} \, ,\ s=0,1\, .
$$
We consider the $d$-dimensional Dirac operator
$$
\widehat {\mathcal D}+\widehat W=-i\, \sum\limits_{j=1}^d\widehat \alpha _j\, \frac 
{\partial}{\partial x_j}+\widehat W(x)\, ,\ x\in {\mathbb R}^d,\,  \eqno (0.1)
$$
with a periodic matrix function $\widehat W:{\mathbb R}^d\to {\mathcal S}_M\, 
$, $d\geqslant 2$ ($i^2=-1$), with a period lattice $\Lambda \subset {\mathbb R}^d$.
In particular, the operator (0.1) can have the form
$$
\widehat {\mathcal D}+\widehat W=\sum\limits_{j=1}^d\widehat \alpha _j\, \bigl( 
-i\, \frac {\partial}{\partial x_j}-A_j\bigr) +\widehat V\, ,\ \ \ 
\widehat V=\widehat V^{\, (0)}+\widehat V^{\, (1)}, \eqno (0.2) 
$$
where the components $A_j$ of the magnetic potential $A:{\mathbb 
R}^d\to {\mathbb R}^d$ and the matrix functions $\widehat V^{\, (s)}:{\mathbb R}^d
\to {\mathcal S}^{(s)}_M$, $s=0,1$, are also periodic with the period lattice 
$\Lambda \subset {\mathbb R}^d$. In the sequel, the matrix functions $\widehat 
V^{\, (s)}$, $s=0,1$, will be also chosen in the form
$$
\widehat V^{\, (0)}=V\widehat I\, ,\ \widehat V^{\, (1)}=V_1\widehat \beta \, ,  
\eqno (0.3)
$$
where $V,V_1$ are $\Lambda$-periodic real-valued functions and $\widehat \beta \in 
{\mathcal S}^{(1)}_M$ is a Hermitian matrix with $\widehat \beta ^{\, 2}=
\widehat I\, $, and in the particular form 
$$
\widehat V^{\, (0)}=V\widehat I\, ,\ \widehat V^{\, (1)}=m\widehat \beta \, ,  
\eqno (0.4)
$$
where $V:{\mathbb R}^d\to {\mathbb R}$ is a $\Lambda $-periodic electric potential 
and $m\in {\mathbb R}$.

The coordinates in ${\mathbb R}^d$ are taken relative to an orthogonal basis 
$\{ {\mathcal E}_j\} $ ($|{\mathcal E}_j|=1$, $j=1,\dots ,d$; $|.|$ and $(.,.)$ 
are the length and the scalar product of vectors in ${\mathbb R}^d$), $A_j(x)=
(A(x),{\mathcal E}_j)$, $x\in {\mathbb R}^d$. Let $\{ E_j\} $ be the basis 
in the lattice $\Lambda \subset {\mathbb R}^d$,
$$
K=\{ x=\sum\limits_{j=1}^d\xi _jE_j : 0\leqslant \xi _j<1\, ,\ j=1,\dots ,d\} \, .
$$ 
Denote by $v(.)$ the Lebesgue measure on ${\mathbb R}^d$; $v(K)$ is the volume of the
fundamental domain $K$. In what follows, the functions defined on the fundamental
domain $K$ will be also identified with their $\Lambda $-periodic extensions to 
${\mathbb R}^d$.

The scalar products and the norms on the spaces ${\mathbb C}^M$, $L^2({\mathbb R}^d;
{\mathbb C}^M)$, and $L^2(K;{\mathbb C}^M)$ are introduced in the usual way 
(as a rule, omitting the notation for the corresponding space). We assume that the
scalar products are linear in the second argument. For matrices $\widehat L\in 
{\mathcal M}_M\, $, we write
$$
\| \widehat L\| _{{\mathcal M}_M}=\max\limits_{u \, \in \, {\mathbb C}^M\, :\, 
\| u\| =1}\| \widehat Lu\| \, .
$$
The zero and the identity matrices and operators in various spaces are denoted by 
$\widehat 0$ and $\widehat I$, respectively.

Let $H^1({\mathbb R}^d;{\mathbb C}^M)$ be the Sobolev class (of order 1) 
of vector functions $\varphi :{\mathbb R}^d\to {\mathbb C}^M$. The operator
$$
\widehat {\mathcal D}=-i\, \sum\limits_{j=1}^d\widehat \alpha _j\, \frac 
{\partial}{\partial x_j}
$$
acts on the space $L^2({\mathbb R}^d;{\mathbb C}^M)$ and has the domain 
$D(\widehat {\mathcal D})=H^1({\mathbb R}^d;{\mathbb C}^M)$. For $a\geqslant 0$,
let ${\mathbb L}^{\Lambda}_M(d;a)$ be the set of $\Lambda $-periodic matrix 
functions $\widehat W\in L^2_{{\rm {loc}}}({\mathbb R}^d;{\mathcal M}
_M)$ which have bounds $b(\widehat W)\leqslant a$ relative to the operator
$\widehat {\mathcal D}$. If $\widehat W\in {\mathbb L}^{\Lambda}
_M(d;a)$ and $\varphi \in H^1({\mathbb R}^d;{\mathbb C}^M)$, then $\widehat W
\varphi \in L^2({\mathbb R}^d;{\mathbb C}^M)$ and for any $\varepsilon >0$ 
there is a number $C_{\varepsilon}(a,\widehat W)>0$ such that for all 
vector functions $\varphi \in H^1({\mathbb R}^d;{\mathbb C}^M)$ the estimate
$$
\| \widehat W\varphi \| \leqslant (a+\varepsilon )\, \| \widehat 
{\mathcal D}\varphi \| +C_{\varepsilon}(a,\widehat W)\, \| \varphi \| \, .  
\eqno (0.5)
$$
holds. In particular, the set ${\mathbb L}^{\Lambda}_M(d;0)$ (with $a=0$) contains 
$\Lambda $-periodic matrix functions $\widehat W:{\mathbb R}^d\to {\mathcal M}_M\, $
for which at least one of the following conditions is satisfied:

1) $d=2$, the function $\| \widehat W(.)\| ^2_{{\mathcal M}_M}$ belongs to the 
Kato class $K_2$ (see \cite{CFKS}); this condition is fulfilled for the functions 
$\widehat W$ from the Zygmund class $L^2\ln L(K;{\mathcal M}_M)$);

2) $d\geqslant 2$, $\widehat W\in L^2(K;{\mathcal M}_M)$ and  
$$
||| \, \widehat W \, |||_{\gamma ,\, M}\doteq
{\rm {ess}}\, \sup\limits_{\hskip -0.7cm x\, \in \, {\mathbb R}^d}\
\biggl(\ \int_0^1\| \widehat W(x-\xi \gamma )\| ^2_{{\mathcal M}_M}\, d\xi \,
\biggr) ^{\frac 12}<+\infty
$$
for some vector $\gamma \in \Lambda \backslash \{ 0\} $ (see, e.g., \cite{Diff});

3) $d\geqslant 3$, $\widehat W\in L^d(K;{\mathcal M}_M)$.

Let $L^d_w(K;{\mathcal M}_M)$ be the space of functions $\widehat W:K\to
{\mathcal M}_M$ for which
$$
\| \widehat W\| _{L^d_w(K;{\mathcal M}_M)}\doteq \sup\limits_{t\, >\, 0}\, t\, (v(\{
x\in K:\| \widehat W(x)\| _{{\mathcal M}_M}>t\} ))^{\frac 1d}<+\infty \, .
$$
For functions $\widehat W\in L^d_w(K;{\mathcal M}_M)$, we write
$$
\| \widehat W\| ^{(\infty )}_{L^d_w(K;{\mathcal M}_M)}\doteq 
\limsup\limits_{t\, \to \, +\infty }\, t\, (v(\{ x\in K:\| \widehat W(x)\| 
_{{\mathcal M}_M}>t\} ))^{\frac 1d}.
$$
For $d\geqslant 3$, the $\Lambda $-periodic function
$\widehat W\in L^d_w(K;{\mathcal M}_M)$ has the bound 
$$
b(\widehat W)\leqslant C\, \| \widehat W\| ^{(\infty )}
_{L^d_w(K;{\mathcal M}_M)}
$$
relative to the operator $\widehat {\mathcal D}$, where $C=C(d)>0$ (see, e.g., 
\cite{RSII}). From this one also derives the estimate
$$
b(\widehat W)\leqslant C\, \| \widehat W\| ^{(\infty ,\, {\rm 
{loc}})}_{L^d_w(K;{\mathcal M}_M)}\, ,  \eqno (0.6)
$$
where
$$
\| \widehat W\| ^{(\infty ,\, {\rm {loc}})}_{L^d_w(K;{\mathcal M}_M)}\doteq 
\lim\limits_{r\, \to \, +0}\, \sup\limits_{x\, \in \, {\mathbb R}^d}\, 
\limsup\limits_{t\, \to \, +\infty }\, t\, (v(\{ y\in B_r(x):\| \widehat W(x)\| 
_{{\mathcal M}_M}>t\} ))^{\frac 1d},
$$
$B_r(x)=\{ y\in {\mathbb R}^d: |x-y|\leqslant r\} $. If $\widehat W_1\, ,\, 
\widehat W_2\in L^d_w(K;{\mathcal M}_M)$, then
$$
\| \widehat W_1+\widehat W_2\| ^{(\infty ,\, {\rm {loc}})}_{L^d_w(K;{\mathcal M}_M)}
\leqslant 2\, \| \widehat W_1\| ^{(\infty ,\, {\rm {loc}})}_{L^d_w(K;{\mathcal M}
_M)}+2\, \| \widehat W_2\| ^{(\infty ,\, {\rm {loc}})}_{L^d_w(K;{\mathcal M}_M)}\, .
$$

Let $\{ E^*_j\} $ be the basis in the reciprocal lattice $\Lambda ^*\subset 
{\mathbb R}^d$, $(E_j,E^*_l)=\delta _{jl}\, $. We let
$$
\psi _N=v^{-1}(K)\, \int_K\psi (x)\, e^{-2\pi i\, (N,x)}dx\, ,\ N\in \Lambda ^*\, ,
$$
denote the Fourier coefficients of the functions $\psi \in L^1(K;{\mathcal U})$, 
where ${\mathcal U}$ is the space ${\mathbb C}^M$ or ${\mathbb R}^d$ or 
${\mathcal M}_M\, $.

If $\widehat W:{\mathbb R}^d\to {\mathcal S}_M$ is a Hermitian matrix function 
and $\widehat W\in {\mathbb L}^{\Lambda}_M(d;a)$ for some $a\in [0,1)$, 
then $\widehat {\mathcal D}+\widehat W$ is a self-adjoint operator on 
$L^2({\mathbb R}^d;{\mathbb C}^M)$ with the domain $D(\widehat {\mathcal D}+
\widehat W)=D(\widehat {\mathcal D})=H^1({\mathbb R}^d;{\mathbb C}^M)$ (see
\cite{RSII,Kato}). The singular spectrum of the operator $\widehat {\mathcal D}+
\widehat W$ is empty and the eigenvalues (if they exist) have an infinite
multiplicity and form a discrete set (see \cite{Dep96} and also \cite{FS}). 
Therefore, if there are no eigenvalues in the spectrum of the operator $\widehat 
{\mathcal D}+\widehat W$, then the spectrum is absolutely continuous (this
assertion is also a consequence of the results of \cite{K}).

The question on the absolute continuity of the spectrum of periodic operators
of mathematical physics (in particular, of the periodic Dirac operator) attracted
a lot of attention in the past decade. Two papers \cite{BSu99} and \cite{KL}
contain a survey of some early results. The assertions on the absolute
continuity of the spectrum of periodic Schr\" odinger operators (including
ones with variable metrics) can be found in \cite{Th} -- \cite{TF} (also see
references therein). The periodic Maxwell operator was considered in \cite{M2,Su}.

The first results on the absolute continuity of the spectrum of the periodic
Dirac operator were obtained in \cite{P87,P88,TMF90}. In \cite{TMF90,Dep91},
the absolute continuity of the spectrum of the operator (0.2), (0.4) was proved
for all $d\geqslant 2$ under the conditions
$V\in C({\mathbb R}^d)$, $A\in L^{\infty}({\mathbb R}^d;{\mathbb R}^d)$, and
$$
\| \, |A|\, \| _{L^{\infty}({\mathbb R}^d)}<\max\limits_{\gamma \, \in \, \Lambda 
\backslash \{ 0\} }\frac {\pi }{|\gamma |}\ .  \eqno (0.7)
$$

In subsequent papers, the restriction on the periodic electric potential
$V$ has been relaxed. The spectrum of the operator (0.2), (0.4) is absolutely
continuous if at least one of the following conditions is satisfied:

1) $d=2$, $V\in L^q(K)$, $q>2$, and the magnetic potential $A\in L^{\infty}
({\mathbb R}^2;{\mathbb R}^2)$ obeys condition (0.7) (see \cite{Dep92});

2) $d\geqslant 3$, $A\equiv 0$, and
$
\sum_{N\, \in \, {\Lambda }^*}|V_N|^p<+\infty ,
$
where $p\in [1,q_d(q_d-1)^{-1})$ and the numbers $q_d>d$ are found as the largest
roots of the algebraic equations
$$
q^4-(3d^2-4d-1)q^3+2(4d^2-6d-3)q^2-(9d^2-16d-4)q-4d(d-2)=0\, ,
$$
$q_3\simeq 11.645$, $d^{-2}q_d\to 3$ as $d\to +\infty $ (see \cite{Dep96} and also
\cite{P88,Dep91,Dep92});

3) $d=3$, $V\in L^q(K)$, $q>3$, and the magnetic potential $A\in L^{\infty}
({\mathbb R}^3;{\mathbb R}^3)$ satisfies (0.7) (see \cite{TMF95});

4) $d\geqslant 2$, $V\in L^2(K)$, $A\in L^{\infty}({\mathbb R}^d;{\mathbb R}^d)$, 
and there exists a vector $\gamma \in \Lambda \backslash \{ 0\} $ such that 
$\| \, |A|\, \| _{L^{\infty}({\mathbb R}^d)}<\pi |\gamma |^{-1}$ and the map
$$
{\mathbb R}^d\ni x\to \{ [0,1]\ni \xi \to V(x-\xi \gamma )\} \in L^2([0,1])
$$
is continuous (see \cite{Diff}).

In \cite{Izv06}, the absolute continuity of the spectrum of the operator (0.2), 
(0.4) was proved for $d=3$ under conditions: the matrix functions $\widehat V^{\, 
(s)}$, $s=0,1$, belong to the Zygmund class $L^3\ln ^{2+\delta }L(K;{\mathcal M}_M)$ 
for some $\delta >0$, and the magnetic potential $A\in L^{\infty}({\mathbb R}^3;
{\mathbb R}^3)$ satisfies (0.7).

In recent paper \cite{ShZh}, it was proved that the spectrum of the Dirac operator
(0.1) is absolutely continuous if
$$
\widehat \alpha _1\widehat W\widehat \alpha _1=\widehat \alpha _2\widehat W\widehat 
\alpha _2= \dots =\widehat \alpha _d\widehat W\widehat \alpha _d
$$
and for some $r\geqslant d$, $\alpha > (d-1)/(2r)$, we have $\widehat W\in L^r(K;
{\mathcal S}_M)$ and
$$
\| \, \widehat W(.+y)-\widehat W(.)\, \| _{\, L^r(K;{\mathcal S}_M)}\leqslant
\widetilde C\, \{ {\mathrm {dist}}\, (y,\Lambda )\} ^{\alpha }
$$
for any $y\in {\mathbb R}^d$, where $\widetilde C\geqslant 0$ and
$$
{\mathrm {dist}}\, (y,\Lambda )=\min\limits_{\gamma \, \in \, \Lambda }\, |y-
\gamma |\, .
$$

For $d=3$, one can set
$$
\widehat \beta=\left( \begin{matrix} \widehat I & \widehat 0 \\ \widehat 0 &
-\widehat I \end{matrix} \right) \, ,\ \ 
\widehat \alpha _j=\left( \begin{matrix} \widehat 0 & \widehat \sigma _j \\ 
\widehat \sigma _j & \widehat 0 \end{matrix} \right) \, ,\ j=1,2,3\, ,  \eqno (0.8)
$$
where $\widehat 0$ and $\widehat I$ are the zero and the identity $2\times 2$
matrices, and $\widehat \sigma _j$ are the Pauli matrices:
$$
\widehat \sigma _1=\left( \begin{matrix} 0 &1\\ 1 &0 \end{matrix} \right) \, ,\
\widehat \sigma _2=\left( \begin{matrix} 0 &-i\cr i &0 \end{matrix} \right) \, ,\
\widehat \sigma _3=\left( \begin{matrix} 1 &0\cr 0 &-1 \end{matrix} \right) \, .
$$
In this case, the matrix functions $\widehat V^{\, (s)}:{\mathbb R}^3\to {\mathcal 
S}^{(s)}_4\, $, $s=0,1$, can be chosen in the form
$$
\widehat V^{\, (0)}=V^{\, (0)}_1\, \widehat I-i\, V^{\, (0)}_2\, \widehat \alpha 
_1\widehat \alpha _2\widehat \alpha _3 \, ,\ \ \widehat V^{\, (1)}=V^{\, (1)}_1\,
\widehat \beta +V^{\, (1)}_2\, \widehat \alpha _1\widehat \alpha _2\widehat \alpha 
_3\widehat \beta \, ,
$$
where
$$
-i\, \widehat \alpha _1\widehat \alpha _2\widehat \alpha _3=\left( \begin{matrix} 
\widehat 0 & \widehat I \\ \widehat I & \widehat 0 \end{matrix} \right) \, ,\ \ 
\widehat \alpha _1\widehat \alpha _2\widehat \alpha _3\widehat \beta =\left(
\begin{matrix} \widehat 0 & -i\widehat I \\ i\widehat I & \widehat 0 \end{matrix}
\right) \, ,
$$
and $V^{(s)}_l$, $l=1,2$, are $\Lambda $-periodic real-valued functions.

For $d=2$, one can identify the matrices $\widehat \alpha _1\, $, $\widehat \alpha 
_2$, and $\widehat \beta $ with the Pauli matrices $\widehat \sigma _1\, $, 
$\widehat \sigma _2$, and $\widehat \sigma _3\, $, respectively.

The two-dimensional periodic Dirac operator (0.2), (0.3) with an unbounded
magnetic potential $A$ was studied in \cite{TMF99,BSuD}. In \cite{BSuD}, the
absolute continuity of the spectrum of the operator (0.2), (0.3) (with $d=2$)
was proved under the conditions $V,V_1\in L^q(K)$ and $A\in L^q(K;{\mathbb R}^2)$, 
$q>2$. A similar result was obtained in \cite{TMF99} (it was assumed, however, that
$V_1\equiv m={\mathrm {const}}$, but the proof carries over to functions
$V_1\in L^q(K)$, $q>2$, without essential modifications). The methods used in 
\cite{TMF99} were the same as in \cite{Dep92}. More general conditions on
$V$, $V_1\, $, and $A$ were obtained in \cite{L}: it suffices to require that
the functions $V^{\, 2}\ln (1+|V|)$, $V^{\, 2}_1\ln (1+|V_1|)$, and $|A|^2\ln 
^{1+\delta }(1+|A|)$ belong to the space $L^1(K)$ for some $\delta >0$.

In \cite{Izv04,AA}, it was proved that there are no eigenvalues in the spectrum
of a generalized two-dimensional periodic Dirac operator
$$
-i\sum\limits_{j=1}^2(h_{j1}\widehat \sigma _1+h_{j2}\widehat \sigma _2)\, \frac
{\partial}{\partial x_j}+\widehat W\, ,  \eqno (0.9)
$$
where $h_{jl}\in L^{\infty}({\mathbb R}^2;{\mathbb R})$, $j,l=1,2$, are 
$\Lambda $-periodic functions, for which
$$
0<\varepsilon \leqslant h_{11}(x)h_{22}(x)-h_{12}(x)h_{21}(x)
$$ 
for a.e. $x\in {\mathbb R}^2$, and $\widehat W\in {\mathbb L}^{\Lambda }
_2(2;0)$. If the operator (0.9) is self-adjoint, then its spectrum is absolutely
continuous. Some particular cases of the operator (0.9) were also considered in 
\cite{M1,Dep01,Izv02} (in \cite{Dep01}, the functions $h_{jl}$ were supposed to
obey the same conditions as in \cite{AA}, but it was assumed that $\widehat W\in 
L^q(K;{\mathcal M}_M)$, $q>2$).

In \cite{BSuD}, the absolute continuity of the spectrum of the d-dimensional
operator (0.2), (0.3) was proved for $d\geqslant 3$ under the conditions 
$V,V_1\in C({\mathbb R}^d;{\mathbb R})$ and $A\in C^{2d+3}({\mathbb 
R}^d;{\mathbb R}^d)$. The proof was based on Sobolev's paper \cite{S}, where
the absolute continuity of the spectrum was proved for the Schr\" odinger
operator with a periodic magnetic potential $A\in C^{2d+3}({\mathbb R}^d;
{\mathbb R}^d)$, $d\geqslant 3$. The last condition was relaxed by Kuchment and
Levendorski$\breve {\rm i}$ in \cite{KL}: it suffices to require that
$A\in H^q_{\mathrm {loc}}({\mathbb R}^d;{\mathbb R}^d)$, $2q>3d-2$, which makes
it possible to relax accordingly the constraint on the magnetic potential $A$ 
also for the periodic Dirac operator (see \cite{BSu99,BSuD}).

Let $S_{d-1}=\{ x\in {\mathbb R}^d:|x|=1\} $. For vectors $x\in {\mathbb 
R}^d\backslash \{ 0\} $, we shall use the notation
$$
S_{d-2}(x)=\{ \, \widetilde e\in S_{d-1}: (\widetilde e,x)=0\, \} \, .
$$
Let ${\mathfrak M}_{\mathfrak h}$, ${\mathfrak h}>0$, be the set of all even
Borel signed measures $\mu$ on ${\mathbb R}$ (with finite total variation) 
for which $\int_{{\mathbb R}}e^{\, ipt}d\mu (t)=1$ for every $p\in (-{\mathfrak h},
{\mathfrak h})$;
$$
\| \mu \| =\sup\limits_{{\mathcal O}\, \in \, {\mathcal B}({\mathbb R})}(|\mu 
({\mathcal O})|+|\mu ({\mathbb R}\backslash {\mathcal O})|)<+\infty \, ,\ \mu \in 
{\mathfrak M}_{\mathfrak h}\, ,
$$
where ${\mathcal B}({\mathbb R})$ is the collection of Borel subsets ${\mathcal O}
\subseteq {\mathbb R}$. In \cite{TMF00,Dep00}, it was proved that the spectrum of
the $\Lambda $-periodic Dirac operator (0.2) is absolutely continuous for 
$d\geqslant 3$ if the following conditions are fulfilled:

1) $\widehat V^{\, (s)}\in C({\mathbb R}^d;{\mathcal S}^{(s)}_M)$, $s=0,1$;

2) $A\in C({\mathbb R}^d;{\mathbb R}^d)$ and there exist a vector $\gamma \in 
\Lambda \backslash \{ 0\} $ and a measure $\mu \in {\mathfrak M}_{\mathfrak h}\, $, 
${\mathfrak h}>0$, such that for every $x\in {\mathbb R}^d$ and every unit
vector $\widetilde e\in S_{d-2}(\gamma )$ we have
$$
\biggl| A_0- \int_{\mathbb R}d\mu (t)\int_0^1A(x-\xi \gamma -t
\widetilde e)\, d\xi \, \biggr| <\frac {\pi}{|\gamma |}\, ,  \eqno (0.10)
$$
where $A_0=v^{-1}(K)\int_KA(x)\, dx\, $.

For the periodic magnetic potential $A\in C({\mathbb R}^d;{\mathbb R}^d)$, 
$d\geqslant 3$, condition (0.10) is fulfilled (under an appropriate choice
of $\gamma \in \Lambda \backslash \{ 0\} $ and $\mu \in {\mathfrak M}
_{\mathfrak h}\, $, ${\mathfrak h}>0$) whenever $A\in H^q_{\mathrm {loc}}
({\mathbb R}^d;{\mathbb R}^d)$, $2q>d-2$, and also in the case where 
$
\sum_{N\, \in \, \Lambda ^*}\| A_N\| _{{\mathbb C}^d}<+\infty 
$
(see \cite{TMF00,Dep00}).

Let a vector $\gamma \in \Lambda \backslash \{ 0\} $ and a measure $\mu \in 
{\mathfrak M}_{\mathfrak h}\, $, ${\mathfrak h}>0$, be fixed. Denote $e\doteq 
|\gamma |^{-1}\gamma \in S_{d-1}\, $. In this paper, we consider the Dirac 
operator (0.2) for $d\geqslant 3$ supposing that the $\Lambda $-periodic 
magnetic potential $A:{\mathbb R}^d\to {\mathbb R}^d$ satisfies the following
conditions $\bf (A_0)$, $\bf (A_1)$, $\bf (\widetilde A_1)$, and $\bf (A_2)$. 
\vskip 0.2cm

$\bf (A_0)$: $A_0=0$.
\vskip 0.2cm

Since one always can make the transformation
$$
\widehat {\mathcal D}+\widehat W\to e^{\, i\, (A_0,x)}(\widehat {\mathcal D}+
\widehat W)e^{-i\, (A_0,x)}\, ,
$$
without loss of generality we can assume that condition $\bf (A_0)$ holds.
\vskip 0.2cm

$\bf (A_1)$: {\it $A\in L^2_{\mathrm {loc}}({\mathbb R}^d;
{\mathbb R}^d)$ and the map
$$
{\mathbb R}^d\ni x\to \{ [0,1]\ni \xi \to A(x-\xi \gamma )\} \in L^2([0,1];{\mathbb
R}^d)
$$
is continuous $\mathrm ( $in particular, this means that for all $x\in {\mathbb 
R}^d$, the function $\xi \to A(x-\xi \gamma )$ is defined for a.e. $\xi \in {\mathbb 
R}$$\mathrm )$}.
\vskip 0.2cm

Since for any $\varepsilon >0$ there exists a number $C^{\, \prime}(\gamma ,
\varepsilon )>0$ such that for all $\Lambda $-periodic matrix functions 
$\widehat W\in L^2_{\mathrm {loc}}({\mathbb R}^d;{\mathcal M}_M)$ (for which 
$|||\, \widehat W\, ||| _{\gamma ,\, M}<+\infty $) and for all vector functions 
$\varphi \in H^1({\mathbb R}^d;{\mathbb C}^M)$ the inequality
$$
\| \widehat W\varphi \| \leqslant \, |||\, \widehat W\, ||| _{\gamma ,\, M}\, \bigl(
\varepsilon \, \bigl\| \, -i\, \frac {\partial \varphi }{\partial x^{\, \prime}_2}\,
\bigr\| +C^{\, \prime}(\gamma ,\varepsilon )\, \| \varphi \| \, \bigr) 
\eqno (0.11)
$$
holds, where $x^{\, \prime}_2\doteq (x,e)$, $x\in {\mathbb R}^d$ (we can put 
$C^{\, \prime}(\gamma ,\varepsilon )=2(1+2\varepsilon ^{-1}|\gamma |^2)$; see, e.g., 
\cite{Diff}), condition $\bf (A_1)$ implies that 
$$
\sum\limits_{j=1}^dA_j\widehat \alpha _j\in {\mathbb L}^{\Lambda }_M(d;0)\, .
$$

The following condition is a consequence of condition $\bf (A_1)$:
\vskip 0.2cm

$\bf (\widetilde A_1)$: {\it there is a constant $C^*>0$ such that 
$$
\sup\limits_{x\, \in \, {\mathbb R}^d}\ \sup\limits_{\widetilde e\, \in \,
S_{d-2}(\gamma )}\ \iint_{\, \xi ^2_1+\xi ^2_2\, \leqslant \, 1}|A(x-\xi _1\widetilde 
e-\xi _2e)|\, \frac {d\xi _1\, d\xi _2}{\sqrt {\xi ^2_1+\xi ^2_2}}\, \leqslant \, 
C^*\, .  \eqno (0.12)
$$
}
\vskip 0.2cm

Indeed, for all $x\in {\mathbb R}^d$ and all $\widetilde e \in S_{d-2}(\gamma )$,
$$
\iint_{\, \xi ^2_1+\xi ^2_2\, \leqslant \, 1\, }|A(x-\xi _1\widetilde e-\xi _2e)|\, 
\frac {d\xi _1\, d\xi _2}{\sqrt {\xi ^2_1+\xi ^2_2}}\, \leqslant
$$ $$
\leqslant \, \int_{\, |\xi _1|\, \leqslant \, 1}\biggl( \, \int_{-1}^1
|A(x-\xi _1\widetilde e-\xi _2e)|^2d\xi _2\biggr) ^{\frac 12}\, \sqrt {\frac {\pi}
{|\xi _1|}}\ d\xi _1\, \leqslant
$$ $$
\leqslant \, 4\sqrt {\pi}\, \biggl( -\biggl[ -\frac 2{|\gamma |} \biggr] \cdot
|\gamma |\biggr) ^{\frac 12} \biggl( \, \max\limits_{y\, \in \, {\mathbb R}^d}
\int_0^1|A(y-\xi \gamma )|^2\, d\xi \biggr) ^{\frac 12},
$$
where $[t]$ is the integral part of a number $t\in {\mathbb R}$. Therefore, we
can put
$$
C^*=4\sqrt {\pi}\, \biggl( -\biggl[ -\frac 2{|\gamma |} \biggr] \cdot
|\gamma |\biggr) ^{\frac 12} \biggl( \, \max\limits_{y\, \in \, {\mathbb R}^d}
\int_0^1|A(y-\xi \gamma )|^2\, d\xi \biggr) ^{\frac 12}.  
$$

Let us denote
$$
\widetilde A(\widetilde e;x)=\widetilde A(\gamma ,\mu ,\widetilde e;x)=
\int_{\mathbb R}d\mu(t)\, \int_0^1A(x-\xi \gamma -t\widetilde e)\, d\xi \, ,\ x\in 
{\mathbb R}^d\, ,\ \widetilde e\in S_{d-2}(\gamma )\, .
$$
From condition $\bf (A_1)$ it follows that the periodic function
$$
{\mathbb R}^d\ni x\to \int_0^1A(x-\xi \gamma )\, d\xi
$$
is continuous. Therefore the function $\widetilde A(.;.):S_{d-2}(\gamma )\times 
{\mathbb R}^d\to {\mathbb R}^d$ is also continuous, and $\Lambda $-periodic in the
second argument. Moreover, $(\widetilde A(\widetilde e;.))_0=A_0=0$ (for all 
$\widetilde e\in S_{d-2}(\gamma )$).
\vskip 0.2cm

$\bf (A_2)$: {\it there is a constant $\widetilde \theta \in [0,1)$ such that
$$
\max\limits_{x\, \in \, {\mathbb R}^d}\ \max\limits_{\widetilde e\, \in \, S_{d-2}
(\gamma )}\ |\widetilde A(\widetilde e;x)| \, \leqslant \, \frac {\widetilde \theta 
\pi }{|\gamma |}\, .  \eqno (0.13)
$$
}
\vskip 0.2cm

If we pick the Dirac measure $\mu =\delta $, then the function $\widetilde A
(\widetilde e;.)$ does not depend on the vector $\widetilde e$ and inequality 
(0.13) means that
$$
\max\limits_{x\, \in \, {\mathbb R}^d}\ \, \bigl| \, \int_0^1A(x-\xi \gamma )
\, d\xi \, \bigr| \, \leqslant \, \frac {\widetilde \theta \pi }{|\gamma |}\, .
$$
The last inequality is valid if
$$
\sum_{N\, \in \, \Lambda ^*\, :\, (N,\gamma )\, =\, 0}\| A_N\| _{{\mathbb C}^d}\,
\leqslant \, \frac {\widetilde \theta \pi }{|\gamma |}\, . 
$$

The following two theorems are the main results of this paper.

\begin{theorem} \label{th0.1}
Suppose $d\geqslant 3$, $\widehat V^{\, (s)}\in L^2(K;{\mathcal S}^{(s)}_M)$, 
$s=0,1$, $A\in L^2(K;{\mathbb R}^d)$, and there exist a vector $\gamma \in \Lambda
\backslash \{ 0\} $ and a measure $\mu \in {\mathfrak M}_{\mathfrak h}\, $, 
${\mathfrak h}>0$, such that conditions $\bf (A_0)$, $\bf (A_1)$, $\bf (A_2)$ 
are fulfilled for the magnetic potential $A$, and the maps
$$
{\mathbb R}^d\ni x\to \{ [0,1]\ni \xi \to \widehat V^{\, (s)}(x-\xi \gamma )\}
\in L^2([0,1];{\mathcal S}^{(s)}_M)\, ,\ s=0,1\, ,
$$
are continuous. Then the spectrum of the periodic Dirac operator (0.2) is absolutely
continuous.
\end{theorem}

Denote
$$
x_{\|}=(x,e)e\, ,\ x_{\perp}=x-(x,e)e\, ,\ \,  x\in {\mathbb R}^d.
$$
For a matrix function $\widehat W\in L^2(K;{\mathcal M}_M)$ and a number $\sigma
\in [0,2]$, we set
$$
\beta _{\gamma ,\, \sigma }(R;\widehat W)\doteq v(K)\, \sup\limits_{N\, \in \, 
\Lambda ^*\, :\, 2\pi |N_{\perp}|\, \geqslant \, R}(2\pi |N_{\perp}|)^{2-\sigma }\, 
(2\pi |N|)^{\sigma }\, \| \widehat W_N\| \, ,\ R\geqslant 0\, ,
$$ $$
\beta _{\gamma ,\, \sigma }(\widehat W)\doteq \lim\limits_{R\, \to \, +\infty }\beta 
_{\gamma ,\, \sigma }(R;\widehat W)\, .
$$

Let ${\mathrm {supp}}\, \widehat W$ be the essential support of a measurable
function $\widehat W: {\mathbb R}^d\to {\mathcal M}_M\, $, ${\mathrm {supp}}\, 
\widehat W={\mathbb R}^d\, \backslash \, \{ x\in {\mathbb R}^d:\widehat W(y)=\widehat 0$ 
for a.e. $y\in B_r(x)$ for some $r=r(x)>0\} $.

\begin{theorem} \label{th0.2}
Suppose $d=3$, $\widehat V^{\, (s)}\in L^2(K;{\mathcal S}^{(s)}_M)$, $s=0,1$, 
$A\in L^2(K;{\mathbb R}^3)$, and there exist a vector $\gamma \in \Lambda
\backslash \{ 0\} $ and a measure $\mu \in {\mathfrak M}_{\mathfrak h}\, $, 
${\mathfrak h}>0$, such that conditions $\bf (A_0)$, $\bf (A_1)$, $\bf (A_2)$
are satisfied for the magnetic potential $A$, and the functions 
$\widehat V^{\, (s)}$, $s=0,1$, can be represented in the form
$$
\widehat V^{\, (s)}=\widehat V^{\, (s)}_1+\widehat V^{\, (s)}_2+\widehat V^{\, 
(s)}_3\, ,
$$
where $\mathrm ($for $s=0,1$$\mathrm )$ the $\Lambda $-periodic matrix functions
$\widehat V^{\, (s)}_{\nu }$, $\nu =1,2,3$, obey the conditions$\mathrm :$

1$\mathrm )$ for the functions $\widehat V^{\, (s)}_1\in L^2(K;{\mathcal S}^{(s)}_M)$, 
the maps
$$
{\mathbb R}^d\ni x\to \{ [0,1]\ni \xi \to \widehat V^{\, (s)}_1(x-\xi \gamma )\}
\in L^2([0,1];{\mathcal S}^{(s)}_M)
$$
are continuous;

2$\mathrm )$ $\widehat V^{\, (s)}_2\in L^3\ln ^{1+\delta }L(K;{\mathcal S}^{(s)}_M)$
for some $\delta >0$;

3$\mathrm )$ $\widehat V^{\, (s)}_3=\sum_{q=1}^{Q_s}\widehat V^{\, (s)}_{3,\, q}\, $,
$\widehat V^{\, (s)}_{3,\, q}\in L^3_w(K;{\mathcal S}^{(s)}_M)$, and $\beta _{\gamma 
,\, \sigma }(0;\widehat V^{\, (s)}_{3,\, q})<+\infty $ for some $\sigma \in (0,2]$, 
$q=1,\dots ,Q_s\, $, moreover, the essential supports ${\mathrm 
{supp}}\, \widehat V^{\, (s)}_{3,\, q}$ of the functions $\widehat V^{\, (s)}_{3,\, 
q}$ $\mathrm ($considered as $\Lambda $-periodic functions defined on ${\mathbb 
R}^3$$\mathrm )$ do not intersect for different $q$, 
$$
\| \widehat V^{\, (s)}_3\| ^{(\infty , \, {\mathrm {loc}})}_{L^3_w(K;{\mathcal M}_M)}
=\max\limits_{q=1,\dots ,Q_s}\, \| 
\widehat V^{\, (s)}_{3,\, q}\| ^{(\infty , \, {\mathrm {loc}})}_{L^3_w(K;{\mathcal 
M}_M)}\leqslant c_1\, \ \text{and} \ \, 
\max\limits_{q=1,\dots ,Q_s}\, \beta _{\gamma ,\, \sigma }(\widehat V^{\, (s)}_{3,\, 
q})\leqslant c_1^{\, \prime}\, ,
$$
where $c_1=c_1({\mathfrak h},\mu ,\gamma ;A)\in (0,C^{-1})$ and $c_1^{\, 
\prime}=c_1^{\, \prime}({\mathfrak h},\mu ,\gamma ,\sigma ;A)>0$ $\mathrm 
($$C$ is the constant from ${\mathrm {(0.6))}}$.

Then the spectrum of the $\Lambda $-periodic Dirac operator (0.2) is absolutely
continuous.

\end{theorem}

{\bf {Remark}}. 
For $d=3$, we can set $M=4$ (and take the matrices $\widehat \alpha _j$ and 
$\widehat \beta $ in the form (0.8)). The condition 3 in Theorem \ref{th0.2} admits
Coulomb singularities for the functions $\widehat V^{\, (s)}_3$, and therefore for
the functions $V,V_1$ in (0.3), and the function $V$ in (0.4). The condition 3 is 
fulfilled for the $\Lambda $-periodic functions $\widehat V^{\, (s)}_3$ if
$\widehat V^{\, (s)}_3\in C^{\infty }({\mathbb R}^3\backslash \bigcup_{j=1}^J(x_j+
\Lambda );{\mathcal S}^{(s)}_M)$, where $x_j\, $, $j=1,\dots ,J$, are different
points in $K$ such that $\widehat V^{\, (s)}_3(x)=\widehat 0$ for $x$ which do not 
belong to some disjoint closed $\Lambda $-periodic neighbourhoods of the sets $x_j+
\Lambda $, $j=1,\dots ,J$, and the functions $\widehat V^{\, (s)}_3(.)$ coincide with 
the functions $|\, .-x_j\, |^{-1}\, \widehat Q^{(s)}_j$ in certain neighbourhoods (in 
${\mathbb R}^3$) of points $x_j\, $, where $\widehat Q^{(s)}_j\in {\mathcal 
S}^{(s)}_4$ and 
$$
\| \widehat Q^{(s)}_j\| _{{\mathcal M}_4}\leqslant \min \, \bigl\{ \bigl( \, \frac 
43\, \pi \bigr) ^{-1/3}\, c_1\, ,\, (2\pi )^{-1}c_1^{\, \prime }\, \bigr\}            
$$ 
for all $j=1,\dots ,J$ and $s=0,1$ (here, the constant $c_1^{\, \prime }$ corresponds 
to the number $\sigma =2$). 

\section{Proof of Theorems \ref{th0.1} and \ref{th0.2}}

Let $\widetilde H^a(K;{\mathbb C}^M)$, $a>0$, be the set of vector functions
$\varphi :K\to {\mathbb C}^M$ whose $\Lambda $-periodic extensions belong to the
Sobolev class $H^a_{\mathrm {loc}}({\mathbb R}^d;{\mathbb C}^M)$ (of order $a$); 
$\widetilde H^0(K;{\mathbb C}^M)\doteq L^2(K;{\mathbb C}^M)$. For all $e\in S_{d-1}\, 
$, all $k\in {\mathbb R}^d$ ($e_j=(e,{\mathcal E}_j)$, $k_j=(k,{\mathcal E}_j)$, 
$j=1,\dots ,d$), and all $\varkappa \geqslant 0$ we introduce the operators
$$
\widehat {\mathcal D}(k+i\varkappa e)=-i\sum\limits_{j=1}^d\widehat \alpha _j 
\frac {\partial}{\partial x_j}+\sum\limits_{j=1}^d\, (k_j+i\varkappa e_j)\widehat 
\alpha _j
$$
acting on $L^2(K;{\mathbb C}^M)$, with the domain $D(\widehat {\mathcal D}
(k+i\varkappa e))=\widetilde H^1(K;{\mathbb C}^M)\subset L^2(K;{\mathbb C}^M)$. 

If $\widehat W\in {\mathbb L}^{\Lambda }_M(d;a)$, $a\geqslant 0$, then for all 
$\varepsilon >0$, all vectors $k\in {\mathbb R}^d$, and all vector functions 
$\varphi \in \widetilde H^1(K;{\mathbb C}^M)$ we have
$$
\| \widehat W\varphi \| \leqslant (a+\varepsilon )\, \| \widehat {\mathcal D}(k)
\varphi \| +C_{\varepsilon }(a,\widehat W)\, \| \varphi \| \, =  \eqno (1.1)
$$ $$
=\, (a+\varepsilon )\, \bigl( \, \sum\limits_{j=1}^d \, \bigl\| \bigl( k_j-i\, 
\frac {\partial}{\partial x_j}\, \bigr)\, \varphi \bigr\| ^2\, \bigr) ^{1/2}
+C_{\varepsilon }(a,\widehat W)\, \| \varphi \| \, ,
$$
where the operator $\widehat {\mathcal D}(k)$ is the operator $\widehat {\mathcal D}
(k+i\varkappa e)$ for $\varkappa =0$, $C_{\varepsilon }(a,\widehat W)$ is the
constant from (0.5), and the function $\widehat W$ is supposed to act on the space 
$L^2(K;{\mathbb C}^M)$ (here and henceforth, the norms and the scalar products are
related to the space $L^2(K;{\mathbb C}^M)$).

Under the conditions of Theorem \ref{th0.1}, we have
$$
\widehat W=\widehat V^{\, (0)}+\widehat V^{\, (1)}-\sum\limits_{j=1}^dA_j\widehat 
\alpha _j\in {\mathbb L}^{\Lambda }_M(d;0)\, ,
$$
furthermore, for any $\varepsilon >0$ there is a constant $C_{\varepsilon }^{\,
\prime }(0,\widehat W)>0$ such that for all $k\in {\mathbb R}^d$ and all
$\varphi \in \widetilde H^1(K;{\mathbb C}^M)$ the estimate
$$
\| \widehat W\varphi \| \leqslant \| \widehat V^{\, (0)}\varphi \| +
\| \widehat V^{\, (1)}\varphi \| + \| |A| \varphi \| \leqslant  
\varepsilon \, \| \widehat {\mathcal D}(k)\varphi \| +C_{\varepsilon }
^{\, \prime }(0,\widehat W)\, \| \varphi \|  \eqno (1.2)
$$
holds (we can set $C_{\varepsilon }^{\, \prime }(0,\widehat W)=C_{\varepsilon /3}
(0,\widehat V^{\, (0)})+C_{\varepsilon /3}(0,\widehat V^{\, (1)})+C_{\varepsilon /3}
(0,\sum_{j=1}^{\, d}A_j\widehat \alpha _j)$). If the conditions of Theorem \ref{th0.2}
are fulfilled, then $\widehat W\in {\mathbb L}^{\Lambda }_M(3;a)$ for some 
$a\in [0,1)$.

Suppose $\widehat W\in {\mathbb L}^{\Lambda }_M(d;a)$, $a\in [0,1)$. Then the
operators $\widehat {\mathcal D}(k+i\varkappa e)+\widehat W$, with the domain
$D(\widehat {\mathcal D}(k+i\varkappa e)+\widehat W)=\widetilde H^1(K;{\mathbb C}^M)
\subset L^2(K;{\mathbb C}^M)$, are closed and have compact resolvent for all 
$k+i\varkappa e\in {\mathbb C}^d$. Furthermore, the operators $\widehat {\mathcal 
D}(k)+\widehat W$, $k\in {\mathbb R}^d$, are self-adjoint and have a discrete
spectrum. Let $\lambda _{\nu}(k)$, $\nu \in {\mathbb Z}$, be the eigenvalues of
the operators $\widehat {\mathcal D}(k)+\widehat W$. We assume that they are
arranged in an increasing order (counting multiplicities). The eigenvalues can be
indexed for different $k\in {\mathbb R}^d$ such that the functions ${\mathbb R}^d\ni 
k\to \lambda _{\nu}(k)$ are continuous (see \cite{TMF90}). Denote by
$$
K^*=\{ y=\sum\limits_{j=1}^d\eta _jE^*_j:1\leqslant \eta _j<1\, ,\ j=1,\dots
,d\}
$$
the fundamental domain of the lattice  $\Lambda ^*$, $v(K^*)$ is the volume of $K^*$;
$v(K^*)v(K)=1$. The periodic Dirac operator (0.1) is unitarily equivalent to the
direct integral 
$$
\int_{2\pi K^*}^{\, \bigoplus}(\widehat {\mathcal D}(k)+\widehat W)\,
\frac {dk}{(2\pi )^d\, v(K^*)}\ .  \eqno (1.3)
$$
The unitary equivalence is established via the Gel'fand transformation (see
\cite{G,RSIV}). The spectrum of the operator (0.1) coincides with $\bigcup_{\nu \,
\in \, {\mathbb Z}}\{ \lambda _{\nu}(k):k\in 2\pi K^*\} $.

To prove the absolute continuity of the spectrum of the operator (0.1), it suffices 
to prove the absence of eigenvalues (of infinite multiplicities) in its spectrum
(see \cite{K}). But, on the other hand, if $\lambda $ is an eigenvalue of the
operator (0.1), then the decomposition of the operator (0.1) into the direct integral
(1.3) and the analytic Fredholm theorem  imply that $\lambda $ is an eigenvalue of 
$\widehat {\mathcal D}(k+i\varkappa e)+\widehat W$ for all $k+i\varkappa e\in 
{\mathbb C}^d$ (see \cite{K,KL}). Hence, it suffices to prove that every number
$\lambda \in {\mathbb R}$ is not an eigenvalue of the operator $\widehat {\mathcal 
D}(k+i\varkappa e)+\widehat W$ for some complex vector $k+i\varkappa e\in {\mathbb 
C}^d$ (dependent on $\lambda $). This method was used by Thomas in \cite{Th}. If,
under the conditions of Theorems \ref{th0.1} and \ref{th0.2}, we change $\widehat 
W-\lambda \widehat I$ to $\widehat W$ (change $\widehat V^{\, (0)}_1-\lambda \widehat 
I$ to $\widehat V^{(0)}_1$ in the case of the operator (0.2)), then the new matrix
function $\widehat W$ satisfies all conditions which are satisfied by the original
matrix function $\widehat W$. Therefore, to prove the absolute continuity of the 
spectrum of the operator (0.1), it suffices to prove the absence of the eigenvalue 
$\lambda =0$ in the spectrum of the operator $\widehat {\mathcal D}(k+i\varkappa e)
+\widehat W$ for some complex vector $k+i\varkappa e\in {\mathbb C}^d$. Thus, 
Theorems \ref{th0.1} and \ref{th0.2} are implied by Theorems \ref{th1.1} and 
\ref{th1.2}, respectively.

Let $\gamma \in \Lambda \backslash \{ 0\} $ be the vector fixed in
Theorems \ref{th0.1} and \ref{th0.2}. In what follows, we denote $e=|\gamma |^{-1}
\gamma $.

\begin{theorem} \label{th1.1}
Let $d\geqslant 3$. Suppose the matrix functions $\widehat V^{\, (s)}$, $s=0,1$, 
and the magnetic potential $A$ $\mathrm ($for the vector $\gamma \in \Lambda 
\backslash \{ 0\} $ and the measure $\mu \in {\mathfrak M}_{\mathfrak h}\, $, 
${\mathfrak h}>0$$\mathrm )$ satisfy the conditions of Theorem \ref{th0.1}. Then
there exist a number $C_1>0$ $\mathrm ($we can put $C_1=\frac 12\, \pi |\gamma |
^{-1}C_2\, $, where $C_2$ is the constant from $\mathrm ($1.10$\mathrm )$$\mathrm )$ 
and a number $\varkappa _0>0$ such that for all $\varkappa \geqslant \varkappa _0\, 
$, all vectors $k\in {\mathbb R}^d$ with $|(k,\gamma )|=\pi $, and all vector
functions $\varphi \in \widetilde H^1(K;{\mathbb C}^M)$ we have
$$
\| (\widehat {\mathcal D}(k+i\varkappa e)+\widehat W)\varphi \| \geqslant 
C_1\, \| \varphi \| \, .
$$
\end{theorem}

\begin{theorem} \label{th1.2}
Let $d=3$. Suppose the matrix functions $\widehat V^{\, (s)}$, $s=0,1$, and the
magnetic potential $A$ $\mathrm ($for the vector $\gamma \in \Lambda \backslash \{ 
0\} $ and the measure $\mu \in {\mathfrak M}_{\mathfrak h}\, $, ${\mathfrak 
h}>0$$\mathrm )$ satisfy the conditions of Theorem \ref{th0.2} $\mathrm ($the
constants $c_1\in (0,C^{-1})$ and $c_1^{\, \prime }>0$ will be specified later
in the course of the proof$\mathrm )$. Then there exists a number $C_1^{\, \prime 
}>0$ $\mathrm ($as in Theorem \ref{th1.1}, we can put $C_1^{\, \prime }=\frac 12\, 
\pi |\gamma |^{-1}C_2\, $, where $C_2$ is the constant from $\mathrm ($1.10$\mathrm 
)$$\mathrm )$ and, for every $\Xi >1$, there exists a number $\varkappa _0^{\, (\Xi 
)}>0$ $\mathrm ($which also depend on $\mathfrak h$, $\mu $, $\Lambda $, $\gamma $, 
$\sigma $, on the functions $\widehat V^{\, (s)}_{\nu }\, $, $s=0,1$, $\nu =1,2,3$, 
and on the magnetic potential $A$$\mathrm )$ such that for any $\Xi >1$ and any 
$\varkappa _1\geqslant \varkappa _0^{\, (\Xi )}$ there is a number $\varkappa \in 
[\varkappa _1,\Xi \varkappa _1]$ such that for all vectors $k\in {\mathbb R}^3$ with 
$|(k,\gamma )|=\pi $, and all vector functions $\varphi \in \widetilde H^1(K;
{\mathbb C}^M)$ the estimate
$$
\| (\widehat {\mathcal D}(k+i\varkappa e)+\widehat W)\varphi \| \geqslant 
C_1^{\, \prime }\, \| \varphi \| 
$$
holds.
\end{theorem}

Theorems \ref{th1.1} and \ref{th1.2} are proved in the end of this section. Theorem
\ref{th1.1} is based on Theorem \ref{th1.3}. Theorem \ref{th1.2} is deduced from
Theorems \ref{th1.4}, \ref{th1.5}, and \ref{th1.6}.

For all vector functions $\varphi \in \widetilde H^1(K;{\mathbb C}^M)$, 
$$
\widehat {\mathcal D}(k+i\varkappa e)\varphi =\sum\limits_{N\, \in \, \Lambda ^*}
\widehat {\mathcal D}_N(k;\varkappa )\, \varphi _N\, e^{\, 2\pi i\, (N,x)},
$$
where 
$$
\widehat {\mathcal D}_N(k;\varkappa )=\sum\limits_{j=1}^d\, (k_j+2\pi N_j+i
\varkappa e_j)\, \widehat \alpha _j\, ,\ N_j=(N,{\mathcal E}_j)\, ,\ 
j=1,\dots ,d\, . 
$$
For $k\in {\mathbb R}^d$ and $\varkappa \geqslant 0$, we introduce the notation
$$
G^{\pm}_N(k;\varkappa )\doteq \bigl( |k_{\parallel}+2\pi N_{\parallel}|^2+
(\varkappa \pm |k_{\perp}+2\pi N_{\perp}|)^2\bigr) ^{1/2}\, ,\ N\in \Lambda
^*\, ;
$$ $$
G^+_N(k;\varkappa )\geqslant G^-_N(k;\varkappa )\, ,\ G^+_N(k;\varkappa )
\geqslant \varkappa \, .
$$
If $k\in {\mathbb R}^d$ and $|(k,\gamma )|=\pi $, then for all $\varkappa 
\geqslant 0$ and all $N\in \Lambda ^*$ we have $G_N(k;\varkappa )\geqslant \pi 
|\gamma |^{-1}$.

In the case where $|(k,\gamma )|=\pi $, we define the operators $\widehat G_{\pm}
^{\, \zeta }$, with $\zeta \in {\mathbb C}$, acting on the space $L^2(K;{\mathbb 
C}^M)$:
$$
\widehat G_{\pm}^{\, \zeta }\varphi =\sum\limits_{N\, \in \, \Lambda ^*}(G^{\pm}
_N(k;\varkappa ))^{\, \zeta }\, \varphi _N\, e^{\, 2\pi i\, (N,x)}\, ,
$$ $$
\varphi \in D(\widehat G_{\pm}^{\, \zeta })=\left\{
\begin{array}{ll}
\widetilde H^{\, {\mathrm {Re}}\, \zeta }(K;{\mathbb C}^M)& 
{\text {if}} \ \, {\mathrm {Re}}\, \zeta >0\, , \\ [0.2cm]
L^2(K;{\mathbb C}^M)  & {\text{if}} \ \, {\mathrm {Re}}\, \zeta \leqslant 0
\end{array}
\right.
$$
(the operators $\widehat G_{\pm}^{\, \zeta }$ depend on $k$ and $\varkappa $, but, 
in the above notation, this dependence is not indicated explicitely).

For all $k\in {\mathbb R}^d$, all $\varkappa \geqslant 0$, and all $N\in \Lambda ^*$, 
the inequalities
$$
G^-_N(k;\varkappa )\, \| u\| \leqslant \| \widehat {\mathcal D}_N(k;\varkappa ) u\| 
\leqslant G^+_N(k;\varkappa )\, \| u\| \, ,\ u\in {\mathbb C}^M\, ,
$$
hold. Hence, for all vector functions $\varphi \in \widetilde H^1(K;{\mathbb C}^M)$,
$$
\| \widehat G_-^{\, 1}\varphi \| \leqslant \| \widehat {\mathcal D}(k+i\varkappa 
e)\varphi \| \leqslant \| \widehat G_+^{\, 1}\varphi \| \, .
$$

Let $\widehat P^{\, \mathcal C}$, where ${\mathcal C}\subseteq \Lambda ^*$, denote
the orthogonal projection on $L^2(K;{\mathbb C}^M)$ that takes a vector function 
$\varphi \in L^2(K;{\mathbb C}^M)$ to the vector function
$$
\widehat P^{\, \mathcal C}\varphi \doteq \varphi ^{\, \mathcal C}=\sum\limits
_{N\, \in \, {\mathcal C}}\varphi _N\, e^{\, 2\pi i\, (N,x)}
$$
(in particular, $\widehat P^{\, \emptyset}=\widehat 0$). We write ${\mathcal H}
(\mathcal C)=\{ \varphi \in L^2(K;{\mathbb C}^M):\varphi _N=0$ for $N\in \Lambda ^*
\backslash \mathcal C\} $.

For vectors $\widetilde e\in S_{d-2}(e)$, we define the orthogonal projections on
${\mathbb C}^M$:
$$
\widehat P^{\, \pm}_{\widetilde e}=\frac 12\, \bigl( \widehat I\mp i\, \bigl( \,
\sum\limits_{j=1}^de_j\widehat \alpha _j\bigr) \bigl( \, \sum\limits_{j=1}^d
\widetilde e_j\widehat \alpha _j\bigr) \bigr) \, ;
$$ $$
\| \widehat P^{\, \pm}_{{\widetilde e}^{\, \prime}}\widehat P^{\, \mp}_{{\widetilde e}
^{\, \prime \prime}}\| =\| \widehat P^{\, \pm}_{{\widetilde e}^{\, \prime}}-
\widehat P^{\, \pm}
_{{\widetilde e}^{\, \prime \prime}}\| =\frac 12\ |\, {\widetilde e}^{\, \prime 
\prime}-{\widetilde e}^{\, \prime}\, |\, ,\ {\widetilde e}^{\, \prime},{\widetilde e}
^{\, \prime \prime}\in S_{d-2}(e)\, .  \eqno (1.4)
$$

We shall use the notation $\widetilde e(y)\doteq |y_{\perp}|^{-1}y_{\perp}\in 
S_{d-2}(e)$ for vectors $y\in {\mathbb R}^d$ with $y_{\perp}=y-(y,e)e\neq 0$.

If $k\in {\mathbb R}^d$, $N\in \Lambda ^*$, and $k_{\perp}+2\pi N_{\perp}\neq 0$,
then
$$
\widehat P^{\, \pm}_{\widetilde e(k+2\pi N)}\widehat {\mathcal D}_N(k;
\varkappa )\widehat P^{\, \pm}_{\widetilde e(k+2\pi N)}=\widehat 0  \eqno (1.5)
$$
and, for all vectors $u\in {\mathbb C}^M$ (and all $\varkappa \geqslant 0$),
$$
\| \widehat {\mathcal D}_N(k;\varkappa )\widehat P^{\, \pm}_{\widetilde 
e(k+2\pi N)}u\| =G^{\pm}_N(k;\varkappa )\| \widehat P^{\, \pm}_{\widetilde 
e(k+2\pi N)}u\| \, . \eqno (1.6)
$$
If $k_{\perp}+2\pi N_{\perp}=0$, then 
$$
G^+_N(k;\varkappa )=G^-_N(k;\varkappa )\, .  \eqno (1.7)
$$

Let denote by $\widehat P^{\, \pm}=\widehat P^{\, \pm}(k;e)$ and by $\widehat 
P^{\, \pm}_*=\widehat P^{\, \pm}_*(k;e)$ the orthogonal projections on 
$L^2(K;{\mathbb C}^M)$: 
$$
\widehat P^{\, +}\varphi =\sum\limits_{N\, \in \, \Lambda ^*\, :\, k_{\perp}+2\pi N
_{\perp}\, \neq \, 0}\widehat P^{\, +}_{\widetilde e(k+2\pi N)}\, \varphi _N\,
e^{\, 2\pi i\, (N,x)}\, ,
$$ $$
\widehat P^{\, +}_*\varphi =\widehat P^{\, +}\varphi +\sum\limits_{N\, \in \, 
\Lambda ^*\, :\, k_{\perp}+2\pi N_{\perp}\, =\, 0}\varphi _N\, e^{\, 
2\pi i\, (N,x)}\, ,
$$ $$
\widehat P^{\, -}_*\varphi =\sum\limits_{N\, \in \, \Lambda ^*\, :\, k_{\perp}+2\pi N
_{\perp}\, \neq \, 0}\widehat P^{\, -}_{\widetilde e(k+2\pi N)}\, \varphi _N\,
e^{\, 2\pi i\, (N,x)}\, , 
$$ $$
\widehat P^{\, -}\varphi =\widehat P^{\, -}_*\varphi +\sum\limits_{N\, \in \, 
\Lambda ^*\, :\, k_{\perp}+2\pi N_{\perp}\, =\, 0}\varphi _N\, e^{\, 2\pi i\, 
(N,x)}\, ,\ \varphi \in L^2(K;{\mathbb C}^M)
$$
(the operators $\widehat P^{\, \pm}$ and $\widehat P^{\, \pm}_*$ depend on $k\in 
{\mathbb R}^d$, but the dependence will not be indicated in the notation). Since
$\widehat P^{\, +}+\widehat P^{\, -}=\widehat I$, $\widehat P^{\, +}_*+
\widehat P^{\, -}_*=\widehat I$, from equalities (1.5), (1.6), and (1.7) it follows 
that 
$$
\| \widehat P^{\, +}_*\widehat {\mathcal D}(k+i\varkappa e)\varphi \| =
\| \widehat G_-^{\, 1}\widehat P^{\, -}\varphi \| \, ,\ \ \| \widehat P^{\, -}_*
\widehat {\mathcal D}(k+i\varkappa e)\varphi \| =\| \widehat G_+^{\, 1}\widehat 
P^{\, +}\varphi \| \, ,  \eqno (1.8)
$$ $$
\| \widehat {\mathcal D}(k+i\varkappa e)\varphi \| ^2=\| \widehat G_-^{\, 1}\widehat 
P^{\, -}\varphi\| ^2+\| \widehat G_+^{\, 1}\widehat P^{\, +}\varphi \| ^2
$$
for all vector functions $\varphi \in \widetilde H^1(K;{\mathbb C}^M)$.

Condition $\bf (A_1)$ and inequality (0.11) imply that for any $\tau \in
(0,1)$ there is a number $Q=Q(\gamma ,A;\tau )>0$ such that for all $k\in 
{\mathbb R}^d$ and all vector functions $\varphi \in \widetilde H^1(K;{\mathbb 
C}^M)$ the inequality
$$
\| |A| \varphi \| \leqslant \tau \, \bigl\| \, \bigl( \, k_2^{\, \prime }
-i\, \frac {\partial}{\partial x_2^{\, \prime }}\, \bigr) \, \varphi \, \bigr\|
+Q\, \| \varphi \|   \eqno (1.9)
$$
is fulfilled, where $k_2^{\, \prime }=(k,e)$, $x_2^{\, \prime }=(x,e)$, $x\in 
{\mathbb R}^d$ (we can put
$$
Q=\frac {\tau }{4|\gamma |^2}+\frac {8|\gamma |^2}{\tau }\ \max\limits_{x\, \in \, 
{\mathbb R}^d}\ \int_0^1| A(x-\xi \gamma )| ^2\, d\xi 
$$
(see (0.11))).

Let $C^*({\mathfrak h})>0$ be the constant defined in Lemma \ref{l3.1}. The
constant $C^*({\mathfrak h})$ depends on $C^*$ and ${\mathfrak h}$. (In what
follows, we shall assume that ${\mathfrak h}\leqslant |\gamma |^{-1}$. Therefore
the constant $C^*({\mathfrak h})$ will depend on the vector $\gamma $ as well.) Fix
a number $\tau \in (0,1)$ and the corresponding number $Q>0$. Denote
$$
C_2=(1-\tau )\bigl( 1+Q(1-\widetilde \theta )^{-1}\, \frac {|\gamma |}{\pi}\
e^{\, \| \mu \| \, C^*({\mathfrak h})}\bigr) ^{-1}\in (0,1)\, .  \eqno (1.10)
$$

\begin{theorem} \label{th1.3}
Let $d\geqslant 3$. Suppose $\widehat V^{\, (s)}\in L^2(K;{\mathcal S}
^{(s)}_M)$, $s=0,1$, $A\in L^2(K;{\mathbb R}^d)$ with $A_0=0$, $R>0$, and there are
a vector $\gamma \in \Lambda \backslash \{ 0\} $ and a measure $\mu \in {\mathfrak 
M}_{\mathfrak h}\, $, ${\mathfrak h}>0$, such that for the magnetic potential $A$,
conditions $\bf (A_1)$, $\bf (\widetilde A_1)$, $\bf (A_2)$ are satisfied and, 
moreover, $\widehat V^{\, (s)}_N=\widehat 0$, $s=0,1$, and $A_N=0$ for all 
$N \in \Lambda ^*$ with $2\pi |N_{\perp}|>R$. Then for any $\delta >0$ there exist 
numbers $\widetilde a=\widetilde a(C_2;\delta ,R)\in (0,C_2]$ and $\varkappa _0>0$ 
such that for all $\varkappa \geqslant \varkappa _0\, $, all vectors $k\in {\mathbb 
R}^d$ with $|(k,\gamma )|=\pi $, and all vector functions $\varphi \in \widetilde 
H^1(K;{\mathbb C}^M)$ the inequality
$$
\| (\widehat P^{\, +}_*+\widetilde a\widehat P^{\, -}_*)(\widehat {\mathcal D}
(k+i\varkappa e)+\widehat W)\varphi \| ^2=
$$ $$
=\| \widehat P^{\, +}_*(\widehat {\mathcal D}(k+i\varkappa e)
+\widehat W)\varphi \| ^2+{\widetilde a}^{\, 2}\| \widehat P^{\, -}_*(\widehat 
{\mathcal D}(k+i\varkappa e)+\widehat W)\varphi \| ^2\geqslant
$$ $$
\geqslant (1-\delta )\, (C^2_2\, \| \widehat G_-^{\, 1}\widehat P^{\, -}\varphi \| ^2+
{\widetilde a}^{\, 2}\, \| \widehat G_+^{\, 1}\widehat P^{\, +}\varphi \| ^2)
$$
holds.
\end{theorem}

Theorem \ref{th1.3} is proved in Section 2. The following Theorem \ref{th1.4},
which is proved in Section 4, is a consequence of Theorem \ref{th1.3}.

\begin{theorem} \label{th1.4}
Let $d\geqslant 3$. Suppose $\widehat V^{\, (s)}\in L^2(K;{\mathcal S}^{(s)}_M)$, 
$s=0,1$, $A\in L^2(K;{\mathbb R}^d)$ with $A_0=0$, and there are a vector $\gamma 
\in \Lambda \backslash \{ 0\} $ and a measure $\mu \in {\mathfrak M}_{\mathfrak h}\, 
$, ${\mathfrak h}>0$, such that the magnetic potential $A$ obeys conditions 
$\bf (A_1)$, $\bf (\widetilde A_1)$, $\bf (A_2)$ and, for the functions $\widehat 
V^{\, (s)}$, $s=0,1$, the maps
$$
{\mathbb R}^d\ni x\to \{ [0,1]\ni \xi \to \widehat V^{\, (s)}(x-\xi \gamma )\}
\in L^2([0,1];{\mathcal M}_M)
$$
are continuous. Then for any $\delta \in (0,1)$ there exists a number $\varkappa _0
^{\, \prime }>0$ such that for all $\varkappa \geqslant \varkappa _0^{\, \prime }\, 
$, all vectors $k\in {\mathbb R}^d$ with $|(k,\gamma )|=\pi $, and all vector
functions $\varphi \in \widetilde H^1(K;{\mathbb C}^M)$ the inequality
$$
\| (C_2^{-\frac 12}\, \widehat G_-^{-\frac 12}\, \widehat P^{\, +}_*+
\widehat G_+^{-\frac 12}\, \widehat P^{\, -}_*)(\widehat {\mathcal D}
(k+i\varkappa e)+\widehat W)\varphi \| ^2=
$$ $$
=C_2^{-1}\, \| \widehat G_-^{-\frac 12}\, \widehat P^{\, +}_*(\widehat 
{\mathcal D}(k+i\varkappa e)+\widehat W)\varphi \| ^2+\| \widehat G_+
^{-\frac 12}\,\widehat P^{\, -}_*(\widehat {\mathcal D}(k+i\varkappa e)+
\widehat W)\varphi \| ^2\geqslant
$$ $$
\geqslant (1-\delta )\, (C_2\, \| \widehat G_-^{\, \frac 12}\, \widehat P^{\, 
-}\varphi \| ^2+\| \widehat G_+^{\, \frac 12}\, \widehat P^{\, +}\varphi \| ^2)
$$
is fulfilled.
\end{theorem}

Fix some nonnegative even function ${\mathbb R}^{d-1}\ni x^{\, \prime}
\to {\mathcal F}(x^{\, \prime})\in {\mathbb R}$ from the Schwartz space ${\mathcal S}
({\mathbb R}^{d-1};{\mathbb R})$ with the following property: the function
$$
{\mathbb R}^{d-1}\ni p\to {\mathcal F}\, \widehat {\phantom f}(p)=\int_{{\mathbb 
R}^{d-1}}{\mathcal F}(x^{\, \prime})\, e^{\, i(p,x^{\, \prime})}\, dx^{\, \prime}
$$
is also nonnegative and satisfies the conditions ${\mathcal F}\, {\widehat {\phantom 
f}}(0)=1$ and ${\mathcal F}\, {\widehat {\phantom f}}(p)=0$ for $|p|\geqslant 1$ 
(see, e.g., \cite{Dep00} and remarks after Chapter 1 in \cite{R}). Then we also
have ${\mathcal F}\, {\widehat {\phantom f}}\in {\mathcal S}({\mathbb R}^{d-1};
{\mathbb R})$ and $0\leqslant {\mathcal F}\, {\widehat {\phantom f}}(p)={\mathcal F}\,
{\widehat {\phantom f}}(-p)\leqslant 1$ for all $p\in {\mathbb R}^{d-1}$. 
Let $\{ {\mathcal E}^{\, \prime}_j\} $ be an orthogonal basis in ${\mathbb R}^d$
with ${\mathcal E}^{\, \prime}_2=e$. Denote $x^{\, \prime}_j=(x,{\mathcal E}^{\, 
\prime}_j)$, $x\in {\mathbb R}^d$, $j=1,\dots ,d$, $x^{\, \prime}\doteq (x^{\, 
\prime}_1,x^{\, \prime}_3,\dots ,x^{\, \prime}_d)\in {\mathbb R}^{d-1}$; 
$x_{\perp}=\sum\limits_{j\, \neq \, 2}x^{\, \prime}_j{\mathcal E}^{\, \prime}_j
\in {\mathbb R}^d$, $x_{\|}=x^{\, \prime}_2e\in {\mathbb R}^d$. For $R>0$, we set
$$
\widehat V^{\, (s)}_{\{ R\} }(x)=R^{d-1}\, \int_{{\mathbb R}^{d-1}}\widehat V^{(\, 
s)}\bigl( x-\sum\limits_{j\, \neq \, 2}x^{\, \prime}_j{\mathcal E}^{\, \prime}_j
\bigr) {\mathcal F}(Rx^{\, \prime}_1,Rx^{\, \prime}_3,\dots ,Rx^{\, \prime}_d)\, 
dx^{\, \prime}\, ,\ s=0,1,  \eqno (1.11)
$$ $$
\widehat V_{\{ R\} }(x)=\widehat V^{\, (0)}_{\{ R\} }(x)+\widehat V^{\, (1)}
_{\{ R\} }(x)\, ,
$$ $$
A_{\{ R\} }(x)=R^{d-1}\, \int_{{\mathbb R}^{d-1}}A\bigl( x-\sum\limits_{j\, \neq \,
2}x^{\, \prime}_j{\mathcal E}^{\, \prime}_j\bigr) {\mathcal F}(Rx^{\, \prime}_1,R
x^{\, \prime}_3,\dots ,Rx^{\, \prime}_d)\, dx^{\, \prime}\, ,\
x\in {\mathbb R}^d\, .  \eqno (1.12)
$$ 
Since
$$
(\widehat V^{\, (s)}_{\{ R\} })_N={\mathcal F}\, {\widehat {\phantom f}} \bigl( 
\frac {2\pi N^{\, \prime}}R\bigr) \, \widehat V^{(s)}_N\, ,\ s=0,1\, ,\ \ 
(A_{\{ R\} })_N={\mathcal F}\, {\widehat {\phantom f}} \bigl( \frac {2\pi N^{\, 
\prime}}R\bigr) \, A_N\, ,
$$
where $N^{\, \prime}=(N^{\, \prime}_1,N^{\, \prime}_3,\dots ,N^{\, \prime}_d)\in 
{\mathbb R}^{d-1}$, $N\in \Lambda ^*$, we have $(\widehat V^{\, (s)}_{\{ R\} })_N=0$ 
and $(A_{\{ R\} })_N=0$ for all $N\in \Lambda ^*$ with $2\pi |N_{\perp}|>R$. By the 
definition of the functions $\widehat V^{\, (s)}_{\{ R\} }\in L^2(K;{\mathcal 
S}^{(s)}_M)$ and $A_{\{ R\} }\in L^2(K;{\mathbb R}^d)$ (in the form of a convolution) 
and by the choice of the function ${\mathcal F}$, the function $A_{\{ R\} }$ (as well 
as the function $A$) satisfies conditions $\bf (A_0)$, $\bf (\widetilde A_1)$, 
$\bf (A_2)$ for all $R>0$, with the constants $C_{\varepsilon }(0,\widehat W_{\{ R\} 
})=C_{\varepsilon }(0,\widehat W)$, $\varepsilon >0$ (and inequalities (1.2) are 
satisfied with the constants $C_{\varepsilon }^{\, \prime }(0,\widehat W_{\{ R\} })
=C_{\varepsilon }^{\, \prime }(0,\widehat W)$, $\varepsilon >0$), moreover, the maps
$$
{\mathbb R}^d\ni x\to \{ [0,1]\ni \xi \to \widehat V^{\, (s)}_{\{ R\} }
(x-\xi \gamma )\} \in L^2([0,1];{\mathcal S}^{(s)}_M)\, ,\ s=0,1\, ,
$$ $$
{\mathbb R}^d\ni x\to \{ [0,1]\ni \xi \to A_{\{ R\} }
(x-\xi \gamma )\} \in L^2([0,1];{\mathbb R}^d)
$$ 
are continuous (the function $A_{\{ R\} }$ obeys condition $\bf (A_1)$), and
$$
||| \, \widehat W-\widehat W_{\{ R\} }\, ||| _{\gamma ,\, M}\to 0
$$ 
as $R\to +\infty $. The last relation and inequality (0.11) imply that for any
$\widetilde \varepsilon >0$ there is a number $R=R(\widetilde \varepsilon
)>0$ (dependent also on $\gamma $, $\mu $, and functions $\widehat V^{\, (s)}$,
$s=0,1$, and $A$) such that for all $\varepsilon >0$ and all vector functions
$\varphi \in H^1({\mathbb R}^d;{\mathbb C}^M)$ the inequality
$$
\| (\widehat W-\widehat W_{\{ R\} })\varphi \| _{L^2({\mathbb R}^d;{\mathbb C}^M)}
\leqslant \widetilde {\varepsilon }\, \bigl( \varepsilon \, \bigl\| -i\, \frac
{\partial \varphi }{\partial x^{\, \prime}_2}\, \bigr\| _{L^2({\mathbb R}^d;
{\mathbb C}^M)}+C^{\, \prime }(\gamma ,\varepsilon )\, \| \varphi \| _{L^2({\mathbb 
R}^d;{\mathbb C}^M)}\bigr)  \eqno (1.13)
$$
holds, where $x^{\, \prime}_2=(x,e)$, $x\in {\mathbb R}^d$, and $C^{\, \prime }(\gamma 
,\varepsilon )$ is the constant from inequality (0.11). From (1.13) it follows that
$$
\| (\widehat W-\widehat W_{\{ R\} })\varphi \| 
\leqslant \widetilde {\varepsilon }\, \bigl( \varepsilon \, \bigl\| \bigl( 
k^{\, \prime}_2-i\, \frac {\partial }{\partial x^{\, \prime}_2}\, \bigr)\, \varphi \,
\bigr\| +C^{\, \prime }(\gamma ,\varepsilon )\, \| \varphi \| \bigr)
\eqno (1.14)
$$
for all $\widetilde \varepsilon >0$ ($R=R(\widetilde \varepsilon 
)$), all $\varepsilon >0$, all $\varphi \in \widetilde H^1(K;{\mathbb C}^M)$, and
all $k\in {\mathbb R}^d$, where $k^{\, \prime}_2=(k,e)$.
\vskip 0.2cm

{\it Proof} of Theorem \ref{th1.1}. Given the vector $\gamma \in \Lambda \backslash 
\{ 0\} $ and the measure $\mu \in {\mathfrak M}_{\mathfrak h}\, $, suppose the 
functions $\widehat V^{\, (s)}$, $s=0,1$, and $A$ satisfy the conditions of 
Theorem \ref{th0.1}. For the number  
$$
\widetilde \varepsilon \doteq \frac 1{4\sqrt 3}\, C_2\bigl( 1+\, \frac {|\gamma 
|^2}{{\pi }^2}\, (C^{\, \prime }(\gamma ,1))^2\bigr) ^{-1/2}\, ,
$$
there is a number $R=R(\widetilde \varepsilon )>0$ such that inequality (1.14) holds 
for all $\varepsilon >0$, all vectors $k\in {\mathbb R}^d$, and all vector functions 
$\varphi \in \widetilde H^1(K;{\mathbb C}^M)$. Furthermore, $(\widehat W_{\{ R\} })_N
=\widehat 0$ for all $N\in \Lambda ^*$ with $2\pi |N_{\perp}|>R$, the function 
$A_{\{ R\} }$ obeys conditions $\bf (A_0)$, $\bf (A_1)$, $\bf (\widetilde A_1)$, 
$\bf (A_2)$ (with the vector $\gamma \in \Lambda \backslash \{ 0\} $, the measure 
$\mu \in {\mathfrak M}_{\mathfrak h}\, $, and the constants $C^*$, $\widetilde \theta 
$), and for the function $A_{\{ R\} }\, $, estimates (1.9) are fulfilled with the 
constants $Q(\gamma , A_{\{ R\} };\tau )=Q(\gamma , A;\tau )$, $\tau \in (0,1)$ 
(including the chosen number $\tau $). Therefore, Theorem \ref{th1.3} applied to the 
functions $\widehat V^{\, (s)}_{\{ R\} }$, $s=0,1$, $A_{\{ R\} }$ and the number 
$\delta =\frac 13\, $, implies that there exist a number $\widetilde a=\widetilde 
a(C_2;\frac 13\, ,R(\widetilde \varepsilon ))\in (0,C_2]$ and a number $\varkappa 
_0>0$ such that the inequality 
$$
\| (\widehat {\mathcal D}(k+i\varkappa e)+\widehat W_{\{ R\} })\varphi \| ^2\geqslant
\| (\widehat P^{\, +}_*+\widetilde a\, \widehat P^{\, -}_*)(\widehat {\mathcal D}(k+
i\varkappa e)+\widehat W_{\{ R\} })\varphi \| ^2\geqslant
$$ $$
\geqslant \frac 23\, (C^{\, 2}_2\, \| \widehat G_-^{\, 1}\widehat P^{\, -}\varphi \| 
^2+{\widetilde a}^{\, 2}\, \| \widehat G_+^{\, 1}\widehat P^{\, +}\varphi \| ^2)
$$
holds for all $\varkappa \geqslant \varkappa _0\, $, all vectors $k\in {\mathbb R}^d$,
with $|(k,\gamma )|=\pi $, and all vector functions $\varphi \in \widetilde H^1(K;
{\mathbb C}^M)$. Set $\varepsilon =\widetilde a\, (4\sqrt 6\, \widetilde \varepsilon 
\, )^{-1}$ and assume that $\widetilde a\varkappa _0\geqslant \pi |\gamma |^{-1}\, 
C_2$ and $\varepsilon \varkappa _0\geqslant C^{\, \prime }(\gamma ,\varepsilon )$. 
Then, for $\varkappa \geqslant \varkappa _0\, $, from (1.14) it follows that
$$
\| (\widehat W-\widehat W_{\{ R\} })\widehat P^{\, -}\varphi \| ^2\leqslant
2{\widetilde \varepsilon }^{\, 2}\, \bigl( \, \| \widehat G_-^{\, 1}
\widehat P^{\, -}\varphi \| ^2+(C^{\, \prime }(\gamma ,1))^2\, \| \widehat P^{\, 
-}\varphi \| ^2\, \bigr) \leqslant \frac 1{24}\, C^{\, 2}_2\, \| \widehat G_-^{\, 1}
\widehat P^{\, -}\varphi \| ^2\, ,
$$ $$
\| (\widehat W-\widehat W_{\{ R\} })\widehat P^{\, +}\varphi \| ^2\leqslant 
2{\widetilde \varepsilon }^{\, 2}\,
\bigl( \, \varepsilon ^2\, \| \widehat G_-^{\, 1}\widehat P^{\, +}\varphi \| ^2+
(C^{\, \prime }(\gamma ,\varepsilon ))^2\| \widehat P^{\, +}\varphi \| ^2\, \bigr) 
\leqslant 
$$ $$
\leqslant 2{\widetilde \varepsilon }^{\, 2}\,
\bigl( \, \varepsilon ^2+{\varkappa }^{-2}\, (C^{\, \prime }(\gamma ,\varepsilon 
))^2\, \bigr) \, \| \widehat G_+^{\, 1}\widehat P^{\, +}\varphi \| ^2 \leqslant
\frac 1{24}\, {\widetilde a}^{\, 2}\, \| \widehat G_+^{\, 1}\widehat P^{\, +}
\varphi \| ^2\, .
$$
Hence, for $\varkappa \geqslant \varkappa _0$ (and for all $k\in {\mathbb R}^d$ with
$|(k,\gamma )|=\pi $, and all $\varphi \in \widetilde H^1(K;{\mathbb C}^M)$) we
have
$$
\| (\widehat {\mathcal D}(k+i\varkappa e)+\widehat W)\varphi \| ^2\geqslant
$$ $$
\geqslant \frac 12\, \| (\widehat {\mathcal D}(k+i\varkappa e)+\widehat W_{\{ R\} })
\varphi \| ^2-2\, \| (\widehat W-\widehat W_{\{ R\} })\widehat P^{\, -}\varphi \| 
^2-2\, \| (\widehat W-\widehat W_{\{ R\} })\widehat P^{\, +}\varphi \| ^2\geqslant
$$ $$
\geqslant \frac 14\, \bigl( C^{\, 2}_2\, \| \widehat G_-^{\, 1}\widehat P^{\, -}
\varphi \| ^2+{\widetilde a}^{\, 2}\, \| \widehat G_+^{\, 1}\widehat P^{\, +}
\varphi \| ^2 \bigr) \geqslant
$$ $$
\geqslant \frac 14\, \bigl( C^{\, 2}_2\, \frac {{\pi }^2}{|\gamma |^2}\, \| 
\widehat P^{\, -}\varphi \| ^2+{\widetilde a}^{\, 2}\varkappa ^2\, \| 
\widehat P^{\, +}\varphi \| ^2 \bigr) \geqslant C_1^{\, 2}\, \| \varphi \| ^2\, ,
$$
where $C_1=\frac 12\, \pi |\gamma |^{-1}\, C_2\, $.  \hfill $\square$
\vskip 0.2cm

{\bf Remark}. Theorem \ref{th1.1} can be also proved using Theorem \ref{th1.4} (see
the proof of Theorem \ref{th1.2}).

\begin{theorem} \label{th1.5}
Let $d=3$, $\gamma \in \Lambda \backslash \{ 0\} $. Suppose 
$\widehat V^{\, (s)}\in L^3\ln ^{1+\delta }L(K;{\mathcal S}^{(s)}_M)$,
$s=0,1$, for some $\delta >0$. Then for any $\varepsilon ^{\, \prime }>0$ and any 
$\Xi >1$ there is a number $\varkappa _0=\varkappa _0(\varepsilon ^{\, \prime },
\Xi )>0$ $\mathrm ($dependent also on $\gamma $ and the functions $\widehat V^{\, 
(s)}$$\mathrm )$ such that for every $\varkappa _1\geqslant \varkappa _0$ there
exists a number $\varkappa \in [\varkappa _1,\Xi \varkappa _1]$ such that for all 
vectors $k\in {\mathbb R}^3$ with $|(k,\gamma )|=\pi $, and all vector functions 
$\varphi \in \widetilde H^1(K;{\mathbb C}^M)$ the inequality
$$
\| \, \widehat G_-^{\, -\frac 12}\, \widehat P^{\, +}_*\, \widehat V\varphi \, \| ^2+
\| \, \widehat G_+^{\, -\frac 12}\, \widehat P^{\, -}_*\, \widehat V\varphi \, \| ^2
\leqslant (\varepsilon ^{\, \prime })^2\, \bigl( \, \| \, \widehat G_+^{\, 
\frac 12}\, \widehat P^{\, +}\varphi \, \| ^2+\| \, \widehat G_-^{\, \frac 12}\, 
\widehat P^{\, -}\varphi \, \| ^2
$$
holds, where $\widehat V=\widehat V^{\, (0)}+\widehat V^{\, (1)}$.
\end{theorem}

\begin{theorem} \label{th1.6}
Let $d=3$, $\gamma \in \Lambda \backslash \{ 0\} $, $\sigma \in (0,2]$. Then there  
exist a universal constant $\widetilde c_1\in (0,C^{-1})$ and a constant
$\widetilde c_1^{\, \prime }>0$, dependent only on $\sigma $, such that for any 
$\varepsilon ^{\, \prime }>0$ and for all matrix functions 
$$
\widehat V^{\, (s)}=\sum\limits_{q=1}^{Q_s}\widehat V^{\, (s)}_q\, ,\ s=0,1\, ,
$$ 
that satisfy ${\mathrm (}$for $s=0$ and $s=1$${\mathrm )}$ the following conditions:

1${\mathrm )}$ $\widehat V^{\, (s)}_q\in L^3_w(K;{\mathcal S}^{(s)}_M)$ and $\beta 
_{\gamma ,\, \sigma }(0;\widehat V^{\, (s)}_q)<+\infty $, $q=1,\dots ,Q_s\, $,

2${\mathrm )}$ the essential supports ${\mathrm {supp}}\, \widehat V^{\, (s)}_q$    
of functions $\widehat V^{\, (s)}_q$ $\mathrm ($assumed to be $\Lambda $-periodic 
functions defined on space ${\mathbb R}^3$$\mathrm )$ do not intersect for different 
$q$,

3${\mathrm )}$ $\| \widehat V^{\, (s)}\| ^{(\infty , \, {\mathrm {loc}})}_{L^3_w(K;
{\mathcal M}_M)}=\max\limits_{q=1,\dots ,Q_s}\, \| \widehat V^{\, (s)}_q\| ^{(\infty 
, \, {\mathrm {loc}})}_{L^3_w(K;{\mathcal M}_M)}\leqslant \widetilde c_1\,
\varepsilon ^{\, \prime }$,

4${\mathrm )}$ $\max\limits_{q\, =\, 1,\dots ,Q_s}\, \beta _{\gamma ,\, \sigma }
(\widehat V^{\, (s)}_q)\leqslant \widetilde c_1^{\, \prime}\, \varepsilon 
^{\, \prime }\, ,$ \\
there exists a number $\varkappa _0^{\, \prime \prime \prime }>0$ such that for all
$\varkappa \geqslant \varkappa _0^{\, \prime \prime \prime }$, all vectors $k\in 
{\mathbb R}^3$ with $|(k,\gamma )|=\pi $, and all vector functions $\varphi \in 
\widetilde H^1(K;{\mathbb C}^M)$ the inequality
$$
\| \, \widehat G_-^{\, -\frac 12}\, \widehat P^{\, +}_*\, \widehat V\varphi \, \| ^2+
\| \, \widehat G_+^{\, -\frac 12}\, \widehat P^{\, -}_*\, \widehat V\varphi \, \| ^2
\leqslant (\varepsilon ^{\, \prime })^2\, \bigl( \, \| \, \widehat G_+^{\, 
\frac 12}\, \widehat P^{\, +}\varphi \, \| ^2+\| \, \widehat G_-^{\, \frac 12}\, 
\widehat P^{\, -}\varphi \, \| ^2\, \bigr)  \eqno (1.15)
$$
holds $\mathrm ($where $\widehat V=\widehat V^{\, (0)}+\widehat V^{\, (1)}$$\mathrm )$.
\end{theorem}

Theorems \ref{th1.5} and \ref{th1.6} are proved in Sections 5 and 6, respectively.

{\it Proof} of Theorem \ref{th1.2}. Suppose the functions $\widehat V^{\,
(s)}_{\nu }$, $s=0,1$, $\nu =1,2,3$, and $A$ satisfy the conditions of Theorem 
\ref{th0.2}, and the constants $c_1\in (0,C^{-1})$ and $c_1^{\, \prime }>0$
are chosen in accordance with Theorem \ref{th1.6} (see the conditions of Theorem 
\ref{th0.2} for the functions $\widehat V^{\, (s)}_3$): $c_1=\widetilde c_1\varepsilon 
^{\, \prime }$, $c_1^{\, \prime }=\widetilde c_1^{\, \prime }\varepsilon ^{\, \prime 
}$, here $\varepsilon ^{\, \prime }\doteq \frac 1{4\sqrt 3}\, C_2$ (the constant 
$C_2\in (0,1)$ (see (1.10)) is determined by the magnetic potential $A$ and by the
choice of the vector $\gamma \in \Lambda \backslash \{ 0\} $ and the measure $\mu 
\in {\mathfrak M}_{\mathfrak h}\, $, ${\mathfrak h}>0$; the choice of the vector 
$\gamma $ depends also on matrix functions $\widehat V^{\, (s)}_1$, $s=0,1$). Denote
$$
\widehat W_1=\widehat V_1-\sum\limits_{j=1}^3A_j\widehat \alpha _j\, .
$$
Given $\delta =\frac 13\, $, from Theorem \ref{th1.4} it follows that there is a
number $\varkappa _0^{\, \prime }>0$ such that for all $\varkappa \geqslant
\varkappa _0^{\, \prime }\, $, all vectors $k\in {\mathbb R}^3$ with $|(k,\gamma )|
=\pi $, and all vector functions $\varphi \in \widetilde H^1(K;{\mathbb C}^M)$
the following inequality is valid:
$$
\| \, \bigl( \, C_2^{-\frac 12}\, \widehat G_-^{\, -\frac 12}\, \widehat P^{\, 
+}_*+\widehat G_+^{\, -\frac 12}\, \widehat P^{\, -}_*\, \bigr) (\widehat 
{\mathcal D}(k+i\varkappa e)+\widehat W_1)\varphi \, \| ^2\geqslant  \eqno (1.16)
$$ $$
\geqslant \frac 23\, (\, C_2\, \| \, \widehat G_-^{\, \frac 12}\, \widehat P^{\, 
-}\varphi \, \| ^2+\| \, \widehat G_+^{\, \frac 12}\, \widehat P^{\, +}\varphi \,
\| ^2)\, .  
$$
Fix a number $\Xi >1$. From Theorems \ref{th1.5} and \ref{th1.6} (applied to the
matrix functions $\widehat V^{\, (s)}_2$ and $\widehat V^{\, (s)}_3\, $, respectively) 
it follows that there exists a number $\varkappa _0^{(\Xi )}\geqslant \max \, \{
\varkappa _0^{\, \prime },\pi |\gamma |^{-1}C_2\} $ such that for any $\varkappa _1
\geqslant \varkappa _0^{(\Xi )}$ there is a number $\varkappa \in [\varkappa _1,
\Xi \varkappa _1]$ such that for all $k\in {\mathbb R}^3$ with $|(k,\gamma )|
=\pi $, and all $\varphi \in \widetilde H^1(K;{\mathbb C}^M)$ we have
$$
\| \, \bigl( \, \widehat G_-^{\, -\frac 12}\, \widehat P^{\, 
+}_*+\widehat G_+^{\, -\frac 12}\, \widehat P^{\, -}_*\, \bigr) \, 
\widehat V_{\nu }\varphi \, \| ^2\leqslant  
(\varepsilon ^{\, \prime })^2\, \bigl( \, \| \, \widehat G_+^{\, \frac 12}\, 
\widehat P^{\, +}\varphi \, \| ^2+\| \, \widehat G_-^{\, \frac 12}\, 
\widehat P^{\, -}\varphi \, \| ^2\, \bigr) \, ,\ \nu =2,3\, ,  \eqno (1.17)
$$
where $\widehat V_{\nu }=\widehat V_{\nu }^{\, (0)}+\widehat V_{\nu }^{\, (1)}\, $.
Now, inequalities (1.16) and (1.17) (and also the relations $C_2\in (0,1)$, $G^{\, +}
_N(k;\varkappa )\geqslant \varkappa $ and $G^{\, -}_N(k;\varkappa )\leqslant G^{\, +}
_N(k;\varkappa )$, $N\in \Lambda ^*$) imply that for every $\varkappa _1\geqslant 
\varkappa _0^{(\Xi )}\geqslant \pi |\gamma |^{-1}C_2$ and for the number $\varkappa 
\in [\varkappa _1,\Xi \varkappa _1]$ chosen as above, the following estimates are
hold for all $k\in {\mathbb R}^3$ with $|(k,\gamma )|=\pi $, and all $\varphi \in 
\widetilde H^1(K;{\mathbb C}^M)$:
$$
C_2^{-1}\, \frac {|\gamma |}{\pi }\ \| \, (\widehat {\mathcal D}(k+i\varkappa 
e)+\widehat W)\varphi \| ^2\geqslant
$$ $$
\geqslant \, \| \, \bigl( \, C_2^{-\frac 12}\, \widehat G_-^{\, -\frac 12}\, 
\widehat P^{\, +}_*+\widehat G_+^{\, -\frac 12}\, \widehat P^{\, -}_*\, \bigr) 
(\widehat {\mathcal D}(k+i\varkappa e)+\widehat W)\varphi \, \| ^2\geqslant
$$ $$
\geqslant \, \frac 12\, \| \, \bigl( \, C_2^{-\frac 12}\, \widehat G_-^{\, 
-\frac 12}\, \widehat P^{\, +}_*+\widehat G_+^{\, -\frac 12}\, \widehat 
P^{\, -}_*\, \bigr) (\widehat {\mathcal D}(k+i\varkappa e)+\widehat W_1)\varphi 
\, \| ^2-
$$ $$
-\, 2\, \sum\limits_{\nu \, =\, 2}^3\| \, \bigl( \, C_2^{-\frac 12}\, \widehat 
G_-^{\, -\frac 12}\, \widehat P^{\, +}_*+\widehat G_+^{\, -\frac 12}\, \widehat 
P^{\, -}_*\, \bigr) \, \widehat V_{\nu }\varphi \, \| ^2\geqslant
$$ $$
\geqslant \frac 13\, \bigl( C_2\, \| \, \widehat G_-^{\, \frac 12}\, \widehat P^{\, -}
\varphi \, \| ^2+\, \| \, \widehat G_+^{\, \frac 12}\, \widehat P^{\, +}
\varphi \, \| ^2 \bigr) -
\, 2C_2^{-1}\sum\limits_{\nu \, =\, 2}^3\| \, \bigl( \, \widehat 
G_-^{\, -\frac 12}\, \widehat P^{\, +}_*+\widehat G_+^{\, -\frac 12}\, \widehat 
P^{\, -}_*\, \bigr) \, \widehat V_{\nu }\varphi \, \| ^2\geqslant
$$ $$
\geqslant \frac 14\, \bigl( C_2\, \| G_-^{\, \frac 12}\, \widehat P^{\, -}
\varphi \, \| ^2+ \| G_+^{\, \frac 12}\, \widehat P^{\, +}\varphi \, \| ^2 
\bigr) \geqslant 
$$ $$
\geqslant \,\frac 14\, \bigl( C_2\, \frac {\pi }{|\gamma |}\, \| \,
\widehat P^{\, -}\varphi \, \| ^2+\varkappa \, \| \, \widehat P^{\, +}\varphi 
\, \| ^2 \bigr) \geqslant C_2^{-1}\, \frac {|\gamma |}{\pi }\, C_1^{\, 2}\, \| 
\, \varphi \, \| ^2\, ,
$$
where $C_1=\frac 12\, \pi |\gamma |^{-1}\, C_2\, $.  \hfill $\square$

\section{Proof of Theorem \ref{th1.3}}

Given the vector $\gamma \in \Lambda \backslash \{ 0\} $ ($e=|\gamma
|^{-1}\gamma $), for all $\varepsilon \in (0,1)$, define the sets
$$
{\mathcal C}(\varepsilon )={\mathcal C}(k,\varkappa ;\varepsilon )=\{ N\in
\Lambda ^*: |\varkappa -|k_{\perp}+2\pi N_{\perp}||<\varepsilon \varkappa \} \, ,
$$
$k\in {\mathbb R}^d$, $\varkappa >0$.

In this Section, Theorem \ref{th1.3} is deduced from Theorem \ref{th2.1} which is 
a weakened variant of Theorem \ref{th1.3}. Theorem \ref{th2.1} is proved in Section 3.

\begin{theorem} \label{th2.1}
Let $d\geqslant 3$. Suppose $\widehat V^{\, (s)}\in L^2(K;{\mathcal S}^{(s)}_M)$, 
$s=0,1$, $A\in L^2(K;{\mathbb R}^d)$ with $A_0=0$, $R>0$, and there are a vector  
$\gamma \in \Lambda \backslash \{ 0\} $ and a measure $\mu \in {\mathfrak M}
_{\mathfrak h}\, $, ${\mathfrak h}>0$, such that for the magnetic potential $A$,
conditions $\bf (A_1)$, $\bf (\widetilde A_1)$, $\bf (A_2)$ are fulfilled, and, 
moreover, $\widehat V^{\, (s)}_N=\widehat 0$, $s=0,1$, and $A_N=0$ for all $N \in 
\Lambda ^*$ with $2\pi |N_{\perp}|>R$. Then for any $\delta \in (0,1)$ there exist 
numbers $a=a(C_2;\delta ,R)\in (0,C_2]$ and $\varkappa _0>2R$ $\mathrm ($the number 
$\varkappa _0$ depends on $\delta $, $|\gamma |$, ${\mathfrak h}$, $\| \mu \| $, $R$ 
and on the constants $C_{\varepsilon }^{\, \prime }(0,\widehat W)$, $C^*$, $\tau $, 
$Q$, $\widetilde \theta $$\mathrm )$ such that for all $\varkappa \geqslant \varkappa 
_0\, $, all vectors $k\in {\mathbb R}^d$ with $|(k,\gamma )|=\pi $, and all vector
functions $\varphi \in \widetilde H^1(K;{\mathbb C}^M)\cap {\mathcal H}({\mathcal C}
(\frac 12))$ the inequality
$$
\| \, \widehat P^{\, +}\, (\widehat {\mathcal D}(k+i\varkappa e)
+\widehat W)\varphi \, \| ^2+ a^{\, 2}\, \| \, \widehat P^{\, -}\, (\widehat 
{\mathcal D}(k+i\varkappa e)+\widehat W)\varphi \| ^2\geqslant  \eqno (2.1)
$$ $$
\geqslant (1-\delta )\, (\, C^2_2\, \| \, \widehat G_-^{\, 1}\, \widehat P^{\, -}
\varphi \, \| ^2+a^{\, 2}\, \| \, \widehat G_+^{\, 1}\, \widehat P^{\, +}\varphi \, 
\| ^2\, )\, ,
$$
holds, where the constant $C_2$ is defined in (1.10).
\end{theorem}

{\bf Remark}. Since $\varkappa _0>2R$, we see that for all $\varkappa \geqslant 
\varkappa _0$ and all vector functions $\varphi \in \widetilde H^1(K;{\mathbb C}^M)
\cap {\mathcal H}({\mathcal C}(\frac 12))$ the equality $(\widehat W\varphi )_N=0$ 
holds for $N \in \Lambda ^*$ with $|k_{\perp}+2\pi N_{\perp}|\leqslant \frac 
{\varkappa}2-R$ (in particular, $(\widehat W\varphi )_N=0$ for $N \in \Lambda ^*$ with 
$k_{\perp}+2\pi N_{\perp}=0$). Hence, in the left hand side of inequality (2.1), the 
orthogonal projections $\widehat P^{\, \pm }$ may be replaced by the orthogonal 
projections $\widehat P^{\, \pm }_*\, $.

\begin{lemma} \label{l2.1}
Under the conditions of Theorem \ref{th1.2}, for any $\varepsilon >0$ there is 
a number $\widetilde \varkappa _0=\widetilde \varkappa _0(\varepsilon )>0$ such that
for all $\varkappa \geqslant \widetilde \varkappa _0\, $, all vectors $k\in {\mathbb 
R}^d$, and all vector functions $\varphi \in \widetilde H^1(K;{\mathbb C}^M)\cap 
{\mathcal H}(\Lambda ^*\backslash {\mathcal C}(\frac 14))$ the estimate
$$
\| \widehat W\varphi \| \leqslant \varepsilon \, \| \widehat G_-^{\, 1}\varphi \| 
$$
holds.
\end{lemma}

\begin{proof}
Indeed, set $\widetilde \varkappa _0=8{\varepsilon}^{-1}C_{\varepsilon
/10}(0,\widehat W)$ (here $C_{\varepsilon /10}(0,\widehat W)$ is the constant from 
(0.5)). Since $G^-_N(k;\varkappa )\geqslant \frac {\varkappa}4$ and $G^-_N(k;
\varkappa )\geqslant \frac 15\, |k+2\pi N|$ for all $N\in \Lambda ^*\backslash 
{\mathcal C}(\frac 14)$, we obtain (see (1.1))
$$
\| \widehat W\varphi \| \leqslant \frac {\varepsilon}{10}\, \bigl\| \, \sum\limits
_{j\, =\, 1}^d\, \bigl( k_j-i\, \frac {\partial}{\partial x_j}\, \bigr) \, \widehat
\alpha _j\varphi \, \bigr\| +C_{\varepsilon /10}(0,\widehat W)\, \| \varphi \| =
$$ $$
=\, \frac {\varepsilon}{10}\ v^{1/2}(K)\, \bigl( \, \sum\limits_{N\, \in 
\, \Lambda ^*\backslash {\mathcal C}(\frac 14)}|k+2\pi N|^2\, \| \varphi _N 
\| ^2\, \bigr) ^{1/2}+C_{\varepsilon /10}(0,\widehat W)\, \| \varphi \| \leqslant
$$ $$
\leqslant \bigl(\, \frac {\varepsilon}2+\frac 4{\varkappa}\, C_{\varepsilon /10}(0,
\widehat W)\, \bigr) \, \| \widehat G_-^{\, 1}\varphi \| \leqslant \varepsilon \, \| 
\widehat G_-^{\, 1}\varphi \| 
$$
for all $\varkappa \geqslant \widetilde \varkappa _0$ (all $k\in {\mathbb R}^d$ and 
all $\varphi \in \widetilde H^1(K;{\mathbb C}^M)\cap {\mathcal H}(\Lambda ^*
\backslash {\mathcal C}(\frac 14))$).
\end{proof}

{\it Proof} of Theorem \ref{th1.3}. We write $\delta _1=\frac {\delta}3\, $. Let 
$\widetilde a=a(C_2\, ;\delta _1 ,R)\in (0,C_2]$ and $\varkappa _0$ be the numbers
defined in Theorem \ref{th2.1}. Denote
$$
\varepsilon =\frac {\delta _1}{\sqrt{6(1-\delta _1)}}\, \min \, \{ C_2,
\widetilde a\} \, .
$$
Without loss of generality we assume that $\varkappa _0\geqslant 4R$ and $\varkappa 
_0\geqslant \widetilde \varkappa _0(\varepsilon )$, where $\widetilde \varkappa 
_0(\varepsilon )$ is the number from Lemma \ref{l2.1}. In what follows, we also 
assume that $\varkappa \geqslant \varkappa _0$ and the vector $k\in {\mathbb R}^d$
satisfies the equality $|(k,\gamma )|=\pi $. For the vector function $\varphi \in 
\widetilde H^1(K;{\mathbb C}^M)$, the equality 
$$
\| (\widehat P^{\, +}+\widetilde a\widehat P^{\, -})(\widehat {\mathcal D}
(k+i\varkappa e)+\widehat W)\, \varphi \| ^2=  \eqno (2.2)
$$ $$
=\, \| \widehat P^{\, {\mathcal C}(\frac 12)}\, (\widehat P^{\, +}+\widetilde 
a\widehat P^{\, -})(\widehat {\mathcal D}(k+i\varkappa e)+\widehat W)\, \varphi 
\| ^2+
$$ $$
+\, \| \widehat P^{\, \Lambda ^*\, \backslash \, {\mathcal C}(\frac 12)}\, 
(\widehat P^{\, +}+\widetilde a\widehat P^{\, -})(\widehat {\mathcal D}
(k+i\varkappa e)+\widehat W)\, \varphi \| ^2
$$
holds. We shall estimate the summands in the right hand side of (2.2). Since $\varkappa 
\geqslant \varkappa _0\geqslant 4R$ and $\widehat W_N=0$ for all $N\in \Lambda ^*$
with $2\pi |N_{\perp}|>R$, we see that $\widehat W_N=0$ for all $N\in \Lambda ^*$
with $2\pi |N_{\perp}|>\frac {\varkappa}4\, $. The last assertion will be used below
to obtain necessary estimates. We have
$$
\| \widehat P^{\, {\mathcal C}(\frac 12)}\, (\widehat P^{\, +}+\widetilde a
\widehat P^{\, -})(\widehat {\mathcal D}(k+i\varkappa e)+\widehat W)\, \varphi \| 
=  \eqno (2.3)
$$ $$
=\, \| \widehat P^{\, {\mathcal C}(\frac 12)}\, (\widehat P^{\, +}+\widetilde a
\widehat P^{\, -})(\widehat {\mathcal D}(k+i\varkappa e)+\widehat W)(\varphi 
^{\, {\mathcal C}(\frac 12)}+\varphi ^{\, {\mathcal C}(\frac 34)\, \backslash 
\, {\mathcal C}(\frac 12)}) \|
\geqslant
$$ $$
\geqslant \| (\widehat P^{\, +}+\widetilde a\widehat P^{\, -})(\widehat {\mathcal D}
(k+i\varkappa e)+\widehat W)\, \varphi ^{\, {\mathcal C}(\frac 12)}\, \| -
$$ $$
-\, \| \widehat P^{\, \Lambda ^*\, \backslash \, {\mathcal C}(\frac 12)}\, 
(\widehat P^{\, +}+\widetilde a\widehat P^{\, -})(\widehat {\mathcal D}
(k+i\varkappa e)+\widehat W)\, \varphi ^{\, {\mathcal C}(\frac 12)}\| -
$$ $$
-\, \| \widehat P^{\, {\mathcal C}(\frac 12)}\, (\widehat P^{\, +}+\widetilde a
\widehat P^{\, -})(\widehat {\mathcal D}(k+i\varkappa e)+\widehat W)\, \varphi 
^{\, {\mathcal C}(\frac 34)\, \backslash \, {\mathcal C}(\frac 12)}\, \| \, .
$$ 
Using Lemma \ref{l2.1}, we get
$$
\| \widehat P^{\, \Lambda ^*\, \backslash \, {\mathcal C}(\frac 12)}\, 
(\widehat P^{\, +}+\widetilde a\widehat P^{\, -})(\widehat {\mathcal D}
(k+i\varkappa e)+\widehat W)\, \varphi ^{\, {\mathcal C}(\frac 12)}\, 
\| =  \eqno (2.4)
$$ $$
=\, \| \, \widehat P^{\, \Lambda ^*\, \backslash \, {\mathcal C}(\frac 12)}\, 
(\widehat P^{\, +}+\widetilde a\widehat P^{\, -})\widehat W \varphi ^{\, 
{\mathcal C}(\frac 12)\, \backslash \, {\mathcal C}(\frac 14)}\, \| \leqslant
\| \widehat W \varphi ^{\, {\mathcal C}(\frac 12)\, \backslash \, {\mathcal C}(\frac 
14)}\, \| \leqslant \varepsilon \, \| \widehat G_-^{\, 1}\varphi ^{\, {\mathcal 
C}(\frac 12)\, \backslash \, {\mathcal C}(\frac 14)}\, \| \, ,
$$ $$
\| \widehat P^{\, {\mathcal C}(\frac 12)}\, (\widehat P^{\, +}+\widetilde a
\widehat P^{\, -})(\widehat {\mathcal D}(k+i\varkappa e)+\widehat W)\, \varphi 
^{\, {\mathcal C}(\frac 34)\, \backslash \, {\mathcal C}(\frac 12)}\, 
\| =  \eqno (2.5)
$$ $$
=\, \| \widehat P^{\, {\mathcal C}(\frac 12)}\, (\widehat P^{\, +}+\widetilde a
\widehat P^{\, -})\widehat W \varphi ^{\, {\mathcal C}(\frac 34)\, \backslash 
\, {\mathcal C}(\frac 12)}\, \|
\leqslant
\| \widehat W \varphi ^{\, {\mathcal C}(\frac 34)\, \backslash \, {\mathcal C}(\frac 
12)}\, \| \leqslant \varepsilon \, \| \widehat G_-^{\, 1}\varphi 
^{\, {\mathcal C}(\frac 34)\, \backslash \, {\mathcal C}(\frac 12)}\, \| \, .
$$
On the other hand, from Theorem \ref{th2.1} we derive
$$
\| (\widehat P^{\, +}+\widetilde a\widehat P^{\, -})(\widehat {\mathcal D}
(k+i\varkappa e)+\widehat W)\, \varphi ^{\, {\mathcal C}(\frac 12)}\, \| ^2\geqslant
$$ $$
\geqslant (1-\delta _1)\, (C^2_2\, \| \widehat G_-^{\, 1}\widehat P^{\, -}
\varphi ^{\, {\mathcal C}(\frac 12)}\, \| ^2+\widetilde a^{\, 2}\, \| 
\widehat G_+^{\, 1}\widehat P^{\, +}\varphi ^{\, {\mathcal C}(\frac 12)}\, \| ^2 )\, .
$$
Consequently, from (2.3), (2.4), and (2.5) we obtain
$$
\| \widehat P^{\, {\mathcal C}(\frac 12)}\, (\widehat P^{\, +}+\widetilde a
\widehat P^{\, -})(\widehat {\mathcal D}(k+i\varkappa e)+\widehat W)\, \varphi 
\| ^2 \geqslant  \eqno (2.6)
$$ $$
\geqslant (1-\delta _1)\, \|  (\widehat P^{\, +}+\widetilde a\widehat P^{\, -})
(\widehat {\mathcal D}(k+i\varkappa e)+\widehat W)\, \varphi ^{\, 
{\mathcal C}(\frac 12)}\, \| ^2-
$$ $$
-\ \frac {2(1-\delta _1)}{\delta _1}\ \| \widehat P^{\, \Lambda ^*\, \backslash \, 
{\mathcal C}(\frac 12)}\, (\widehat P^{\, +}+\widetilde a\widehat P^{\, -})
(\widehat {\mathcal D}(k+i\varkappa e)+\widehat W)\, \varphi ^{\, {\mathcal C}
(\frac 12)}\, \| ^2 -
$$ $$
-\ \frac {2(1-\delta _1)}{\delta _1}\ \| \widehat P^{\, {\mathcal C}(\frac 12)}\, 
(\widehat P^{\, +}+\widetilde a\widehat P^{\, -})(\widehat {\mathcal D}
(k+i\varkappa e)+\widehat W)\, \varphi ^{\, {\mathcal C}(\frac 34)\, \backslash 
\, {\mathcal C}(\frac 12)}\, \| ^2\geqslant
$$ $$
\geqslant (1-\delta _1)^2\, (C^2_2\, \| \widehat G_-^{\, 1}\widehat P^{\, -}
\varphi ^{\, {\mathcal C}(\frac 12)}\, \| ^2+\widetilde a^{\, 2}\, \| \widehat 
G_+^{\, 1}\widehat P^{\, +}\varphi ^{\, {\mathcal C}(\frac 12)}\, \| ^2 )\, -
$$ $$
-\ \frac {2(1-\delta _1)}{\delta _1}\ \varepsilon ^2\, \bigl( \, \| \widehat 
G_-^{\, 1}\varphi ^{\, {\mathcal C}(\frac 12)\, \backslash \, {\mathcal C}
(\frac 14)}\, \| ^2+\| \widehat G_-^{\, 1}\varphi ^{\, {\mathcal C}(\frac 34)\, 
\backslash \, {\mathcal C}(\frac 12)}\, \| ^2 \bigr) \geqslant
$$ $$
\geqslant (1-2\delta _1)\, (C^2_2\, \| \widehat G_-^{\, 1}\widehat P^{\, -}
\varphi ^{\, {\mathcal C}(\frac 12)}\, \| ^2+\widetilde a^{\, 2}\, \| \widehat 
G_+^{\, 1}\widehat P^{\, +}\varphi ^{\, {\mathcal C}(\frac 12)}\, \| ^2 )\, -
$$ $$
-\ \frac {2(1-\delta _1)}{\delta _1}\ \varepsilon ^2\, \| \widehat G_-^{\, 1}
\varphi ^{\, {\mathcal C}(\frac 34)\, \backslash \, {\mathcal C}(\frac 14)}\,
\| ^2\, .
$$
Let us estimate the second summand in the right hand side of (2.2). By Lemma 
\ref{l2.1}, we have
$$
\, \| \widehat P^{\, \Lambda ^*\, \backslash \, {\mathcal C}(\frac 12)}\, 
(\widehat P^{\, +}+\widetilde a\widehat P^{\, -})(\widehat {\mathcal D}
(k+i\varkappa e)+\widehat W)\, \varphi \| =
$$ $$
=\, \| \widehat P^{\, \Lambda ^*\, \backslash \, {\mathcal C}(\frac 12)}\, 
(\widehat P^{\, +}+\widetilde a\widehat P^{\, -})(\widehat {\mathcal D}
(k+i\varkappa e)+\widehat W)\, (\varphi ^{\, \Lambda ^*\, \backslash \, 
{\mathcal C}(\frac 12)}+\varphi ^{\, {\mathcal C}(\frac 12)\, \backslash \, 
{\mathcal C}(\frac 14)}\, )\| \geqslant
$$ $$
\geqslant \| (\widehat P^{\, +}+\widetilde a\widehat P^{\, -})(\widehat 
{\mathcal D}(k+i\varkappa e)+\widehat W)\, \varphi ^{\, \Lambda ^*\, \backslash 
\, {\mathcal C}(\frac 12)}\, \| \, -
$$ $$
-\, \| \widehat P^{\, {\mathcal C}(\frac 12)\, \backslash \, {\mathcal C}(\frac 14)}\,
(\widehat P^{\, +}+\widetilde a\widehat P^{\, -})(\widehat {\mathcal D}
(k+i\varkappa e)+\widehat W)\, \varphi ^{\, \Lambda ^*\, \backslash \, 
{\mathcal C}(\frac 12)}\, \| \, -
$$ $$
-\, \| \widehat P^{\, \Lambda ^*\, \backslash \, {\mathcal C}(\frac 12)}\, (\widehat 
P^{\, +}+\widetilde a\widehat P^{\, -})(\widehat
{\mathcal D}(k+i\varkappa e)+\widehat W)\, \varphi ^{\, {\mathcal C}(\frac 12)\, 
\backslash \, {\mathcal C}(\frac 14)}\, \| =
$$ $$
=\, \| (\widehat P^{\, +}+\widetilde a\widehat P^{\, -})(\widehat 
{\mathcal D}(k+i\varkappa e)+\widehat W)\, \varphi ^{\, \Lambda ^*\, 
\backslash \, {\mathcal C}(\frac 12)}\, \| \, -
$$ $$
-\, \| \widehat P^{\, {\mathcal C}(\frac 12)\, \backslash \, {\mathcal C}(\frac 14)}\,
(\widehat P^{\, +}+\widetilde a\widehat P^{\, -})\widehat W \varphi ^{\, 
{\mathcal C}(\frac 34)\, \backslash \, {\mathcal C}(\frac 12)}\, \| \, -
$$ $$
-\, \| \widehat P^{\, \Lambda ^*\, \backslash \, {\mathcal C}(\frac 12)}\, 
(\widehat P^{\, +}+\widetilde a\widehat P^{\, -})\widehat W \varphi ^{\, {\mathcal 
C}(\frac 12)\, \backslash \, {\mathcal C}(\frac 14)}\, \| \geqslant
$$ $$
\geqslant \| (\widehat P^{\, +}+\widetilde a\widehat P^{\, -})\widehat 
{\mathcal D}(k+i\varkappa e)\, \varphi ^{\, \Lambda ^*\, \backslash \, {\mathcal 
C}(\frac 12)}\, \| \, -
$$ $$
-\, \| \widehat W\varphi ^{\, \Lambda ^*\, \backslash \, {\mathcal C}(\frac 12)}\, 
\| -\| \widehat W\varphi ^{\, {\mathcal C}(\frac 34)\, \backslash \, {\mathcal 
C}(\frac 12)}\, \| - \| \widehat W\varphi ^{\, {\mathcal C}(\frac 12)\, \backslash 
\, {\mathcal C}(\frac 14)}\, \| \geqslant
$$ $$
\geqslant \| (\widehat P^{\, +}+\widetilde a\widehat P^{\, -})\widehat 
{\mathcal D}(k+i\varkappa e)\, \varphi ^{\, \Lambda ^*\, \backslash \, 
{\mathcal C}(\frac 12)}\, \| \, -
$$ $$
-\, \varepsilon \, \bigl( \, \| \widehat G_-^{\, 1}\varphi ^{\, \Lambda ^*\, 
\backslash \, {\mathcal C}(\frac 12)}\, \| + \| \widehat G_-^{\, 1}\varphi ^{\, 
{\mathcal C}(\frac 34)\, \backslash \, {\mathcal C}(\frac 12)}\, \| + \| 
\widehat G_-^{\, 1}\varphi ^{\, {\mathcal C}(\frac 12)\, \backslash \, 
{\mathcal C}(\frac 14)}\, \| \, \bigr) \, .
$$
Therefore (taking into account inequalities (1.8) and the choice of the number
$C_2\in (0,1)$),
$$
\, \| \widehat P^{\, \Lambda ^*\, \backslash \, {\mathcal C}(\frac 12)}\, 
(\widehat P^{\, +}+\widetilde a\widehat P^{\, -})(\widehat {\mathcal D}
(k+i\varkappa e)+\widehat W)\, \varphi \| ^2 \geqslant
$$ $$
\geqslant (1-2\delta _1)\, \| (\widehat P^{\, +}+\widetilde a\widehat P^{\, -})
\widehat {\mathcal D}(k+i\varkappa e)\, \varphi ^{\, \Lambda ^*\, \backslash \, 
{\mathcal C}(\frac 12)}\, \| ^2 -
$$ $$
-\ \frac {3(1-2\delta _1)}{2\delta _1}\ \varepsilon ^2\, \| \widehat G_-^{\, 1}
\varphi ^{\, \Lambda ^*\, \backslash \, {\mathcal C}(\frac 12)}\, \| ^2 -
$$ $$
-\ \frac {3(1-2\delta _1)}{2\delta _1}\ \varepsilon ^2\, \bigl( \, \| \widehat 
G_-^{\, 1}\varphi ^{\, {\mathcal C}(\frac 34)\, \backslash \, {\mathcal C}
(\frac 12)}\, \| ^2+ \| \widehat G_-^{\, 1}\varphi ^{\, {\mathcal C}(\frac 12)\, 
\backslash \, {\mathcal C}(\frac 14)}\, \| ^2\, \bigr) =
$$ $$
=\, (1-2\delta _1)\, \bigl( \, \| \widehat G_-^{\, 1}\widehat P^{\, -}\varphi 
^{\, \Lambda ^*\, \backslash \, {\mathcal C}(\frac 12)}\, \| ^2 +\widetilde a^{\, 
2}\, \| \widehat G_+^{\, 1}\widehat P^{\, +}\varphi ^{\, \Lambda ^*\, \backslash 
\, {\mathcal C}(\frac 12)}\, \| ^2\, \bigr) \, -
$$ $$
-\ \frac {3(1-2\delta _1)}{2\delta _1}\ \varepsilon ^2\, \bigl( \, \| \widehat 
G_-^{\, 1}\varphi ^{\, \Lambda ^*\, \backslash \, {\mathcal C}(\frac 12)}\, \| 
^2 + \| \widehat G_-^{\, 1}\varphi ^{\, {\mathcal C}(\frac 34)\, \backslash \, 
{\mathcal C}(\frac 14)}\, \| ^2\, \bigr) \geqslant
$$ $$
\geqslant (1-2\delta _1)\, \bigl( \, C_2^2\, \| \widehat G_-^{\, 1}\widehat P^{\, 
-}\varphi ^{\, \Lambda ^*\, \backslash \, {\mathcal C}(\frac 12)}\, \| ^2 +
\widetilde a^{\, 2}\, \| \widehat G_+^{\, 1}\widehat P^{\, +}\varphi ^{\, \Lambda 
^*\, \backslash \, {\mathcal C}(\frac 12)}\, \| ^2\, \bigr) \, -
$$ $$
-\ \frac {2(1-\delta _1)}{\delta _1}\ \varepsilon ^2\, \bigl( \, \| \widehat 
G_-^{\, 1}\varphi ^{\, \Lambda ^*\, \backslash \, {\mathcal C}(\frac 12)}\, \| ^2 + 
\| \widehat G_-^{\, 1}\varphi ^{\, {\mathcal C}(\frac 34)\, \backslash \, 
{\mathcal C}(\frac 14)}\, \| ^2\, \bigr)
\, .
$$
From the last inequality and from (2.2), (2.6), it follows that the estimate
$$
\| (\widehat P^{\, +}+\widetilde a\widehat P^{\, -})(\widehat {\mathcal D}
(k+i\varkappa e)+\widehat W)\, \varphi \| ^2 \geqslant  \eqno (2.7)
$$ $$
\geqslant (1-2\delta _1)\, \bigl( \, C_2^2\, \| \widehat G_-^{\, 1}\widehat 
P^{\, -}\varphi  \| ^2 +\widetilde a^{\, 2}\, \| \widehat G_+^{\, 1}\widehat P^{\, 
+}\varphi \| ^2\, \bigr) \, -\frac {6(1-\delta _1)}{\delta _1}\ \varepsilon ^2\, 
\| \widehat G_-^{\, 1}\varphi \| ^2 
$$
holds. Finally, estimate (2.7) and the estimate
$$
\frac {6(1-\delta _1)}{\delta _1}\ \varepsilon ^2\, \| \widehat G_-^{\, 1}\varphi 
\| ^2=\frac {6(1-\delta _1)}{\delta _1}\ \varepsilon ^2\, \bigl( \, \| \widehat 
G_-^{\, 1}\widehat P^{\, -}\varphi \| ^2+\| \widehat G_-^{\, 1}\widehat P^{\, +}
\varphi \| ^2\, \bigr) \leqslant
$$ $$
\leqslant \delta _1\, \bigl( \, C_2^2\, \| \widehat G_-^{\, 1}\widehat P^{\, -}
\varphi \| ^2+\widetilde a^{\, 2}\, \| \widehat G_-^{\, 1}\widehat P^{\, +}\varphi 
\| ^2\, \bigr) \leqslant \delta _1\, \bigl( \, C_2^2\, \| \widehat G_-^{\, 1}
\widehat P^{\, -}\varphi \| ^2+\widetilde a^{\, 2}\, \| \widehat G_+^{\, 1}
\widehat P^{\, +}\varphi \| ^2\, \bigr)
$$
imply inequality (2.1). This completes the proof.

\section{Proof of Theorem \ref{th2.1}}

First we shall show that it suffices to assume that the magnetic potential $A$ is a
trigonometric polynomial. Indeed, suppose the functions $\widehat V^{\, (s)}$, $s=0,
1$, and $A$ satisfy the conditions of Theorem \ref{th2.1} (in particular, for the 
vector $\gamma \in \Lambda \backslash \{ 0\} $ and the measure $\mu \in {\mathfrak M}
_{\mathfrak h}\, $, ${\mathfrak h}>0$, conditions $\bf (A_1)$, $\bf (\widetilde A_1)$, 
$\bf (A_2)$ are fulfilled, moreover, $A_0=0$ and $\widehat V^{\, (s)}_N=\widehat 0$, 
$s=0,1$, $A_N=0$ for all $N\in \Lambda ^*$ with $2\pi |N_{\perp}|>R$). Let 
${\mathcal G}:{\mathbb R}^d\to {\mathbb R}$ be a nonnegative function from the
Schwartz space ${\mathcal S}({\mathbb R}^d;{\mathbb R})$ such that for the function
$$
{\mathbb R}^d\ni p\to {\mathcal G}\, \widehat {\phantom f}(p)=
\int_{{\mathbb R}^d}{\mathcal G}(x)\, e^{\, i(p,x)}\, dx
$$
we have $\mathcal G\, \widehat {\phantom f}(0)=1$ and $\mathcal 
G\, \widehat {\phantom f}(p)=0$ for $|p|\geqslant 1$ (we may assume that the
function ${\mathcal G}$ coincides with the function ${\mathcal F}:{\mathbb R}^{d-1}
\to {\mathbb R}$ considered in Section 1 if we change ${\mathbb R}^{d-1}$ to 
${\mathbb R}^d$). Let us denote
$$
A^{\{ r\} }(x)=r^d \int_{{\mathbb R}^d}A(x-y)\, {\mathcal G}
(ry)\, dy\, ,\ r>0\, ,\ x\in {\mathbb R}^d\, .
$$
For any $r>0$, the functions $\widehat V^{\, (s)}$ and $A^{\{ r\} }$ (as well as the
functions $\widehat V^{\, (s)}$ and $A$) satisfy all conditions of Theorem \ref{th2.1} 
(with the vector $\gamma $, the measure $\mu $ and the constants $C_{\varepsilon 
}^{\, \prime }(0,\widehat V-\sum_{j=1}^dA_j^{\{ r\} }\widehat \alpha _j)=
C_{\varepsilon }^{\, \prime }(0,\widehat W)$, $C^*$, $\tau $, $Q$, $\widetilde 
\theta $). Furthermore,
$$
(A^{\{ r\} })_N={\mathcal G}\, \widehat {\phantom f}\bigl( \, -\, \frac 
{2\pi N}r\, \bigr) \, A_N\, ,\ N\in \Lambda ^*\, ,
$$
hence, $(A^{\{ r\} })_N=0$ for $2\pi |N|\geqslant r\, $. Now, if we suppose that
Theorem \ref{th2.1} is already proved for the matrix functions $\widehat V^{\, 
(s)}$, $s=0,1$, and the magnetic potentials $A^{\{ r\} }$ (which are trigonometric
polynomials), $r>0$, then for any $\delta \in (0,1)$ there exist numbers $\varkappa 
_0>2R$ and $a\in (0,1)$ (independent of the number $r>0$) such that for all 
$\varkappa \geqslant \varkappa _0\, $, all vectors $k\in {\mathbb R}^d$ with $|(k,
\gamma )|=\pi $, and all vector functions $\varphi \in \widetilde H^1(K;{\mathbb C}
^M)\cap {\mathcal H}({\mathcal C}(\frac 12))$ the inequality
$$
\| (\widehat P^{\, +}+a\widehat P^{\, -})\, (\widehat {\mathcal D}(k+i\varkappa e)
+\widehat V - \sum\limits_{j\, =\, 1}^dA^{\{ r\} }_j\widehat \alpha _j\, )\phi \| ^2
\geqslant  \eqno (3.1)
$$ $$
\geqslant (1-\delta )\, (\, C^2_2\, \| \widehat G_-^{\, 1}\widehat P^{\, -}\varphi \| 
^2+a^{\, 2}\, \| \widehat G_+^{\, 1}\widehat P^{\, +}\varphi \| ^2\, )
$$
holds. On the other hand, $A\in L^2(K;{\mathbb R}^d)$ and $\| A-A^{\{ r\} }\| _{L^2(K;
{\mathbb R}^d)}\to 0$ as $r\to +\infty $. Hence, assuming that the vector function
$\varphi $ is a trigonometric polynomial and taking the limit in
(3.1) as $r\to +\infty $, we obtain inequality (2.1). Since trigonometric
polynomials from the set ${\mathcal H}({\mathcal C}(\frac 12))$ are dense in 
$\widetilde H^1(K;{\mathbb C}^M)\cap {\mathcal H}({\mathcal C}(\frac 12))$ (with
respect to the norm of the space $\widetilde H^1(K;{\mathbb C}^M)$) and the magnetic 
potential $A$ obeys condition $\bf (A_1)$ (therefore, $\sum_{j=1}^dA_j\widehat 
\alpha _j\in {\mathbb L}^{\Lambda }_M(d;0)$), we see that inequality (2.1) holds for 
all $\varkappa \geqslant \varkappa _0\, $, all $k\in {\mathbb R}^d$ with $|(k,\gamma 
)|=\pi $, and all vector functions $\varphi \in \widetilde H^1(K;{\mathbb C}^M)\cap 
{\mathcal H}({\mathcal C}(\frac 12))$. Thus, without loss of generality we shall 
assume that the magnetic potential $A:{\mathbb R}^d\to {\mathbb R}^d$ is a 
trigonometric polynomial.

Since the number ${\mathfrak h}>0$ can be chosen arbitrarily small, we shall also
assume that ${\mathfrak h}\leqslant |\gamma |^{-1}\, $.

In this Section, we use the method suggested in \cite{VUU} (also see \cite{TMF00}). 
Lower bounds for the number $\varkappa _0>2R$ are specified in the course of the 
proof. Adding new lower bounds, we assume that all previous bounds are valid as well. 
Let $\delta \in (0,1)$. We write $\delta _1=\frac {\delta}8\, $, $\delta _2=\frac 
{\delta}4\, $. Let us denote
$$
C_3=1+\tau +\frac {|\gamma |}{\pi}\, Q\, .
$$
Suppose a number $\widetilde \varepsilon \in (0,\frac {\delta}8\, ]$ satisfies the
inequality
$$
\widetilde \varepsilon \, C_3^2<\frac1{400}\ \delta ^2\, (1-\widetilde \varepsilon )
\, C_2^2\, .
$$

If $\varkappa \geqslant \varkappa _0>2R$, $k\in {\mathbb R}^d\, $, $N\in
{\mathcal C}(\, \frac 12\, )$ and $N^{\, \prime}\in \Lambda ^*$ with $2\pi |N^{\, 
\prime}_{\perp}|\leqslant R$, then $|k_{\perp}+2\pi (N_{\perp}+N^{\, \prime}
_{\perp})|>\frac {\varkappa}2-R>0$ and
$$
|\, \widetilde e(k+2\pi (N+N^{\, \prime}))-\widetilde e(k+2\pi N)\, |<\frac {2R}
{\varkappa }\ .  \eqno (3.2)
$$

There are numbers $\widetilde c=\widetilde c(\widetilde \varepsilon )>0$ and 
$\varkappa _0^{\, \prime}>(\widetilde c+4)R$ such that for all $\varkappa \geqslant 
\varkappa _0 \geqslant \varkappa _0^{\, \prime}\, $, there are nonintersecting
(nonempty) open (in $S_{d-2}(e)$) sets $\widetilde \Omega _{\lambda}\subset S_{d-2}
(e)$ and vectors $\widetilde e^{\, \lambda}\in \widetilde \Omega _{\lambda}\, $, 
$\lambda \in {\mathcal L}=\{ 1,\dots ,\lambda _0(d,\widetilde \varepsilon ,R;
\varkappa )\} $, such that

1) $|\, \widetilde e-\widetilde e^{\, \lambda}\, |<\widetilde \rho = \widetilde
c\, \frac R{\varkappa}\ $ for all $\widetilde e\in \widetilde \Omega _{\lambda}\, $;

2) $|\, \widetilde e^{\, \prime}-\widetilde e^{\, \prime \prime }\, |> \frac
{8R}{\varkappa}\ $ for all $\widetilde e^{\, \prime}\in \widetilde \Omega _{{\lambda 
}_1}\, $, $\widetilde e^{\, \prime \prime }\in \widetilde \Omega _{{\lambda }_2}\, $,        
${\lambda }_1\neq {\lambda }_2\, $;

3a) ${\rm {meas}}\, S_{d-2}(e)\backslash \bigcup_{\lambda}\widetilde \Omega 
_{\lambda}<\frac 12\ \widetilde \varepsilon \ {\rm {meas}}\, S_{d-2}(e)$, where
${\it {meas}}$ stands for the (invariant) surface measure on the $(d-2)$-dimensional 
sphere $S_{d-2}(e)$.

The sets $\widetilde \Omega _{\lambda}$ and the vectors $\widetilde e^{\, \lambda}$
may be substituted by the sets $\widehat \Theta \, \widetilde \Omega _{\lambda}$ and 
the vectors $\widehat \Theta \, \widetilde e^{\, \lambda}$ (which also have
properties 1, 2, and 3a), where $\widehat \Theta $ is any orthogonal transformation
of the subspace $\{ x\in {\mathbb R}^d:(x,\gamma )=0\} $. Therefore, choosing an
appropriate orthogonal transformation $\widehat \Theta $ of the subspace $\{ x\in 
{\mathbb R}^d:(x,\gamma )=0\} $ and using the notation $\widetilde \Omega _{\lambda}$ 
and $\widetilde e^{\, \lambda}$ instead of $\widehat \Theta \widetilde \Omega 
_{\lambda}$ and $\widehat \Theta \, \widetilde e^{\, \lambda}$ respectively, we
can also assume that for the vector $k\in {\mathbb R}^d$ with $|(k,\gamma )|=\pi $, 
and for the vector function $\varphi \in \widetilde H^1(K;{\mathbb C}^M)\cap {\mathcal 
H}({\mathcal C}(\frac 12))$, the following property is fulfilled (for both signs $+$
and $-$):
\vskip 0.1cm

3b)
$$
\sum\limits_{N\, \in \, {\mathcal C}\, (\frac 12)\, :\ \widetilde e(k+2\pi N)\, 
\not\in \, \bigcup\limits_{\lambda}\, \widetilde \Omega _{\lambda}}\| \, \widehat 
G^{\pm}_N(k;\varkappa )\, \widehat P^{\pm}_{\widetilde e(k+2\pi N)}\, \varphi 
_N\, \| ^2\leqslant \widetilde \varepsilon \, v^{-1}(K)\, \| \, \widehat G_{\pm}^{\,
1}\widehat P^{\pm}\varphi \, \| ^2\, .
$$ 

The sets $\widetilde \Omega _{\lambda}$ and the vectors $\widetilde e^{\, \lambda}$
depend on $d$, $\gamma $, $\widetilde \varepsilon $, $R$, $\varkappa $, and also on
the chosen vector $k\in {\mathbb R}^d$ with $|(k,\gamma )|=\pi $, and on the vector
function $\varphi \in \widetilde H^1(K;{\mathbb C}^M)\cap {\mathcal H}({\mathcal C}
(\frac 12))$.
 
We introduce the notation $\rho =\widetilde \rho +\frac {2R}{\varkappa}\ $, $\rho 
^{\, \prime}=\widetilde \rho +\frac {4R}{\varkappa}\ $. Let
$$
\Omega _{\lambda}=\{ \widetilde e\in S_{n-2}(e): |\,  \widetilde e-\widetilde e
^{\, \prime }\, |< \frac {2R}{\varkappa}\ \ \text{for\ some} \ \ \widetilde 
e^{\, \prime }\in \widetilde \Omega _{\lambda}\} \, ,
$$ $$
\Omega _{\lambda}^{\, \prime}=\{ \widetilde e\in S_{n-2}(e): |\,  \widetilde e-
\widetilde e^{\, \prime }\, |< \frac {4R}{\varkappa}\ \ \text{for\ some} \
\ \widetilde e^{\, \prime }\in \widetilde \Omega _{\lambda}\} \, ;
$$
$\widetilde \Omega _{\lambda}\subset \Omega _{\lambda}\subset \Omega _{\lambda}^{\, 
\prime}\, $. The sets $\Omega _{\lambda}^{\, \prime}$ do not intersect for
different $\lambda \in {\mathcal L}$. Moreover, $|\, \widetilde e^{\, \prime}-
\widetilde e^{\, \prime \prime }\, |> \frac {4R}{\varkappa}\ $ for all 
$\widetilde e^{\, \prime}\in \Omega _{{\lambda }_1}\, $, $\widetilde e^{\, \prime 
\prime }\in \Omega _{{\lambda }_2}\, $, ${\lambda }_1\neq {\lambda }_2\, $.

We write
$$
\widetilde {\mathcal K}_{\lambda}=\widetilde {\mathcal K}_{\lambda}(k,\varkappa ;
\varphi )=\{ N\in {\mathcal C}(\, \frac 12\, ): \widetilde e(k+2\pi N)\in 
\widetilde \Omega _{\lambda}\} \, ,
$$ $$
{\mathcal K}_{\lambda}={\mathcal K}_{\lambda}(k,\varkappa ;
\varphi )=\{ N\in {\mathcal C}(\, \frac 12\, ): \widetilde e(k+2\pi N)\in 
\Omega _{\lambda}\} \, ,
$$ $$
{\mathcal K}_{\lambda}^{\, \prime}={\mathcal K}_{\lambda}^{\, \prime}(k,\varkappa ;
\varphi )=\{ N\in {\mathcal C}(\, \frac 12\, ): \widetilde e(k+2\pi N)\in 
\Omega _{\lambda}^{\, \prime}\} \, .
$$
Property 3b implies that for the vector function $\varphi \in \widetilde H^1(K;
{\mathbb C}^M)\cap {\mathcal H}({\mathcal C}(\frac 12)$, we have (for each 
sign)
$$
\| \, \widehat G_{\pm}^{\, 1}\widehat P^{\pm}\varphi ^{\, {\mathcal C}(\frac 12)
\, \backslash \, \bigcup\limits_{\lambda} \, \widetilde {\mathcal K}_{\lambda}}\, 
\| ^2\leqslant \widetilde \varepsilon \, \| \, \widehat G_{\pm}^{\, 1}\widehat 
P^{\pm}\varphi \, \| ^2\, .  \eqno (3.3)
$$ 

Without loss of generality we assume that ${\mathcal E}_2=e$. For each index 
$\lambda \in {\mathcal L}$ (and for already chosen $k$, $\varkappa $, and
$\varphi $), we take an orthogonal system of vectors ${\mathcal E}_j^{(\lambda )}
\in S_{d-1}\subset {\mathbb R}^d$, $j=1,\dots ,d$, such that ${\mathcal E}_1
^{(\lambda )}=\widetilde e^{\, \lambda }$, ${\mathcal E}_2^{(\lambda )}={\mathcal E}
_2=e$. We let $x_j^{(\lambda )}=(x,{\mathcal E}_j^{(\lambda )})$ denote the 
coordinates of the vectors $x=\sum_{j=1}^dx_j{\mathcal E}_j\in {\mathbb R}^d$ ($k_j
^{(\lambda )}$, $N_j^{(\lambda )}$, $A_j^{(\lambda )}$, and $\widetilde A_j^{(\lambda 
)}$ are the coordinates of the vectors $k\in {\mathbb R}^d$, $N\in \Lambda ^*$, $A$, 
and $\widetilde A(\widetilde e^{\, \lambda };.)$). Let ${\mathcal E}_j^{(\lambda )}
=\sum_{l=1}^dT^{(\lambda )}_{lj}{\mathcal E}_l\, $, $j=1,\dots ,d$. Then $A_j
^{(\lambda )}=\sum_{l=1}^dT^{(\lambda )}_{lj}A_l$ and $\widetilde 
A_j^{(\lambda )}(.)=\sum_{l=1}^dT^{(\lambda )}_{lj}\widetilde A_l(\widetilde 
e^{\, \lambda };.)$. We introduce the notation $\widehat \alpha _j^{(\lambda )}=
\sum_{l=1}^dT^{(\lambda )}_{lj}\widehat \alpha _l\, $, $j=1,\dots ,d$ (the matrices 
$\widehat \alpha _j^{(\lambda )}\in {\mathcal S}_M$ satisfy the same commutation
relations as the matrices $\widehat \alpha _j$). The following equality is valid:
$$
\widehat {\mathcal D}(k+i\varkappa e)-\sum\limits_{j\, =\, 1}^dA_j\widehat \alpha _j=
\sum\limits_{j\, =\, 1}^d(k_j^{(\lambda )}+i\varkappa e_j^{(\lambda )}-i\, \frac 
{\partial}{\partial x_j^{(\lambda )}}-A_j^{(\lambda )})\, \widehat \alpha _j
^{(\lambda )}\, ,
$$
where $e_j^{(\lambda )}=1$ for $j=2$ and $e_j^{(\lambda )}=0$ for $j\neq 2$.

For the Fourier coefficients $(\widetilde A_j^{(\lambda )})_N$ of the functions 
$\widetilde A_j^{(\lambda )}$, $j=1,\dots ,d$, we have $(\widetilde A_j^{(\lambda 
)})_N=\widehat \mu (2\pi N_1^{(\lambda )})(A_j^{(\lambda )})_N\, $ if $N_2^{(\lambda 
)}=N_2=0$ and $(\widetilde A_j^{(\lambda )})_N=0$ if $N_2^{(\lambda )}\neq 0$. (Here,
$(A_j^{(\lambda )})_N$ are the Fourier coefficients of $A_j^{(\lambda )}$, $N\in
\Lambda ^*$).

For $s=1,2$ and $\lambda \in {\mathcal L}$, let $\Phi ^{(s,\lambda )}:{\mathbb
R}^d\to {\mathbb R}$ be the $\Lambda $-periodic trigonometric polynomials with the 
Fourier coefficients $\Phi ^{(1,\lambda )}_N=\Phi ^{(2,\lambda )}_N=0$ if $N_1
^{(\lambda )}=N_2^{(\lambda )}=0$ and
$$
\Phi ^{(1,\lambda )}_N=\bigl( \, 2\pi i\, ((N_1^{(\lambda )})^2+(N_2^{(\lambda )})^2
)\, \bigr) ^{-1}\, (\, N_1^{(\lambda )}(A_1^{(\lambda )}-\widetilde A_1^{(\lambda )})
_N+N_2^{(\lambda )}(A_2^{(\lambda )}-\widetilde A_2^{(\lambda )})_N\, )\, ,
$$ $$
\Phi ^{(2,\lambda )}_N=-\bigl( \, 2\pi i\, ((N_1^{(\lambda )})^2+(N_2^{(\lambda )})^2
)\, \bigr) ^{-1}\, (\, N_2^{(\lambda )}(A_1^{(\lambda )}-\widetilde A_1^{(\lambda )})
_N-N_1^{(\lambda )}(A_2^{(\lambda )}-\widetilde A_2^{(\lambda )})_N\, )
$$
otherwise.

\begin{lemma} \label{l3.1}
There is a constant $C^*({\mathfrak h})>0$ such that 
$$
\| \Phi ^{(s,\lambda )}\| _{L^{\infty}({\mathbb R}^d)}\leqslant \frac 14\ \| \mu
\| \, C^*({\mathfrak h})\, ,\ \, s=1,2\, ,\ \, \lambda \in {\mathcal L}\, .
$$
\end{lemma}

\begin{proof}
Let $\eta (.)\in C^{\infty}({\mathbb R};{\mathbb R})$, $\eta (\xi )=0$ for $\xi 
\leqslant \pi $, $0\leqslant \eta (\xi )\leqslant 1$ for $\pi <\xi \leqslant 
2\pi $, and $\eta (\xi )=1$ for $\xi >2\pi $. For $\xi _1\, ,\xi _2 \in {\mathbb 
R}$ (and $\xi _1^2+\xi _2^2>0$), we set
$$
G(\xi _1,\xi _2)=\frac {\xi _1}{\xi _1^2+\xi _2^2}\ \int_0^{+\infty}
\frac {\partial \eta (\xi )}{\partial \xi }\, J_0\, (\xi \sqrt {\xi _1^2+\xi _2^2}
\, )\, d\xi \, ,
$$
where $J_0(.)$ is the Bessel function of the first kind of order zero. The estimate
$$
|G(\xi _1,\xi _2)|\leqslant \frac 1{\sqrt {\xi _1^2+\xi _2^2}}  \eqno (3.4)
$$
holds. The choice of the function $\eta (.)$ implies that
$$
|G(\xi _1,\xi _2)|\cdot (\xi _1^2+\xi _2^2)^{\beta}\to 0  \eqno (3.5)
$$
as $\xi _1^2+\xi _2^2\to +\infty $ for all $\beta \geqslant 0$ (whence $G(.,.)\in 
L^q({\mathbb R}^2)$, $q\in [1,2)$). We write
$$
G_1(t;\xi _1,\xi _2)=t^{-1}G(t^{-1}\xi _1\, ,t^{-1}\xi _2)\, ,\ t>0\, ,
$$
and $G_2(t;\xi _1,\xi _2)\doteq G_1(t;\xi _2,\xi _1)$. For an arbitrary continuous
$\Lambda $-periodic function $F:{\mathbb R}^d\to {\mathbb R}$, we set
$$
(F*_{\lambda}G_s(t;.,.))(x)=
\iint_{{\mathbb R}^2}G_s(t;\xi _1,\xi _2)\, F(x-\xi _1\widetilde e
^{\, \lambda }-\xi _2e)\, d\xi _1d\xi _2\, ,\ x\in {\mathbb R}^d\, ,\
s=1,2\, .
$$
In this case, $(F*_{\lambda}G_s(t;.,.))_N=0$ if $N_1^{(\lambda )}=N_2^{(\lambda )}
=0$ and
$$
(F*_{\lambda}G_s(t;.,.))_N=-\, \frac {iN^{(\lambda )}_s}{(N^{(\lambda )}_1)^2+
(N^{(\lambda )}_2)^2}\ \eta \, \bigl( 2\pi t\sqrt {(N^{(\lambda )}_1)^2+(N^{(\lambda 
)}_2)^2}\ \bigr) \, F_N
$$
otherwise (here $N\in \Lambda ^*$). If $N_2^{(\lambda )}=0$ and $|N_1^{(\lambda )}|
\leqslant {\mathfrak h}$, then $(A(.)-\widetilde A(\widetilde e^{\, \lambda };.))
_N=0$. On the other hand, if $N_2^{(\lambda )}\neq 0$, then $|N_2^{(\lambda )}|=
|\gamma |^{-1}|(N,\gamma )|\geqslant |\gamma |^{-1}\geqslant {\mathfrak h}$. 
Therefore,
$$
2\pi \, \Phi ^{(1,\lambda )}=(A_1^{(\lambda )}-\widetilde A_1^{(\lambda )})
*_{\lambda }G_1({\mathfrak h}^{-1};.,.)+(A_2^{(\lambda )}-\widetilde A_2^{(\lambda )})
*_{\lambda }G_2({\mathfrak h}^{-1};.,.)\, ,  \eqno (3.6)
$$ $$
2\pi \, \Phi ^{(2,\lambda )}= -\, (A_1^{(\lambda )}-\widetilde A_1^{(\lambda )})
*_{\lambda }G_2({\mathfrak h}^{-1};.,.)+(A_2^{(\lambda )}-\widetilde A_2^{(\lambda )})
*_{\lambda }G_1({\mathfrak h}^{-1};.,.)\, .  \eqno (3.7)
$$
Estimate (3.4) yields
$$
|G_s({\mathfrak h}^{-1};\xi _1,\xi _2)|\leqslant \frac 1{\sqrt {\xi _1^2+\xi _2^2}}\ \, ,\
s=1,2\, .  
$$
Hence, from (0.12), for all $x\in {\mathbb R}^d$, we obtain
$$
\iint_{\xi ^2_1+\xi ^2_2\, \leqslant \, 1} |G_s({\mathfrak h}^{-1};\xi _1,\xi _2)|
\cdot |A(x-\xi _1\widetilde e^{\, \lambda }-\xi _2e)|\, d\xi _1d\xi _2
\leqslant C^*\, .  
$$
The last inequality and (3.5) imply that there exists a constant $C^*({\mathfrak h})
>0$ (dependent on ${\mathfrak h}$ and the constant $C^*$) such that for all $x\in 
{\mathbb R}^d$ (and all $s\in \{ 1,2\} $, $\lambda \in {\mathcal L}$), we have
$$
\iint_{{\mathbb R}^2} |G_s({\mathfrak h}^{-1};\xi _1,\xi _2)|
\cdot |A(x-\xi _1\widetilde e^{\, \lambda }-\xi _2e)|\, d\xi _1d\xi _2
\leqslant \frac {\pi}8\ C^*({\mathfrak h})\, .
$$
Consequently, we also have
$$
\iint_{{\mathbb R}^2} |G_s({\mathfrak h}^{-1};\xi _1,\xi _2)| \cdot |\widetilde A(
\widetilde e^{\, \lambda };x-\xi _1\widetilde e^{\, \lambda }-\xi _2e)|\, 
d\xi _1d\xi _2\leqslant \frac {\pi}8\ C^*({\mathfrak h})\, .  
$$
Finally, using (3.6), (3.7), and the inequality $\| \mu \| \geqslant \widehat \mu (0)
=1$, for $s=1,2$ and $\lambda \in {\mathcal L}$, we derive the claimed estimate.
\end{proof}

Let us denote 
$$
\widehat {\mathcal D}^{(\lambda )}_0=\bigl( k^{(\lambda )}_1-i\, \frac {\partial}
{\partial x^{(\lambda )}_1}\, \bigr) \, \widehat \alpha ^{(\lambda )}_1+
\bigl( k^{(\lambda )}_2+i\varkappa -i\, \frac {\partial}{\partial x^{(\lambda )}_2}
\, \bigr) \, \widehat \alpha ^{(\lambda )}_2\, ,
$$ $$
\widehat {\mathcal D}^{(\lambda )}=\widehat {\mathcal D}^{(\lambda )}_0-A^{(\lambda )}_1
\widehat \alpha ^{(\lambda )}_1-A^{(\lambda )}_2\widehat \alpha ^{(\lambda )}_2\, ,
$$ $$
\widehat {\mathcal D}^{(\lambda )}_{\perp}=\sum\limits_{j=3}^n\, (k^{(\lambda )}_j-i\, 
\frac {\partial}{\partial x^{(\lambda )}_j}-A^{(\lambda )}_j)\, \widehat \alpha ^
{(\lambda )}_j+\widehat V\, .
$$
We have
$$
\widehat {\mathcal D}(k+i\varkappa e)+\widehat W=\widehat {\mathcal D}^{(\lambda )}+
\widehat {\mathcal D}^{(\lambda )}_{\perp}\, .
$$
We also introduce the notation
$$
\widehat {\mathcal D}^{(\lambda )}_*=\widehat {\mathcal D}^{(\lambda )}_0-\widetilde
A^{(\lambda )}_1\widehat \alpha ^{(\lambda )}_1-\widetilde A^{(\lambda )}_2\widehat 
\alpha ^{(\lambda )}_2\, .
$$
Since
$$
\frac {\partial \Phi ^{\, (1,\lambda )}}{\partial x^{(\lambda )}_1}-
\frac {\partial \Phi ^{\, (2,\lambda )}}{\partial x^{(\lambda )}_2}=A_1^{(\lambda 
)}-\widetilde A_1^{(\lambda )}\, ,\ \
\frac {\partial \Phi ^{\, (1,\lambda )}}{\partial x^{(\lambda )}_2}+
\frac {\partial \Phi ^{\, (2,\lambda )}}{\partial x^{(\lambda )}_1}=A_2^{(\lambda 
)}-\widetilde A_2^{(\lambda )}\, ,
$$
the following identity is true:
$$
\widehat {\mathcal D}^{(\lambda )}=e^{-i\, \widehat \alpha ^{(\lambda )}_1
\widehat \alpha ^{(\lambda )}_2\, \Phi ^{\, (2,\lambda )}}\, e^{\, i\, \Phi ^{\, 
(1,\lambda )}}\, \widehat {\mathcal D}_*^{(\lambda )}\, e^{-i\, \Phi ^{\, (1,
\lambda )}}\, e^{-i\, \widehat \alpha ^{(\lambda )}_1\widehat \alpha 
^{(\lambda )}_2\, \Phi ^{\, (2,\lambda )}}\, .  \eqno (3.8)
$$
We shall use the brief notation
$$
\widehat P^{\, \pm }_{\lambda }\doteq \widehat P^{\, \pm }_{\widetilde e^{\, 
\lambda }}=\frac 12\ (\widehat I\pm i\widehat \alpha ^{(\lambda )}_1\widehat 
\alpha ^{(\lambda )}_2)\, .
$$ 
The relations
$$
\widehat {\mathcal D}^{(\lambda )}_{\perp}\widehat P^{\, \pm }_{\lambda }=
\widehat P^{\, \pm }_{\lambda }\, \widehat {\mathcal D}^{(\lambda )}_{\perp}\, 
,\ \ \widehat {\mathcal D}^{(\lambda )}\widehat P^{\, \pm }_{\lambda }=\widehat 
P^{\, \mp }_{\lambda }\, \widehat {\mathcal D}^{(\lambda )}
$$
hold (these relations are important for the sequel). For all vector functions
$\psi \in \widetilde H^1(K;{\mathbb C}^M)$, we get
$$
\widehat {\mathcal D}^{(\lambda )}_0\widehat P^{\, \pm }_{\lambda }\, 
\psi = 
\sum\limits_{N\, \in \, \Lambda ^*}\bigl( (k^{(\lambda )}_2+2\pi N^{(\lambda )}_2)
+i(\varkappa \pm (k^{(\lambda )}_1+2\pi N^{(\lambda )}_1))\bigr) \, \widehat \alpha 
^{(\lambda )}_2\widehat P^{\, \pm }_{\lambda }\, \psi _N\, e^{\, 2\pi i\, (N,x)}\, .
\eqno (3.9)
$$
Let us introduce the operators (acting on the space $L^2(K;{\mathbb C}^M)$) 
$$
\widehat G_{\lambda ,\, \pm }^{\, 1}\, \psi = 
\sum\limits_{N\, \in \, \Lambda ^*}\bigl( (k^{(\lambda )}_2+2\pi N^{(\lambda )}
_2)^2+(\varkappa \pm |k^{(\lambda )}_1+2\pi N^{(\lambda )}_1|)^2\bigr) ^{1/2}\, 
\psi _N\, e^{\, 2\pi i\, (N,x)}\, ,  \eqno (3.10)
$$
$\psi \in D(\widehat G_{\lambda ,\, \pm }^{\, 1})=\widetilde H^1(K;{\mathbb C}^M)
\subset L^2(K;{\mathbb C}^M)$, $\lambda \in {\mathcal L}$. Since $\varkappa 
\geqslant \varkappa _0>(\widetilde c+4)R$, we see that for all $N\in {\mathcal K}
^{\, \prime}_{\lambda}$ the condition
$$
|\, \widetilde e(k+2\pi N)-\widetilde e^{\, \lambda }\, | <\rho ^{\, \prime}=
\frac {(\widetilde c+4)R}{\varkappa }<1  \eqno (3.11)
$$
is fulfilled. Hence, $k^{(\lambda )}_1+2\pi N^{(\lambda )}_1>0$ and from (3.9),
(3.10), for all vector functions $\psi \in \widetilde H^1(K;{\mathbb C}^M)\cap
{\mathcal H}({\mathcal K}^{\, \prime}_{\lambda})$, we derive
$$
\| \, \widehat {\mathcal D}^{(\lambda )}_0\widehat P^{\, \pm }_{\lambda }\, \psi \, \|
= \| \, \widehat G_{\lambda ,\, \pm }^{\, 1}\widehat P^{\, \pm }_{\lambda }\, \psi 
\, \| \, .  \eqno (3.12)
$$

Denote
$$
b_1=\frac 12\, (\widetilde c+2)R+\frac 34\, \frac {|\gamma |}{\pi}\ (\widetilde 
c+2)^2R^2\, ,\ \
b_2=\frac 12\, (\widetilde c+4)R+\frac 34\, \frac {|\gamma |}{\pi}\ (\widetilde 
c+4)^2R^2\, .
$$

\begin{lemma} \label{l3.2}
Given a number $\varkappa \geqslant \varkappa _0\, $, a vector $k\in {\mathbb R}^d$,
and a vector function $\varphi \in \widetilde H^1(K;{\mathbb C}^M)\cap {\mathcal H}
({\mathcal C}(\frac 12))$, the estimates
$$
\| (\widehat G_{\pm }^{\, 1}\widehat P^{\, \pm }-\widehat G_{\lambda ,\, \pm }^{\, 1}
\widehat P^{\, \pm }_{\lambda })\, \psi \, \| \leqslant
\frac {b_1}{\varkappa }\ \| \, \widehat G_{\pm }^{\, 1}\psi \, \| \, ,\ \, \psi \in 
{\mathcal H}({\mathcal K}_{\lambda})\, ,  \eqno (3.13)
$$ $$
\| (\widehat G_{\pm }^{\, 1}\widehat P^{\, \pm }-\widehat G_{\lambda ,\, \pm }^{\, 1}
\widehat P^{\, \pm }_{\lambda })\, \psi ^{\, \prime}\, \| \leqslant \frac {b_2}
{\varkappa }\ \| \, \widehat G_{\pm }^{\, 1}\psi ^{\, \prime}\, \| \, ,\ \, \psi ^{\, 
\prime}\in {\mathcal H}({\mathcal K}^{\, \prime}_{\lambda})\, ,  \eqno (3.14)
$$
hold for all $\lambda \in {\mathcal L}$.
\end{lemma}

\begin{proof}
For vectors $N\in {\mathcal K}_{\lambda}$, we have
$$
|\, \widetilde e(k+2\pi N)-\widetilde e^{\, \lambda }\, | <\rho =
\frac {(\widetilde c+2)R}{\varkappa }<1\, . 
$$
Hence, from (1.4), we obtain
$$
\| \, (\widehat P^{\, \pm }-\widehat P^{\, \pm }_{\lambda })\, \psi \, \| \leqslant \,
\rho \, \| \psi \| \, ,\ \psi \in {\mathcal H}({\mathcal K}_{\lambda})\, .  
\eqno (3.15)
$$
The estimate 
$$
\| \, (\widehat P^{\, \pm }-\widehat P^{\, \pm }_{\lambda })\, \psi ^{\, \prime}\, \| 
\leqslant \, \rho ^{\, \prime}\, \| \psi ^{\, \prime}\| \, ,\ \psi ^{\, \prime}
\in {\mathcal H}({\mathcal K}^{\, \prime}_{\lambda})  \eqno (3.16)
$$
follows from (3.11) and (1.4).

If $\widetilde e\in \Omega _{\lambda}\, $, then $|\widetilde e-\widetilde e
^{\, \lambda }|<\rho $ and $1-(\widetilde e,\widetilde e^{\, \lambda })=\frac 12\, 
|\widetilde e-\widetilde e^{\, \lambda }|<\frac 12\, \rho ^2$. Therefore, for all 
vectors $N\in {\mathcal K}_{\lambda}$ (for which $\widetilde e(k+2\pi N)\in \Omega 
_{\lambda}$), we get
$$
\bigl| \, |\varkappa \pm |k_{\perp}+2\pi N_{\perp}|\, |-|\varkappa \pm (k_1
+2\pi N_1)|\, \bigr| \leqslant |k_{\perp}+2\pi N_{\perp}|-
(k_1+2\pi N_1)=
$$ $$
=\, (1-(\widetilde e(k+2\pi N),\widetilde e^{\, \lambda }))\, |k_{\perp}+2\pi 
N_{\perp}|< \frac 12\ \rho ^2\, |k_{\perp}+2\pi N_{\perp}|<\frac 34\, \rho ^2
\varkappa \, .
$$
Whence
$$
\| \, (\widehat G_{\pm }^{\, 1}-\widehat G_{\lambda ,\, \pm }^{\, 1})\, \psi \| 
\leqslant \frac 34\, \rho ^2\varkappa \, \| \psi \| = \frac 34\, (\widetilde c+2)
^2\, \frac {R^2}{\varkappa }\ \| \psi \| \, ,\ \  
\psi \in \widetilde H^1(K;{\mathbb C}^M)\cap {\mathcal H}({\mathcal K}_{\lambda})
\, .  \eqno (3.17)
$$
Analogously, it follows that
$$
\| \, (\widehat G_{\pm }^{\, 1}-\widehat G_{\lambda ,\, \pm }^{\, 1})\, \psi ^{\, 
\prime}\| \leqslant \frac 34\, (\widetilde c+4)^2\, \frac {R^2}{\varkappa }\ \| \psi 
^{\, \prime}\| \, ,\ \psi ^{\, \prime}\in \widetilde H^1(K;{\mathbb C}^M)\cap 
{\mathcal H}({\mathcal K}^{\, \prime }_{\lambda})\, .  \eqno (3.18)
$$
Now, from (3.15) and (3.17), for the vector function $\psi \in \widetilde H^1(K;
{\mathbb C}^M)\cap {\mathcal H}({\mathcal K}_{\lambda})$), we obtain the inequality
(3.13):
$$
\| (\widehat G_{\pm }^{\, 1}\widehat P^{\, \pm }-\widehat G_{\lambda ,\, \pm }^{\, 1}
\widehat P^{\, \pm }_{\lambda })\, \psi \, \| \leqslant
\| (\widehat P^{\, \pm }-\widehat P^{\, \pm }_{\lambda })\, \widehat 
G_{\pm }^{\, 1}\, \psi \, \| +\| (\widehat G_{\pm }^{\, 1}-\widehat G_{\lambda ,\, 
\pm }^{\, 1})\widehat P^{\, \pm }_{\lambda }\, \psi \, \| \leqslant
$$ $$
\leqslant \frac 12\, \rho \, \| \, \widehat G_{\pm }^{\, 1}\, \psi \, \| +\frac 34\, 
\rho ^2\varkappa \, \| \, \widehat P^{\, \pm }_{\lambda }\, \psi \, \| \leqslant
$$ $$
\leqslant \frac 12\ \frac {(\widetilde c+2)R}{\varkappa}\ \| \, \widehat G_{\pm }
^{\, 1}\, \psi \, \| +\frac 34\ \frac {(\widetilde c+2)^2R^2}{\varkappa}\ \| \, \psi 
\, \| \leqslant \frac {b_1}{\varkappa}\ \| \, \widehat G_{\pm }^{\, 1}\, \psi 
\, \| \, .
$$
Inequality (3.14) is proved similarly (using estimates (3.16) and (3.18)).
\end{proof}

In the sequel, we shall use the notation
$$
\widehat P^{\, \pm }_{\lambda }\widehat P^{\, {\mathcal K}_{\lambda}}\, \varphi =
\widehat P^{\, \pm }_{\lambda }\, \varphi ^{\, {\mathcal K}_{\lambda}}=\varphi
^{\, \pm}_{\lambda}\, ,\ \, \widehat P^{\, \pm }_{\lambda }\widehat P^{\, \widetilde 
{\mathcal K}_{\lambda}}\, \varphi =\widehat P^{\, \pm }_{\lambda }\, \varphi ^{\, 
\widetilde {\mathcal K}_{\lambda}}=\widetilde \varphi ^{\, \pm}_{\lambda}\, .
$$

On the space $L^2(K;{\mathbb C}^M)$, we define the orthogonal projections $\widehat 
P^{\, \{ \Omega _{\lambda}\} }$, $\lambda \in {\mathcal L}$, that take a vector
function $\psi \in L^2(K;{\mathbb C}^M)$ to the vector function $\widehat P^{\, \{ 
\Omega _{\lambda}\} }\, \psi $ for which $(\widehat P^{\, \{ \Omega _{\lambda}\} }\, 
\psi )_N=\psi _N\, $ if $\widetilde e(k+2\pi N)\in \Omega _{\lambda}\, $ and 
$(\widehat P^{\, \{ \Omega _{\lambda}\} }\, \psi )_N=0$ if either the vector 
$\widetilde e(k+2\pi N)$ is not defined (for $k_{\perp}+2\pi N_{\perp}=0$) or
$\widetilde e(k+2\pi N)\not\in \Omega _{\lambda}\, $. By analogy with the orthogonal 
projections $\widehat P^{\, \{ \Omega _{\lambda}\} }$, define the orthogonal 
projections $\widehat P^{\, \{ \widetilde \Omega _{\lambda}\} }$ (replacing the
sets $\Omega _{\lambda}\subset S_{n-2}(e)$ by the sets $\widetilde \Omega 
_{\lambda}$).

For all vector functions $\psi \in L^2(K;{\mathbb C}^M)$, we have
$$
\| \, (\widehat P^{\, \pm }-\widehat P^{\, \pm }_{\lambda})\, \widehat P^{\, \{ \Omega 
_{\lambda}\} }\, \psi \, \| \leqslant \frac 12\, \rho \ \| \, \widehat P^{\, \{ 
\Omega _{\lambda}\} }\, \psi \, \| \, .  \eqno (3.19)
$$

Since $\varkappa \geqslant \varkappa _0>4R$ (and $\varphi \in \widetilde H^1(K;
{\mathbb C}^M)\cap {\mathcal H}({\mathcal C}(\frac 12))$), we see that in the case, 
where $(\widehat {\mathcal D}(k+i\varkappa e)+\widehat W)_N\neq 0$, $N\in \Lambda 
^*$, the estimate
$$
|\, k_{\perp}+2\pi N_{\perp}|>\frac {\varkappa}2-R>\frac {\varkappa}4
$$
holds.

\begin{lemma} \label{l3.3}
If $N\in {\mathcal K}_{\lambda}\, $, then
$$
\bigl| \, \sum_{j=3}^d\, (k_j^{\, (\lambda )}+2\pi N_j^{\, (\lambda )})\,
{\mathcal E}_j^{\, (\lambda )}\, \bigr| <\frac 32\, (\widetilde c+2)R\, .
\eqno (3.20)
$$
\end{lemma}

\begin{proof}
Indeed,
$$
\bigl| \, \sum_{j=3}^d\, (k_j^{\, (\lambda )}+2\pi N_j^{\, (\lambda )})\,
{\mathcal E}_j^{\, (\lambda )}\, \bigr| =|\, k_{\perp}+2\pi N_{\perp}-(k_{\perp}
+2\pi N_{\perp}\, ,\widetilde e^{\, \lambda })\, \widetilde e^{\, \lambda }\, |
\leqslant  \eqno (3.21)
$$ $$
\leqslant |\, \widetilde e(k+2\pi N)-\widetilde e^{\, \lambda }\, | \cdot
|\, k_{\perp}+2\pi N_{\perp}|<\rho \, |\, k_{\perp}+2\pi N_{\perp}|=\frac
{(\widetilde c+2)R}{\varkappa}\ |\, k_{\perp}+2\pi N_{\perp}|\, .
$$
At the same time, the definition of the set ${\mathcal C}(\frac 12)$ implies that
$|\, k_{\perp}+2\pi N_{\perp}|<\frac 32\, \varkappa $. Therefore inequality (3.20)
follows from (3.21).
\end{proof}

From estimate (1.2) (under the change $k\to k-\varkappa \widetilde e^{\, \lambda}$) 
and Lemma \ref{l3.3}, for all $\varepsilon >0$, we obtain
$$
\| \, \widehat W \varphi ^{\, \pm }_{\lambda}\, \| \leqslant \| \, \widehat V 
\varphi ^{\, \pm }_{\lambda}\, \| + \| \, |A|\, \varphi ^{\, \pm }_{\lambda}\, \|
\leqslant  \eqno (3.22)
$$ $$
\leqslant 
\varepsilon \, \bigl\| \, \sum\limits_{j\, =\, 1}^d\, (k_j-\varkappa \widetilde e^{\, 
\lambda}_j-i\, \frac {\partial}{\partial x_j}\, )\, \widehat \alpha _j\varphi ^{\, 
\pm }_{\lambda}\, \bigr\| +C_{\varepsilon }^{\, \prime }(0,\widehat W)\, \| \, 
\varphi ^{\, \pm }_{\lambda}\, \| \leqslant  
$$ $$
\leqslant \varepsilon \, \bigl\| \, ((k_1^{\, (\lambda )}-\varkappa -i\, \frac 
{\partial}{\partial x_1^{(\lambda )}}\, )\, \widehat \alpha _1^{(\lambda )}+
(k_2^{\, (\lambda )}-i\, \frac {\partial}{\partial x_2^{(\lambda )}}\, )\, \widehat 
\alpha _2^{(\lambda )}\, )\, \varphi ^{\, \pm }_{\lambda}\, \bigr\| \, +
$$ $$
+\, \varepsilon \, \bigl\| \, \sum_{j\, =\, 3}^d\, (k_j^{\, (\lambda )}-i\, \frac 
{\partial}{\partial x_j^{\, (\lambda )}}\, )\, \widehat \alpha _j^{\, (\lambda )}\varphi 
^{\, \pm }_{\lambda}\, \bigr\| +C_{\varepsilon }^{\, \prime }(0,\widehat W)\, \| 
\, \varphi ^{\, \pm }_{\lambda}\, \| \leqslant 
$$ $$
\leqslant \varepsilon \, \| \, \widehat G^{\, 1}_{\lambda ,\, -}\varphi ^{\, \pm }
_{\lambda}\, \| +\bigl( \, \frac {3\varepsilon }2\, (\widetilde c+2)R+
C_{\varepsilon }^{\, \prime }(0,\widehat W)\bigr) \, \| \, \varphi ^{\, \pm }
_{\lambda}\, \| \, .
$$
Whence
$$
\| \, \widehat {\mathcal D}^{\, (\lambda )}_{\perp}\, \varphi ^{\, \pm }_{\lambda}\, \|
\leqslant \| \, \widehat W \varphi ^{\, \pm }_{\lambda}\, \| +\bigl\| \, \sum_{j\, 
=\, 3}^d\, (k_j^{\, (\lambda )}-i\, \frac {\partial}{\partial x_j^{(\lambda )}}\, 
)\, \widehat \alpha _j^{(\lambda )}\varphi ^{\, \pm }_{\lambda}\, \bigr\|
\leqslant  \eqno (3.23)
$$ $$
\leqslant \varepsilon \, \| \, \widehat G^{\, 1}_{\lambda ,\, -}\varphi ^{\, \pm }
_{\lambda}\, \| +\bigl( \, \frac {3(\varepsilon +1)}2\, (\widetilde c+2)R+
C_{\varepsilon }^{\, \prime }(0,\widehat W) \bigr) \, \| \, \varphi ^{\, \pm }
_{\lambda}\, \| \, .
$$

Given $N\in {\mathcal K}_{\lambda}^{\, \prime }\, $, the following inequality is
proved in the same way as inequality (3.20) (see the proof of Lemma \ref{l3.3}): 
$$
\bigl| \, \sum_{j\, =\, 3}^d\, (k_j^{\, (\lambda )}+2\pi N_j^{\, (\lambda )})\,
{\mathcal E}_j^{\, (\lambda )}\, \bigr| <\frac 32\, (\widetilde c+4)R\, .
$$
Therefore, by analogy with inequality (3.23), for all $\varepsilon >0$ and all
vector functions $\psi \in \widetilde H^1(K;{\mathbb C}^M)\cap {\mathcal H}
({\mathcal K}^{\, \prime }_{\lambda})$, we also derive
$$
\| \, \widehat {\mathcal D}^{\, (\lambda )}_{\perp}\, \psi \, \| \leqslant 
\varepsilon \, \| \, \widehat G^{\, 1}_{\lambda ,\, -}\psi \, \| +C^{\, \sharp }_3\,
(\varepsilon )\, \| \psi \| \, ,  \eqno (3.24)
$$
where $C^{\, \sharp }_3\, (\varepsilon )=\frac {3(\varepsilon +1)}2\, (\widetilde 
c+4)R+C_{\varepsilon }^{\, \prime }(0,\widehat W)$.

We write 
$$
C^{\, \prime}_3=3(\widetilde c+2)R+C_1^{\, \prime }(0,\widehat W)\, ,\ \, C_4=1+\frac {|\gamma |}{\pi}\
C^{\, \prime}_3\, .
$$
Choose a number $a\in (0,1)$ such that
$$
a^2\, \max \, \{ \, C_4^2\, ,\frac 94\, \frac {|\gamma |^2}{\pi ^2}\ b_2^2\, \}
< \frac {\delta \delta _1}{50}\ (1-\widetilde \varepsilon )\, C_2^2\, .
$$
Suppose a number $\varepsilon _1\in (0,\frac 12\, ]$ satisfies the inequalities
$$
(2\varepsilon _1)^2<\frac {\delta \delta _1}{40}\, (1-\widetilde \varepsilon )\, a^2
\, ,\ \
\frac 92\, \frac {|\gamma |^2}{\pi ^2}\ b_2^2\, (2\varepsilon _1)^2<
\frac {\delta \delta _1}{50}\, (1-\widetilde \varepsilon )\, C_2^2\, .
$$
Let us denote
$$
C^{\, \prime }_3\, (\varepsilon _1)=\frac {3\varepsilon _1}2\, (\widetilde c+2)R
+C_{\varepsilon _1}^{\, \prime }(0,\widehat W)\, .
$$

From (3.16), (3.22) (for $\varepsilon =\varepsilon _1$), and (3.23) (for 
$\varepsilon =1$), it follows that
$$
\| \, \widehat P^{\, \{ \widetilde \Omega _{\lambda}\} }\, \widehat P^{\, -}_{\lambda}\,
(\widehat {\mathcal D}(k+i\varkappa e)+\widehat W)\, \varphi ^{\, {\mathcal K}_{\lambda}}
\, \| =  \eqno (3.25)
$$ $$
=\, \| \, \widehat P^{\, \{ \widetilde \Omega _{\lambda}\} }\, \widehat P^{\, -}_{\lambda}\,
(\widehat {\mathcal D}^{\, (\lambda )}+\widehat {\mathcal D}^{\, (\lambda )}_{\perp})\, 
\varphi ^{\, {\mathcal K}_{\lambda}}\, \| \geqslant
\| \, \widehat P^{\, \{ \widetilde \Omega _{\lambda}\} }\, \widehat {\mathcal 
D}^{\, (\lambda )}\varphi ^{\, +}_{\lambda}\, \| - \| \, \widehat P^{\, \{ \widetilde 
\Omega _{\lambda}\} }\, \widehat {\mathcal D}^{\, (\lambda )}_{\perp}\varphi ^{\, -}
_{\lambda}\, \| \geqslant
$$ $$
\geqslant \| \, \widehat P^{\, \{ \widetilde \Omega _{\lambda}\} }\, \widehat {\mathcal 
D}^{\, (\lambda )}_0\, \varphi ^{\, +}_{\lambda}\, \| - \| \, (A^{(\lambda )}_1
\widehat \alpha ^{(\lambda )}_1+A^{(\lambda )}_2\widehat \alpha ^{(\lambda )}_2)\, 
\varphi ^{\, +}_{\lambda}\, \| - \| \, \widehat {\mathcal D}^{\, (\lambda )}_{\perp}
\varphi ^{\, -}_{\lambda}\, \| \geqslant
$$ $$
\geqslant  \| \, \widehat {\mathcal D}^{\, (\lambda )}_0\, \widetilde \varphi ^{\, +}
_{\lambda}\, \| - \| \, |A|\, \varphi ^{\, +}_{\lambda}\, \| - \| \, \widehat 
{\mathcal D}^{\, (\lambda )}_{\perp}\varphi ^{\, -}_{\lambda}\, \| \geqslant
$$ $$
\geqslant \| \, \widehat G^{\, 1}_{\lambda ,\, +}\, \widetilde \varphi ^{\, +}
_{\lambda}\, \| -\bigl( \, \varepsilon _1\, \| \, \widehat G^{\, 1}_{\lambda ,\, -}
\, \varphi ^{\, +}_{\lambda}\, \| +C_3^{\, \prime}\, (\varepsilon _1)\, \| \, 
\varphi ^{\, +}_{\lambda}\, \| \, \bigr) - \bigl( \, \| \, \widehat G^{\, 1}_{\lambda
,\, -}\, \varphi ^{\, -}_{\lambda}\, \| +C_3^{\, \prime}\, \| \, \varphi ^{\, -} 
_{\lambda}\, \| \, \bigr) \, .
$$

\begin{lemma} \label{l3.4}
For all $\lambda \in {\mathcal L}$,
$$
\| \, \widehat {\mathcal D}^{\, (\lambda )}\, \widetilde \varphi ^{\, -}
_{\lambda}\, \| \geqslant C_2\, \| \, \widehat G^{\, 1}_{\lambda ,\, -}\, \widetilde 
\varphi ^{\, -}_{\lambda}\, \| \, .
$$
\end{lemma}

\begin{proof}
From (3.12) it follows that
$$
\| \, \widehat {\mathcal D}^{\, (\lambda )}_0\, \widetilde \varphi ^{\, -}_{\lambda}\, 
\| = \| \, \widehat G^{\, 1}_{\lambda ,\, -}\, \widetilde \varphi ^{\, -}_{\lambda}
\, \| \, .
\eqno (3.26)
$$
By (1.9), (3.9), and the equality (3.26), we obtain
$$
\| \, (A^{(\lambda )}_1\widehat \alpha ^{(\lambda )}_1+A^{(\lambda )}_2\widehat \alpha 
^{(\lambda )}_2)\, \widetilde \varphi ^{\, -}_{\lambda}\, \| \leqslant \| \, |A|
\, \widetilde \varphi ^{\, -}_{\lambda}\, \| \leqslant  \eqno (3.27)
$$ $$
\leqslant \tau \, \bigl( \, \bigl\| \, \bigl( k_1^{\, (\lambda )}-\varkappa -i\, \frac 
{\partial}{\partial x_1^{(\lambda )}}\, \bigr) \, \widetilde \varphi ^{\, -}_{\lambda}
\, \bigr\| ^2+\bigl\| \bigl( k_2^{\, (\lambda )}-i\, \frac {\partial}{\partial x_2
^{(\lambda )}}\, \bigr) \, \widetilde \varphi ^{\, -}_{\lambda}\, \bigr\| ^2\, \bigr)
^{1/2}+Q\, \| \, \widetilde \varphi ^{\, -}_{\lambda}\, \| \leqslant
$$ $$
\leqslant \tau \, \| \, \widehat {\mathcal D}^{\, (\lambda )}_0\, \widetilde \varphi 
^{\, -}_{\lambda}\, \| + Q\, \| \, \widetilde \varphi ^{\, -}_{\lambda}\, \| =
\tau \, \| \, \widehat G^{\, 1}_{\lambda ,\, -}\, \widetilde \varphi ^{\, -}
_{\lambda}\, \| + Q\, \| \, \widetilde \varphi ^{\, -}_{\lambda}\, \| \, .
$$
Whence
$$
\| \, \widehat {\mathcal D}^{\, (\lambda )}\, \widetilde \varphi ^{\, -}_{\lambda}\, \|
\geqslant \| \, \widehat {\mathcal D}^{\, (\lambda )}_0\, \widetilde \varphi ^{\, -}
_{\lambda}\, \| - \| \, (A^{(\lambda )}_1\widehat \alpha ^{(\lambda )}_1+A^{(\lambda )}
_2\widehat \alpha ^{(\lambda )}_2)\, \widetilde \varphi ^{\, -}_{\lambda}\, \|
\geqslant  \eqno (3.28)
$$ $$
\geqslant (1-\tau )\, \| \, \widehat G^{\, 1}_{\lambda ,\, -}\, \widetilde \varphi 
^{\, -}_{\lambda}\, \| - Q\, \| \, \widetilde \varphi ^{\, -}_{\lambda}\, \| \, .
$$
Now, we use identity (3.8). Denote
$$
\chi _{\lambda}=e^{-i\, \Phi ^{\, (1,\lambda )}}\, e^{-i\, \widehat \alpha ^{(\lambda 
)}_1\widehat \alpha ^{(\lambda )}_2\, \Phi ^{\, (2,\lambda )}}\, \widetilde \varphi 
^{\, -}_{\lambda}\, .
$$
Since the functions $\Phi ^{\, (s,\lambda )}$, $s=1,2$, $\lambda \in {\mathcal L}$,
are trigonometric polynomials (and $\Phi ^{\, (s,\lambda )}_N=0$ for all vectors
$N\in \Lambda ^*$ with $A_N=0$), we have $\chi _{\lambda}\in \widetilde H^1(K;
{\mathbb C}^M)$. Furthermore, the operator $\widehat P^{\, -}_{\lambda}$
commutes with the operators of multiplication by the function $e^{-i\, \Phi ^{\, 
(1,\lambda )}}$ and by the matrix function $e^{-i\, \widehat \alpha ^{(\lambda )}
_1\widehat \alpha ^{(\lambda )}_2\, \Phi ^{\, (2,\lambda )}}$. Using 
Lemma \ref{l3.1}, inequality (0.13) (condition $\bf (A_2)$), and inequality
$$
\| \, \widehat {\mathcal D}^{\, (\lambda )}_0\, \chi _{\lambda}\, \| \geqslant
\frac {\pi}{|\gamma |}\ \| \, \chi _{\lambda}\, \| \, ,  \eqno (3.29)
$$
which is a consequence of the choice of the vector $k\in {\mathbb R}^d$ with $|(k,
\gamma )|=\pi $ (see (3.9)), we get
$$
\| \, \widehat {\mathcal D}^{\, (\lambda )}\, \widetilde \varphi ^{\, -}_{\lambda}\, \|
\geqslant e^{\, -\, \frac 12\, \| \mu \| \, C^*\, (h)}\, \| \, \widehat {\mathcal D}
^{\, (\lambda )}_*\, \chi _{\lambda}\, \| \geqslant
$$ $$
\geqslant e^{\, -\, \frac 12\, \| \mu \| \, C^*\, (h)}\, \bigl( \, \| \, \widehat 
{\mathcal D}^{\, (\lambda )}_0\, \chi _{\lambda}\, \| - \| \, (\widetilde A
^{(\lambda )}_1\widehat \alpha ^{(\lambda )}_1+\widetilde A^{(\lambda )}_2\widehat \alpha 
^{(\lambda )}_2)\, \chi _{\lambda}\, \| \, \bigr) \geqslant
$$ $$
\geqslant e^{\, -\, \frac 12\, \| \mu \| \, C^*\, (h)}\, \bigl( \, \| \, \widehat 
{\mathcal D}^{\, (\lambda )}_0\, \chi _{\lambda}\, \| - \| \, |\widetilde A
(\widetilde e^{\, \lambda };.)|\, \chi _{\lambda}(.)\, \| \, \bigr) \geqslant
$$ $$
\geqslant e^{\, -\, \frac 12\, \| \mu \| \, C^*\, (h)}\, (1-\widetilde \theta )\,
\frac {\pi}{|\gamma |}\ \| \chi _{\lambda} \| \geqslant
(1-\widetilde \theta )\, \frac {\pi}{|\gamma |}\ e^{\, - \| \mu \| \, C^*\, (h)}\,
\| \, \widetilde \varphi ^{\, -}_{\lambda}\, \| \, .
$$
Multiplying inequality (3.28) by $(1-\widetilde \theta )\, \frac {\pi}{|\gamma 
|}\, e^{\, - \| \mu \| \, C^*\, ({\mathfrak h})}\, $, multiplying inequality (3.29) by 
$Q$, and adding them together, we derive the claimed estimate.
\end{proof}

Now, let us get estimate (3.33), which complements estimate (3.25). By (3.2), (3.12),
and Lemma \ref{l3.4}, we have 
$$
\| \, \widehat P^{\, \{ \Omega _{\lambda}\} }\, \widehat {\mathcal D}^{\, (\lambda )}
\, \widehat P^{\, -}_{\lambda}\, \varphi ^{\, {\mathcal K}^{\, \prime}_{\lambda}}\, 
\| =  \eqno (3.30)
$$ $$
=\, \| \, \widehat P^{\, \{ \Omega _{\lambda}\} }\, (\widehat {\mathcal D}^{\, (\lambda )}
_0-A_1^{(\lambda )}\widehat \alpha ^{(\lambda )}_1-A_2^{(\lambda )}\widehat \alpha 
^{(\lambda )}_2)\, \widehat P^{\, -}_{\lambda}\, \varphi ^{\, {\mathcal K}^{\, \prime}
_{\lambda}}\, \| = 
$$ $$
=\, \| \, \widehat {\mathcal D}^{\, (\lambda )}_0\, \widehat P^{\, -}_{\lambda}\,
(\varphi ^{\, {\mathcal K}_{\lambda}\, \backslash \, \widetilde {\mathcal K}
_{\lambda}}+\varphi ^{\, \widetilde {\mathcal K}_{\lambda}}\, )-
\widehat P^{\, \{ \Omega _{\lambda}\} }\, (A_1^{(\lambda )}\widehat \alpha ^{(\lambda )}_1
+A_2^{(\lambda )}\widehat \alpha ^{(\lambda )}_2\, )\, \widehat P^{\, -}_{\lambda}\,
(\varphi ^{\, {\mathcal K}^{\, \prime }_{\lambda}\, \backslash \, \widetilde 
{\mathcal K}_{\lambda}}+\varphi ^{\, \widetilde {\mathcal K}_{\lambda}}\, )\,
\| =
$$ $$
=\, \| \, \widehat {\mathcal D}^{\, (\lambda )}\, \widetilde \varphi ^{\, -}
_{\lambda}+\widehat {\mathcal D}^{\, (\lambda )}_0\, \widehat P^{\, -}_{\lambda}\,
\varphi ^{\, {\mathcal K}_{\lambda}\, \backslash \, \widetilde {\mathcal K}
_{\lambda}}-\widehat P^{\, \{ \Omega _{\lambda}\} }\, (A_1^{(\lambda )}\widehat \alpha 
^{(\lambda )}_1+A_2^{(\lambda )}\widehat \alpha ^{(\lambda )}_2\, )\, \widehat P^{\, -}
_{\lambda}\, \varphi ^{\, {\mathcal K}^{\, \prime }_{\lambda}\, \backslash \, 
\widetilde {\mathcal K}_{\lambda}}\, \| \geqslant
$$ $$
\geqslant \| \, \widehat {\mathcal D}^{\, (\lambda )}\, \widetilde \varphi ^{\, -}
_{\lambda}\, \| -\| \, \widehat {\mathcal D}^{\, (\lambda )}_0\, \widehat P^{\, -}
_{\lambda}\, \varphi ^{\, {\mathcal K}_{\lambda}\, \backslash \, \widetilde 
{\mathcal K}_{\lambda}}\, \| - \| \, (A_1^{(\lambda )}\widehat \alpha 
^{(\lambda )}_1+A_2^{(\lambda )}\widehat \alpha ^{(\lambda )}_2\, )\, \widehat P^{\, -}
_{\lambda}\, \varphi ^{\, {\mathcal K}^{\, \prime }_{\lambda}\, \backslash \, 
\widetilde {\mathcal K}_{\lambda}}\, \| \geqslant
$$ $$
\geqslant C_2\, \| \, \widehat G_{\lambda ,\, -}^{\, 1}\, \widetilde \varphi ^{\, -}
_{\lambda}\, \| -\| \, \widehat G^{\, 1}_{\lambda ,\, -}\, \widehat P^{\, -}
_{\lambda}\, \varphi ^{\, {\mathcal K}_{\lambda}\, \backslash \, \widetilde {\mathcal 
K}_{\lambda}}\, \| - \| \, |A|\, \widehat P^{\, -}_{\lambda}\, \varphi ^{\, {\mathcal 
K}^{\, \prime }_{\lambda}\, \backslash \, \widetilde {\mathcal K}_{\lambda}}\, \|
\, .
$$
At the same time, from (1.9), (3.9) and (3.12) (by analogy with the estimate (3.27)) 
it follows that
$$
\| \, |A|\, \widehat P^{\, -}_{\lambda}\, \varphi ^{\, {\mathcal 
K}^{\, \prime }_{\lambda}\, \backslash \, \widetilde {\mathcal K}_{\lambda}}\, \|
\leqslant \tau \, \| \, \widehat G^{\, 1}_{\lambda ,\, -}\, \widehat P^{\, -}
_{\lambda}\, \varphi ^{\, {\mathcal K}^{\, \prime }_{\lambda}\, \backslash \, 
\widetilde {\mathcal K}_{\lambda}}\, \| + Q\, \| \, \widehat P^{\, -}_{\lambda}\, 
\varphi ^{\, {\mathcal K}^{\, \prime }_{\lambda}\, \backslash \, \widetilde 
{\mathcal K}_{\lambda}}\, \| \, ,  \eqno (3.31)
$$
and (for $\varepsilon =\varepsilon _1$ and $\psi =\widehat P^{\, +}_{\lambda}\, 
\varphi ^{\, {\mathcal K}^{\, \prime }_{\lambda}}\, $) inequality (3.24) implies
$$
\| \, \widehat {\mathcal D}^{\, (\lambda )}_{\perp}\, \widehat P^{\, +}_{\lambda}\, 
\varphi ^{\, {\mathcal K}^{\, \prime }_{\lambda}}\, \| \leqslant
\varepsilon _1\, \| \, \widehat G_{\lambda ,\, -}^{\, 1}\, \widehat P^{\, +}_{\lambda}
\, \varphi ^{\, {\mathcal K}^{\, \prime }_{\lambda}}\, \| +C_3^{\, \sharp }\,
(\varepsilon _1)\, \| \, \widehat P^{\, +}_{\lambda}\, \varphi ^{\, {\mathcal K}^{\, 
\prime }_{\lambda}}\, \| \, .  \eqno (3.32)
$$
Whence (see (3.30), (3.31), (3.32))
$$
\| \, \widehat P^{\, \{ \Omega _{\lambda}\} }\, \widehat P^{\, +}_{\lambda}\,
(\widehat {\mathcal D}(k+i\varkappa e)+\widehat W)\, \varphi ^{\, {\mathcal K}^{\, 
\prime}_{\lambda}}\, \| =  \eqno (3.33)
$$ $$
=\, \| \, \widehat P^{\, \{ \Omega _{\lambda}\} }\, \widehat P^{\, +}_{\lambda}\,
(\widehat {\mathcal D}^{\, (\lambda )}+\widehat {\mathcal D}^{\, (\lambda )}_{\perp})\, 
\varphi ^{\, {\mathcal K}^{\, \prime}_{\lambda}}\, \| \geqslant
\| \, \widehat P^{\, \{ \Omega _{\lambda}\} }\, \widehat {\mathcal D}^{\, (\lambda )}\,
\widehat P^{\, -}_{\lambda}\, \varphi ^{\, {\mathcal K}^{\, \prime}_{\lambda}}\, \|
- \| \, \widehat {\mathcal D}^{\, (\lambda )}_{\perp}\, \widehat P^{\, +}_{\lambda}\, 
\varphi ^{\, {\mathcal K}^{\, \prime}_{\lambda}}\, \| \geqslant
$$ $$
\geqslant C_2\, \| \, \widehat G_{\lambda ,\, -}^{\, 1}\, \widetilde \varphi 
_{\lambda}^{\, -}\, \| - \| \, \widehat G_{\lambda ,\, -}^{\, 1}\, \widehat P^{\, 
-}_{\lambda}\, \varphi ^{\, {\mathcal K}_{\lambda}\, \backslash \, \widetilde 
{\mathcal K}_{\lambda}}\, \| -
\bigl( \tau \, \| \, \widehat G^{\, 1}_{\lambda ,\, -}\, \widehat P^{\, -}
_{\lambda}\, \varphi ^{\, {\mathcal K}^{\, \prime }_{\lambda}\, \backslash \, 
\widetilde {\mathcal K}_{\lambda}}\, \| + Q\, \| \, \widehat P^{\, -}_{\lambda}\, 
\varphi ^{\, {\mathcal K}^{\, \prime }_{\lambda}\, \backslash \, \widetilde 
{\mathcal K}_{\lambda}}\, \| \, \bigr) -
$$ $$
-\, \bigl( \, \varepsilon _1\, \| \, \widehat G_{\lambda ,\, -}^{\, 1}\, \widehat P
^{\, +}_{\lambda}\, \varphi ^{\, {\mathcal K}^{\, \prime }_{\lambda}}\, \| +C_3^{\, 
\sharp }\, (\varepsilon _1)\, \| \, \widehat P^{\, +}_{\lambda}\, \varphi 
^{\, {\mathcal K}^{\, \prime }_{\lambda}}\, \| \, \bigr) \geqslant
$$ $$
\geqslant C_2\, \| \, \widehat G_{\lambda ,\, -}^{\, 1}\, \widetilde \varphi 
_{\lambda}^{\, -}\, \| - C_3\, \| \, \widehat G^{\, 1}_{\lambda ,\, -}\, 
\widehat P^{\, -}_{\lambda}\, \varphi ^{\, {\mathcal K}^{\, \prime }_{\lambda}\, 
\backslash \, \widetilde {\mathcal K}_{\lambda}}\, \| - 
\varepsilon _1\, \| \, \widehat G_{\lambda ,\, -}^{\, 1}\, 
\widehat P^{\, +}_{\lambda}\, \varphi ^{\, {\mathcal K}^{\, \prime }_{\lambda}}\, \| 
- C_3^{\, \sharp }\, (\varepsilon _1)\, \| \, \widehat P^{\, +}_{\lambda}\, \varphi 
^{\, {\mathcal K}^{\, \prime }_{\lambda}}\, \| \, .
$$
In what follows, we assume that $C_3^{\, \sharp }\, (\varepsilon _1)\leqslant 
\varepsilon _1 \varkappa _0\, $. Then $C_3^{\, \prime }\, (\varepsilon 
_1)\leqslant \varepsilon _1 \varkappa _0\, $ as well. Since for all vector functions
$\psi \in \widetilde H^1(K;{\mathbb C}^M)\cap {\mathcal H}({\mathcal K}^{\, 
\prime }_{\lambda})$ the inequalities $\| \widehat G^{\, 1}_{\lambda ,\, -}\, \psi 
\| \leqslant \| \widehat G^{\, 1}_{\lambda ,\, +}\, \psi \| $ and $\varkappa \, \| 
\psi \| \leqslant \| \widehat G^{\, 1}_{\lambda ,\, +}\, \psi \| $ hold, we derive
(for all $\varkappa \geqslant \varkappa _0$ and all $\lambda \in {\mathcal L}$)
$$
\varepsilon _1\, \| \, \widehat G_{\lambda ,\, -}^{\, 1}\, \varphi _{\lambda}^{\, +}
\, \| + C_3^{\, \prime }\, (\varepsilon _1)\, \| \, \varphi _{\lambda}^{\, +}\, \|
\leqslant 2\varepsilon _1\, \| \, \widehat G_{\lambda ,\, +}^{\, 1}\, \varphi 
_{\lambda}^{\, +}\, \| \, ,
$$ $$
\varepsilon _1\, \| \, \widehat G_{\lambda ,\, -}^{\, 1}\, \widehat P^{\, +}_{\lambda}
\, \varphi ^{\, {\mathcal K}^{\, \prime }_{\lambda}}\, \| +C_3^{\, \sharp }\,
(\varepsilon _1)\, \| \, \widehat P^{\, +}_{\lambda}\, \varphi ^{\, {\mathcal K}^{\, 
\prime }_{\lambda}}\, \| \leqslant 2\varepsilon _1\, \| \, \widehat G_{\lambda ,\, +}
^{\, 1}\, \widehat P^{\, +}_{\lambda}\, \varphi ^{\, {\mathcal K}^{\, \prime }
_{\lambda}}\, \| \, .
$$
Hence, from (3.25) and (3.31) it follows that
$$
\| \, \widehat P^{\, \{ \widetilde \Omega _{\lambda}\} }\, \widehat P^{\, -}_{\lambda}\,
(\widehat {\mathcal D}(k+i\varkappa e)+\widehat W)\, \varphi ^{\, {\mathcal K}
_{\lambda}}\, \| \geqslant
\| \, \widehat G_{\lambda ,\, +}^{\, 1}\, \widetilde \varphi _{\lambda}
^{\, +}\, \| - 2\varepsilon _1\, \| \, \widehat G_{\lambda ,\, +}^{\, 1}\, \varphi 
_{\lambda}^{\, +}\, \| - C_4\, \| \, \widehat G_{\lambda ,\, -}^{\, 1}\, \varphi 
_{\lambda}^{\, -}\, \| \, ,
$$ $$
\| \, \widehat P^{\, \{ \Omega _{\lambda}\} }\, \widehat P^{\, +}_{\lambda}\,
(\widehat {\mathcal D}(k+i\varkappa e)+\widehat W)\, \varphi ^{\, {\mathcal K}^{\, 
\prime }_{\lambda}}\, \| \geqslant
$$ $$
\geqslant C_2\, \| \, \widehat G_{\lambda ,\, -}^{\, 1}\, \widetilde \varphi 
_{\lambda}^{\, -}\, \| - 2\varepsilon _1\, \| \, \widehat G_{\lambda ,\, +}^{\, 1}\, 
\widehat P^{\, +}_{\lambda}\, \varphi ^{\, {\mathcal K}^{\, \prime }_{\lambda}}\, \|
- C_3\, \| \, \widehat G^{\, 1}_{\lambda ,\, -}\, \widehat P^{\, -}_{\lambda}\, 
\varphi ^{\, {\mathcal K}^{\, \prime }_{\lambda}\, \backslash \, \widetilde 
{\mathcal K}_{\lambda}}\, \| \, .
$$
Using Lemma \ref{l3.2} (inequalities (3.13) and (3.14)), from the last estimate we get
$$
\| \, \widehat P^{\, \{ \widetilde \Omega _{\lambda}\} }\, \widehat P^{\, -}
_{\lambda}\, (\widehat {\mathcal D}(k+i\varkappa e)+\widehat W)\, \varphi ^{\, 
{\mathcal K}_{\lambda}}\, \| \geqslant  \eqno (3.34)
$$ $$
\geqslant \| \, \widehat G^{\, 1}_+ \widehat P^{\, +} \varphi ^{\, \widetilde
{\mathcal K}_{\lambda}}\, \| - 2\varepsilon _1\, \| \, \widehat G^{\, 1}_+ 
\widehat P^{\, +} \varphi ^{\, {\mathcal K}_{\lambda}}\, \| - C_4\, \| \, 
\widehat G^{\, 1}_- \widehat P^{\, -} \varphi ^{\, {\mathcal K}_{\lambda}}\, \| -
$$ $$
-\ \frac {b_1}{\varkappa }\, \bigl( \, \| \, \widehat G^{\, 1}_+ \varphi ^{\, 
\widetilde {\mathcal K}_{\lambda}}\, \| + 2\varepsilon _1\, \| \, \widehat G^{\, 1}_+ 
\varphi ^{\, {\mathcal K}_{\lambda}}\, \| + C_4\, \| \, \widehat G^{\, 1}_- \varphi 
^{\, {\mathcal K}_{\lambda}}\, \| \, \bigr) \, ,
$$ $$
\| \, \widehat P^{\, \{ \Omega _{\lambda}\} }\, \widehat P^{\, +}_{\lambda}\,
(\widehat {\mathcal D}(k+i\varkappa e)+\widehat W)\, \varphi ^{\, {\mathcal K}^{\, 
\prime }_{\lambda}}\, \| \geqslant  \eqno (3.35)
$$ $$
\geqslant C_2\, \| \, \widehat G^{\, 1}_- \widehat P^{\, -} \varphi ^{\, \widetilde
{\mathcal K}_{\lambda}}\, \| - 2\varepsilon _1\, \| \, \widehat G^{\, 1}_+ 
\widehat P^{\, +} \varphi ^{\, {\mathcal K}^{\, \prime }_{\lambda}}\, \|
- C_3\, \| \, \widehat G^{\, 1}_- \widehat P^{\, -} 
\varphi ^{\, {\mathcal K}^{\, \prime }_{\lambda}\, \backslash \, \widetilde 
{\mathcal K}_{\lambda}}\, \| -
$$ $$
-\ \frac {b_2}{\varkappa }\, \bigl( \, C_2\, \| \, \widehat G^{\, 1}_- \varphi ^{\, 
\widetilde {\mathcal K}_{\lambda}}\, \| + 2\varepsilon _1\, \| \, \widehat G^{\, 1}_+ 
\varphi ^{\, {\mathcal K}^{\, \prime }_{\lambda}}\, \|
+ C_3\, \| \, \widehat G^{\, 1}_- \varphi ^{\, {\mathcal K}^{\, \prime }_{\lambda}\, 
\backslash \, \widetilde {\mathcal K}_{\lambda}}\, \| \, \bigr) \, .
$$ 
From (3.2) and (3.19) (for a number $\varkappa \geqslant \varkappa _0\, $,
a vector $k\in {\mathbb R}^d$ with $|(k,\gamma )|=\pi $, and a vector function 
$\varphi \in \widetilde H^1(K;{\mathbb C}^M)\cap {\mathcal H}({\mathcal C}(\frac 
12))$) we obtain 
$$
\| \, \widehat P^{\, \{ \widetilde \Omega _{\lambda}\} }\, \widehat P^{\, -}\,
(\widehat {\mathcal D}(k+i\varkappa e)+\widehat W)\, \varphi \, \| \geqslant  
$$ $$
\geqslant \| \, \widehat P^{\, \{ \widetilde \Omega _{\lambda}\} }\, \widehat P^{\, -}
_{\lambda}\, (\widehat {\mathcal D}(k+i\varkappa e)+\widehat W)\, \varphi ^{\, 
{\mathcal K}_{\lambda}}\, \| - \frac {\rho}2\ \| \, \widehat P^{\, \{ \widetilde 
\Omega _{\lambda}\} }\, (\widehat {\mathcal D}(k+i\varkappa e)+\widehat W)\, 
\varphi \, \| \, ,
$$ $$
\| \, \widehat P^{\, \{ \Omega _{\lambda}\} }\, \widehat P^{\, +}
(\widehat {\mathcal D}(k+i\varkappa e)+\widehat W)\, \varphi \, \| \geqslant
$$ $$
\geqslant \| \, \widehat P^{\, \{ \Omega _{\lambda}\} }\, \widehat P^{\, +}_{\lambda}\,
(\widehat {\mathcal D}(k+i\varkappa e)+\widehat W)\, \varphi ^{\, {\mathcal K}^{\, 
\prime }_{\lambda}}\, \| - \frac {\rho}2\ \| \, \widehat P^{\, \{ \Omega _{\lambda}\} 
}\, (\widehat {\mathcal D}(k+i\varkappa e)+\widehat W)\, \varphi \, \| \, .
$$
Whence 
$$
\| \, \widehat P^{\, +}(\widehat {\mathcal D}(k+i\varkappa e)+\widehat W)\, \varphi \, \| ^2+
a^{\, 2}\, \| \, \widehat P^{\, -}(\widehat {\mathcal D}(k+i\varkappa e)+\widehat W)\, 
\varphi \, \| ^2 \geqslant  \eqno (3.36)
$$ $$
\geqslant \sum\limits_{\lambda}\, \| \, \widehat P^{\, \{ \Omega _{\lambda}\} }\, 
\widehat P^{\, +}(\widehat {\mathcal D}(k+i\varkappa e)+\widehat W)\, \varphi \, \| ^2+
$$ $$
+\, a^{\, 2}\, \sum\limits_{\lambda}\, \| \, \widehat P^{\, \{ \widetilde \Omega 
_{\lambda}\} }\, \widehat P^{\, -}(\widehat {\mathcal D}(k+i\varkappa e)+\widehat W)\, 
\varphi \, \| ^2\geqslant
$$ $$
\geqslant (1-\delta _1)\, \sum\limits_{\lambda}\, \| \, \widehat P^{\, \{ \Omega 
_{\lambda}\} }\, \widehat P^{\, +}_{\lambda}\, (\widehat {\mathcal D}(k+i\varkappa e)+\widehat W)\, 
\varphi ^{\, {\mathcal K}^{\, \prime }_{\lambda}}\, \, \| ^2+
$$ $$
+\, (1-\delta _1)\, a^{\, 2}\, \sum\limits_{\lambda}\, \| \, \widehat P^{\, \{ \widetilde
\Omega _{\lambda}\} }\, \widehat P^{\, -}_{\lambda}\, (\widehat {\mathcal D}(k+i\varkappa 
e)+\widehat W)\, \varphi ^{\, {\mathcal K}_{\lambda}}\, \, \| ^2\, -
$$ $$
-\ \frac {(1-\delta _1)}{\delta _1}\ \frac {1+a^2}4\ \rho ^2\, \sum\limits
_{\lambda}\, \| \, \widehat P^{\, \{ \Omega _{\lambda}\} }\, (\widehat {\mathcal D}(k+i
\varkappa e)+\widehat W)\, \varphi \, \| ^2\, .
$$
Furthermore,
$$
\sum\limits_{\lambda}\, \| \, \widehat P^{\, \{ \Omega _{\lambda}\} }\, (\widehat 
{\mathcal D}(k+i\varkappa e)+\widehat W)\, \varphi \, \| ^2 \leqslant \| \, (\widehat 
{\mathcal D}(k+i\varkappa e)+\widehat W)\, \varphi \, \| ^2 \leqslant  \eqno (3.37)
$$ $$
\leqslant \frac 1{a^2}\ \bigl( \, \| \, \widehat P^{\, +}(\widehat {\mathcal D}(k+i
\varkappa e)+\widehat W)\, \varphi \, \| ^2+a^{\, 2}\, \| \, \widehat P^{\, -}(\widehat 
{\mathcal D}(k+i\varkappa e)+\widehat W)\, \varphi \, \| ^2 \, \bigr) \, .
$$
Next, suppose the number $\varkappa _0$ satisfies the condition
$$
\frac {(1-\delta _1)}{\delta _1}\ \frac {1+a^2}{4a^2}\ (\widetilde c+2)^2
R^2\leqslant \frac {\delta _2-\delta _1}{1-\delta _2}\ \varkappa _0^2
$$
and $\varkappa \geqslant \varkappa _0\, $, then (3.36) and (3.37) yield
$$
(1-\delta _2)^{-1}\, \bigl( \, \| \, \widehat P^{\, +}(\widehat {\mathcal D}(k+i
\varkappa e)+\widehat W)\, \varphi \, \| ^2+  
a^{\, 2}\, \| \, \widehat P^{\, -}(\widehat 
{\mathcal D}(k+i\varkappa e)+\widehat W)\, \varphi \, \| ^2 \, \bigr) \geqslant  
\eqno (3.38)
$$ $$
\geqslant \sum\limits_{\lambda}\, \| \, \widehat P^{\, \{ \Omega _{\lambda}\} }\, 
\widehat P^{\, +}_{\lambda}\, (\widehat {\mathcal D}(k+i\varkappa e)+\widehat W)\, 
\varphi ^{\, {\mathcal K}^{\, \prime }_{\lambda}}\, \, \| ^2+
a^{\, 2}\, \sum\limits_{\lambda}\, \| \, \widehat P^{\, \{ \widetilde
\Omega _{\lambda}\} }\, \widehat P^{\, -}_{\lambda}\, (\widehat {\mathcal D}(k+i\varkappa 
e)+\widehat W)\, \varphi ^{\, {\mathcal K}_{\lambda}}\, \, \| ^2\, .
$$ 
On the other hand, from estimates (3.34) and (3.35) it follows that the right hand
side of previous inequality (3.38) is greater than or equal to
$$
(1-\delta _1)\, \biggl( \, C_2^2\, \sum\limits_{\lambda}\, \| \, \widehat G^{\, 1}_- 
\widehat P^{\, -} \varphi ^{\, \widetilde {\mathcal K}_{\lambda}}\, \| ^2+
a^{\, 2}\, \sum\limits_{\lambda}\, \| \, \widehat G^{\, 1}_+ \widehat P^{\, +} \varphi 
^{\, \widetilde {\mathcal K}_{\lambda}}\, \| ^2\, \biggr) \, -  \eqno (3.39)
$$ $$
-\ \frac {5(1-\delta _1)}{\delta _1}\ \biggl( \, (2\varepsilon _1)^2\,
\sum\limits_{\lambda}\, \| \, \widehat G^{\, 1}_+ \widehat P^{\, +} \varphi ^{\, 
{\mathcal K}^{\, \prime }_{\lambda}}\, \| ^2+C_3^2\, \sum\limits_{\lambda}\,
\| \, \widehat G^{\, 1}_- \widehat P^{\, -} \varphi ^{\, {\mathcal K}^{\, \prime }
_{\lambda}\, \backslash \, \widetilde {\mathcal K}_{\lambda}}\, \| ^2+
$$ $$
+\, a^2\, (2\varepsilon _1)^2\, \sum\limits_{\lambda}\, \| \, \widehat G^{\, 1}_+ 
\widehat P^{\, +} \varphi ^{\, {\mathcal K}_{\lambda}}\, \| ^2+a^2\, C_4^2\,
\sum\limits_{\lambda}\, \| \, \widehat G^{\, 1}_- \widehat P^{\, -} \varphi ^{\, 
{\mathcal K}_{\lambda}}\, \| ^2+
$$ $$
+\ \frac {b_2^2}{\varkappa ^2}\ \biggl( C_2^2\, \sum\limits_{\lambda}\, \| \, 
\widehat G^{\, 1}_- \varphi ^{\, \widetilde {\mathcal K}_{\lambda}}\, \| ^2 + (2
\varepsilon _1)^2\, \sum\limits_{\lambda}\, \| \, \widehat G^{\, 1}_+ \varphi ^{\, 
{\mathcal K}^{\, \prime }_{\lambda}}\, \| ^2+ C_3^2\, \sum\limits_{\lambda}\,
\| \, \widehat G^{\, 1}_- \varphi ^{\, {\mathcal K}^{\, \prime }_{\lambda}\, 
\backslash \, \widetilde {\mathcal K}_{\lambda}}\, \| ^2 \, \biggr) \, +
$$ $$
+\ \frac {a^2b_1^2}{\varkappa ^2}\ \biggl( \, \sum\limits_{\lambda}\, \| \, \widehat 
G^{\, 1}_+ \varphi ^{\, \widetilde {\mathcal K}_{\lambda}}\, \| ^2 + (2
\varepsilon _1)^2\, \sum\limits_{\lambda}\, \| \, \widehat G^{\, 1}_+ \varphi ^{\, 
{\mathcal K}_{\lambda}}\, \| ^2 + C_4^2\, \sum\limits_{\lambda}\, \| \, \widehat 
G^{\, 1}_- \varphi ^{\, {\mathcal K}_{\lambda}}\, \| ^2 \, \biggr) \, \biggr) 
\geqslant
$$ $$
\geqslant
(1-\delta _1)\, \biggl( C_2^2\, \sum\limits_{\lambda}\, \| \, \widehat G^{\, 1}_- 
\widehat P^{\, -} \varphi ^{\, \widetilde {\mathcal K}_{\lambda}}\, \| ^2+
a^{\, 2}\, \sum\limits_{\lambda}\, \| \, \widehat G^{\, 1}_+ \widehat P^{\, +} \varphi 
^{\, \widetilde {\mathcal K}_{\lambda}}\, \| ^2\, \biggr) \, -  
$$ $$
-\ \frac {5(1-\delta _1)}{\delta _1}\ \biggl( \, (1+a^2)(2\varepsilon _1)^2\,
\| \, \widehat G^{\, 1}_+ \widehat P^{\, +} \varphi \| ^2 + C_3^2\, 
\sum\limits_{\lambda}\, \| \, \widehat G^{\, 1}_- \widehat P^{\, -} \varphi ^{\, 
{\mathcal K}^{\, \prime }_{\lambda}\, \backslash \, \widetilde {\mathcal K}_{\lambda}}
\, \| ^2 +
$$ $$
+\, a^2\, C_4^2\, \| \, \widehat G^{\, 1}_- \widehat P^{\, -} \varphi \, \| ^2+
+ \frac 1{\varkappa ^2}\ \bigl( \, b_2^2\, C_2^2+b_2^2\, C_3^2+a^2\, b_1^2\, C_4^2\, 
\bigr) \, \| \, \widehat G^{\, 1}_- \varphi \, \| ^2 + 
$$ $$
+\, \frac 1{\varkappa ^2}\ \bigl( \, b_2^2\, (2\varepsilon _1)^2+a^2\, b_1^2+a^2\,
b_1^2\, (2\varepsilon _1)^2\, \bigr) \, \| \, \widehat G^{\, 1}_+ \varphi \, \| ^2\, 
\biggr) \, .
$$
Let us now use condition 3b (also see (3.3)) which yields
$$
\sum\limits_{\lambda}\, \| \, \widehat G_{\pm }^{\, 1}\widehat P^{\, \pm } \varphi ^{\, 
\widetilde {\mathcal K}_{\lambda}}\, \| ^2\geqslant (1-\widetilde \varepsilon )\,
\| \, \widehat G^{\, 1}_{\pm}\widehat P^{\, \pm } \varphi \, \| ^2 \, ,
$$ $$
\sum\limits_{\lambda}\, \| \, \widehat G^{\, 1}_- \widehat P^{\, -} \varphi ^{\, 
{\mathcal K}^{\, \prime }_{\lambda}\, \backslash \, \widetilde {\mathcal K}
_{\lambda}}\, \| ^2 \leqslant \widetilde \varepsilon \, \| \, \widehat G^{\, 1}_- 
\widehat P^{\, -} \varphi \, \| ^2\, .
$$
Since
$$
\| \, \widehat G^{\, 1}_- \varphi \, \| ^2=\| \, \widehat G^{\, 1}_- \widehat P
^{\, -} \varphi \, \| ^2+\| \, \widehat G^{\, 1}_- \widehat P^{\, +} \varphi \, \| 
^2\leqslant \| \, \widehat G^{\, 1}_- \widehat P^{\, -} \varphi \, \| ^2+\| \, 
\widehat G^{\, 1}_+ \widehat P^{\, -} \varphi \, \| 
^2 \, ,
$$ $$
\| \, \widehat G^{\, 1}_+ \varphi \, \| ^2=\| \, \widehat G^{\, 1}_+ \widehat P^{\, 
+} \varphi \, \| ^2+\| \, \widehat G^{\, 1}_+ \widehat P^{\, -} \varphi \, \| ^2
\leqslant \| \, \widehat G^{\, 1}_+ \widehat P^{\, +} \varphi \, \| ^2+\frac 94\, 
\varkappa ^2\, \| \, \widehat P^{\, 1}_-\varphi \, \| ^2\leqslant 
$$ $$
\leqslant \| \, \widehat G^{\, 1}_+ \widehat P^{\, +} \varphi \, \| ^2 + \, \frac 94\ 
\frac {|\gamma |^2}{\pi ^2}\ \varkappa ^2\, \| \, \widehat G^{\, 1}_- \widehat P^{\, 
-} \varphi \, \| ^2\, , 
$$
and $C_2\in (0,1)$, $a\in (0,1)$, $\varepsilon _1\in (0,\frac 12\, ]$, $b_1<b_2\, $,
using (3.38) and (3.39) we deduce the estimate
$$
(1-\delta _2)^{-1}\, \bigl( \bigl( \, \| \, \widehat P^{\, +}(\widehat {\mathcal 
D}(k+i\varkappa e)+\widehat W)\, \varphi \, \| ^2+  
a^{\, 2}\, \| \, \widehat P^{\, -}(\widehat 
{\mathcal D}(k+i\varkappa e)+\widehat W)\, \varphi \, \| ^2 \, \bigr) \geqslant  
\eqno (3.40)
$$ $$
\geqslant
(1-\delta _1)\, (1-\widetilde \varepsilon )\, \bigl( C_2^2\, \| \, \widehat G^{\, 1}_- 
\widehat P^{\, -} \varphi \, \| ^2+a^{\, 2}\, \| \, \widehat G^{\, 1}_+ 
\widehat P^{\, +} \varphi \, \| ^2\, \bigr) \, -  
$$ $$
-\ \frac {5(1-\delta _1)}{\delta _1}\ \biggl( \, \bigl( \, \widetilde \varepsilon
\, C_3^2\,+a^2 C_4^2+\frac 94\, \frac {|\gamma |^2}{\pi ^2}\, (2(2\varepsilon 
_1)^2+a^2)\, b_2^2\, +
$$ $$
+\, (1+C_3^2+C_4^2)\, \frac {b_2^2}{\varkappa ^2}\, \bigr) \,
\| \, \widehat G^{\, 1}_- \widehat P^{\, -} \varphi \, \| ^2 +
\bigl( 2(2\varepsilon _1)^2+(4+C_3^2+C_4^2)\, \frac {b_2^2}{\varkappa ^2}\, 
\bigr) \, \| \, \widehat G^{\, 1}_+ \widehat P^{\, +} \varphi \, \| ^2\, \biggr) \, .
$$
Finally, suppose that the number $\varkappa _0\, $ also satisfies the conditions
$$
(1+C_3^2+C_4^2)\, b_2^2\leqslant \frac {\delta _1}5\, (1-\widetilde \varepsilon )\,
C_2^2\, \frac {\delta}{10}\, \varkappa _0^2\, ,\ \
(4+C_3^2+C_4^2)\, b_2^2\leqslant \frac {\delta _1}5\, (1-\widetilde \varepsilon )\,
a^{\, 2}\, \frac {\delta}{10}\, \varkappa _0^2\, .
$$
From the choice of the numbers $\widetilde \varepsilon $, $a$, and $\varepsilon _1$ 
it follows that
$$
2(2\varepsilon _1)^2<\frac {\delta _1}5\, (1-\widetilde \varepsilon )\,
a^{\, 2}\, \frac {\delta}4\ ,
$$ $$
\max \bigl\{ \, \widetilde \varepsilon \, C_3^2\, ,\ a^2\, C_4^2\, ,\ \frac 92\, 
\frac {|\gamma |^2}{\pi ^2}\ (2\varepsilon _1)^2\, b_2^2\, ,\
\frac 94\, \frac {|\gamma |^2}{\pi ^2}\ a^{\, 2}\, b_2^2 \, \bigr\} <
\frac {\delta _1}5\, (1-\widetilde \varepsilon )\, C_2^2\, \frac {\delta}{10}
$$
and $\widetilde \varepsilon \leqslant \frac {\delta}8\, $. Therefore, from (3.40) for
all $\varkappa \geqslant \varkappa _0$ (all $k\in {\mathbb R}^d$ with $|(k,\gamma )|
=\pi $, and all vector functions $\varphi \in \widetilde H^1(K;{\mathbb C}^M)\cap 
{\mathcal H}({\mathcal C}(\frac 12))$) we get the inequality
$$
\| \widehat P^{\, +}(\widehat {\mathcal D}(k+i\varkappa e)
+\widehat W)\varphi \| ^2+ a^{\, 2}\| \widehat P^{\, -}(\widehat {\mathcal D}
(k+i\varkappa e)+\widehat W)\varphi \| ^2\geqslant  
$$ $$
\geqslant (1-\delta _2)\, (1-\delta _1)\, (1-\widetilde \varepsilon )\, (1-\frac
{\delta}2\, )\, (C^2_2\, \| \widehat G_-^{\, 1}\widehat P^{\, -}\varphi \| ^2+
a^{\, 2}\, \| \widehat G_+^{\, 1}\widehat P^{\, +}\varphi \| ^2)\geqslant
$$ $$
\geqslant (1-\delta )\, (C^2_2\, \| \widehat G_-^{\, 1}\widehat P^{\, -}\varphi \| ^2+
a^{\, 2}\, \| \widehat G_+^{\, 1}\widehat P^{\, +}\varphi \| ^2)\, .
$$
This completes the proof of Theorem \ref{th2.1}.

\section{Proof of Theorem \ref{th1.4}}

\begin{lemma} \label{l4.1}
Let $d\geqslant 3$. Suppose $\gamma \in \Lambda \backslash \{ 0\} $, $\widehat W\in
L^2(K;{\mathcal M}_M)$, and $||| \, \widehat W\, ||| _{\gamma ,\, M}<+\infty $.
Then there is a number $c^{\, *}=c^{\, *}(\gamma )>0$ such that for all
$\varkappa \geqslant 0$, all vectors $k\in {\mathbb R}^d$ with $|(k,\gamma )|=\pi $,
and all vector functions $\varphi \in \widetilde H^1(K;{\mathbb C}^M)$ the inequality
$$
\| \, \widehat G_-^{\, -\frac 12}\, \widehat W\varphi \, \| \leqslant  c^{\, *}\, 
||| \, \widehat W\, ||| _{\gamma ,\, M}\, \| \, \widehat G_-^{\, \frac 12}
\varphi \, \| 
$$
holds.
\end{lemma}

\begin{proof}
From (0.11) (for $\varepsilon =1$) it follows that for all $\varkappa \geqslant 0$, 
all $k\in {\mathbb R}^d$, and all $\varphi \in \widetilde H^1(K;{\mathbb C}^M)$ the
estimate
$$
\| \widehat W\varphi \| \leqslant  
||| \, \widehat W\, ||| _{\gamma ,\, M}\, \bigl( v^{\, \frac 12}\, (K) 
\bigl( \, \sum\limits_{N\, \in \, \Lambda ^*}|k_{\| }+2\pi N_{\| }|^2\, \| \varphi
_N\| ^2\, \bigr) ^{\frac 12}+C^{\, \prime }(\gamma ,1)\, \| \varphi \| \,
\bigr) \leqslant  \eqno (4.1) 
$$ $$
\leqslant ||| \, \widehat W\, ||| _{\gamma ,\, M}\, \bigl( \, \| \, \widehat G_-
^{\, 1}\varphi \, \| +C^{\, \prime }(\gamma ,1)\, \| \varphi \| \, \bigr) 
$$
holds. For any vector function $\psi \in L^2(K;{\mathbb C}^M)$, we have $\widehat 
G_-^{\, -1}\psi \in \widetilde H^1(K;{\mathbb C}^M)$. Therefore, from (4.1)
(taking into account that $G^{\, -}_N(k;\varkappa )\geqslant \pi |\gamma |^{-1}$, 
$N\in \Lambda ^*$, if $|(k,\gamma )|=\pi $) for all $\varkappa \geqslant 0$, all 
$k\in {\mathbb R}^d$ with $|(k,\gamma )|=\pi $, and all $\psi \in L^2(K;{\mathbb 
C}^M)$ we obtain
$$
\| \,  \widehat W\, \widehat G_-^{\, -1}\, \psi \, \| \leqslant  c^{\, *}\, 
||| \, \widehat W\, ||| _{\gamma ,\, M}\, \| \psi \| \, ,  \eqno (4.2)
$$
where $c^{\, *}=1+\pi ^{-1}|\gamma |C^{\, \prime }(\gamma ,1)$. The same inequality
is fulfilled for the adjoint matrix function $\widehat W^{\, *}\in L^2(K;{\mathcal 
M}_M)$ (for which $||| \, \widehat W^{\, *}\, ||| _{\gamma ,\, M}=||| \, \widehat 
W\, ||| _{\gamma ,\, M}<+\infty $). Hence, for the operator $(\widehat W^{\, *}
\widehat G_-^{\, -1})^{\, *}$ which is adjoint to the operator $\widehat W^{\, *}
\widehat G_-^{\, -1}$, estimate (4.2) is also satisfied. Since $(\widehat W^{\, *}
\widehat G_-^{\, -1})^{\, *}\psi =\widehat G_-^{\, -1}\widehat W\psi $ for all
vector functions $\psi \in \widetilde H^1(K;{\mathbb C}^M)$, we conclude that for
such vector functions $\psi $ the following inequality is also valid:
$$
\| \, \widehat G_-^{\, -1}\, \widehat W\, \psi \, \| \leqslant  c^{\, *}\, 
||| \, \widehat W\, ||| _{\gamma ,\, M}\, \| \psi \| \, .  \eqno (4.3)
$$
Let now use the interpolation of the operators $\widehat W\widehat G_-^{\, -1}$ and
$\widehat G_-^{\, -1}\widehat W$ (see \cite{BL,RSII}). Consider the analytic
operator function $\zeta \to \widehat G_-^{\, \zeta }\widehat W\widehat G_-^{\, 
-1-\zeta }$ defined for $\zeta \in {\mathbb C}$ with $-1<{\mathrm {Re}}\ \zeta <0$. 
For a vector function $\psi \in \widetilde H^1(K;{\mathbb C}^M)$, the function 
$\zeta \to \widehat G_-^{\, \zeta }\widehat W\widehat G_-^{\, -1-\zeta }\psi 
\in L^2(K;{\mathbb C}^M)$ is also analytic. Moreover, it has continuous bounded
extension to the closed set $\{ \zeta \in {\mathbb C}:-1\leqslant {\mathrm 
{Re}}\ \zeta \leqslant 0\} $. If either ${\mathrm {Re}}\ \zeta =0$ or ${\mathrm 
{Re}}\ \zeta =-1$, then from (4.2) and (4.3) we get
$$
\| \, \widehat G_-^{\, \zeta }\, \widehat W\, \widehat G_-^{\, -1-\zeta }\psi  
\, \| \leqslant  c^{\, *}\, ||| \, \widehat W\, ||| _{\gamma ,\, M}\, \| \psi \| 
\, .  \eqno (4.4)
$$ 
Therefore estimate (4.4) is true for all $\zeta \in {\mathbb C}$ with $-1\leqslant 
{\mathrm {Re}}\ \zeta \leqslant 0$. In particular, for $\zeta =-\frac 12 $ (and 
for $\psi \in \widetilde H^1(K;{\mathbb C}^M)$) we obtain
$$
\| \, \widehat G_-^{\, -\frac 12}\, \widehat W\, \widehat G_-^{\, -\frac 12}\psi  
\, \| \leqslant  c^{\, *}\, ||| \, \widehat W\, ||| _{\gamma ,\, M}\, \| \psi \| 
\, .  \eqno (4.5)
$$
By continuity, inequality (4.5) holds for all vector functions $\psi \in \widetilde 
H^{\, \frac 12}(K;{\mathbb C}^M)$. Finally, since any vector function $\varphi \in 
\widetilde H^1(K;{\mathbb C}^M)$ can be represented in the form $\varphi =\widehat 
G_-^{\, -\frac 12}\psi $, where $\psi \in \widetilde H^{\, \frac 12}(K;{\mathbb 
C}^M)$, the inequality claimed in Lemma \ref{l4.1} is an immediate consequence of
inequality (4.5).
\end{proof}

{\it Proof} of Theorem \ref{th1.4}. Fix a number $\delta \in (0,1)$.
Let us choose a sufficiently large number $R>0$ such that
$$
||| \, \widehat W-\widehat W_{\{ R\} }\, ||| _{\gamma ,\, M}^{\, 2}\leqslant
\frac 1{18}\, \delta ^{\, 2}\, (c^{\, *})^{-2}\, C_2^{\, 2}\, ,
$$
where $c^{\, *}=c^{\, *}(\gamma )>0$ is the constant from Lemma \ref{l4.1}. The
functions $\widehat V^{\, (s)}_{\{ R\} }\, $, $s=0,1$, and $A_{\{ R\} }$ satisfy
the conditions of Theorem \ref{th1.3}, therefore there are numbers $\widetilde a
=\widetilde a(C_2;\frac {\delta }2,R)\in (0,C_2]$ and $\varkappa _0^{\, \prime }>0$ 
such that for all $\varkappa \geqslant \varkappa _0^{\, \prime }\, $, all vectors 
$k\in {\mathbb R}^d$ with $|(k,\gamma )|=\pi $, and all vector functions $\varphi 
\in \widetilde H^1(K;{\mathbb C}^M)$ the inequality  
$$
\| (\widehat P^{\, +}_*+\widetilde a\widehat P^{\, -}_*)(\widehat {\mathcal D}
(k+i\varkappa e)+\widehat W_{\{ R\} })\varphi \| ^2\geqslant  
(1-\frac {\delta }2\, )\, (C^{\, 2}_2\, \| \widehat G_-^{\, 1}\widehat P^{\, 
-}\varphi \| ^2+{\widetilde a}^{\, 2}\, \| \widehat G_+^{\, 1}\widehat P^{\, +}
\varphi \| ^2)  \eqno (4.6)
$$
holds. Instead of the vector $\gamma \in \Lambda \backslash \{ 0\} $, in conditions 
$\bf (A_1)$, $\bf (\widetilde A_1)$, and $\bf (A_2)$ we can pick the vector $-\gamma 
$. Under the replacement of the vector $\gamma $ by the vector $-\gamma $, the measure 
$\mu $ and the constants $C_2$ and $\widetilde a$ do not change. The number $\varkappa 
_0^{\, \prime }\, $ (which coincides with the number $\varkappa _0$ from Theorem 
\ref{th1.3} and is determined by the number $\frac {\delta }2$ and the function 
$\widehat W_{\{ R\} }$) do not change as well. Nevertheless, the orthogonal
projections $\widehat P^{\, +}$ and $\widehat P^{\, -}$ are replaced by the orthogonal
projections $\widehat P^{\, -}_*\, $ and $\widehat P^{\, +}_*\, $, respectively. 
Therefore, along with inequality (4.6), the inequality
$$
\| (\widehat P^{\, -}+\widetilde a\widehat P^{\, +})(\widehat {\mathcal D}
(k-i\varkappa e)+\widehat W_{\{ R\} })\varphi \| ^2\geqslant  
(1-\frac {\delta }2\, )\, (C^{\, 2}_2\, \| \widehat G_-^{\, 1}\widehat P^{\, 
+}_*\varphi \| ^2+{\widetilde a}^{\, 2}\, \| \widehat G_+^{\, 1}\widehat P^{\, -}_*
\varphi \| ^2)  \eqno (4.7)
$$
holds. Since $\widehat W\in {\mathbb L}^{\Lambda }_M(d;0)$ (hence, also $\widehat 
W_{\{ R\} }\in {\mathbb L}^{\Lambda }_M(d;0)$), the operators $\widehat {\mathcal 
D}(k\pm i\varkappa e)+\widehat W_{\{ R\} }$ are closed and from (4.6), (4.7) it
follows that their ranges $R(\widehat {\mathcal D}(k\pm i\varkappa e)+
\widehat W_{\{ R\} })$ are closed subspaces in $L^2(K;{\mathbb C}^M)$. We have
$$
(\widehat {\mathcal D}(k+i\varkappa e)+\widehat W_{\{ R\} })^{\, *}=\widehat 
{\mathcal D}(k-i\varkappa e)+\widehat W_{\{ R\} }\, .
$$
Consequently, from (4.6), (4.7) we see that
$$
{\mathrm {Ker}}\, (\widehat {\mathcal D}(k\pm i\varkappa e)+\widehat W_{\{ R\} }) 
={\mathrm {Coker}}\, (\widehat {\mathcal D}(k\pm i\varkappa e)+\widehat W_{\{ R\} 
})=\{ 0\} \, .
$$
The last equalities mean that the operators $\widehat {\mathcal D}(k\pm i\varkappa 
e)+\widehat W_{\{ R\} }$ (for $\varkappa \geqslant \varkappa _0^{\, \prime }\, $)
are bijective maps of $D(\widehat {\mathcal D}(k\pm i\varkappa e)+\widehat W_{\{ R\} 
})=\widetilde H^1(K;{\mathbb C}^M)$ onto $L^2(K;{\mathbb C}^M)$. From this, using 
(4.6), (4.7), for all $\varkappa \geqslant \varkappa _0^{\, \prime }\, $, all $k\in 
{\mathbb R}^d$ with $|(k,\gamma )|=\pi $, and all $\psi \in L^2(K;{\mathbb C}^M)$
we obtain  
$$
\| \, \bigl( C_2\, \widehat G_-^{\, 1}\, \widehat P^{\, -}+\widetilde a\, 
\widehat G_+^{\, 1}\, \widehat P^{\, +}\bigr) (\widehat {\mathcal D}
(k+i\varkappa e)+\widehat W_{\{ R\} })^{-1}\bigl( \widehat P^{\, +}_*+
\widetilde a^{-1}\, \widehat P^{\, -}_*\bigr) \psi \, \| ^2 \leqslant  \eqno (4.8)
$$ $$
\leqslant \bigl( \, 1-\frac {\delta }2\, \bigr) ^{-1}\, \| \psi \| ^2\, ,   
$$ $$
\| \, \bigl( C_2\, \widehat G_-^{\, 1}\, \widehat P^{\, +}_*+\widetilde a\, 
\widehat G_+^{\, 1}\, \widehat P^{\, -}_*\bigr) (\widehat {\mathcal D}
(k-i\varkappa e)+\widehat W_{\{ R\} })^{-1}\bigl( \widehat P^{\, -}+
\widetilde a^{-1}\, \widehat P^{\, +}\bigr) \psi \, \| ^2 \leqslant  \eqno (4.9)
$$ $$
\leqslant \bigl( \, 1-\frac {\delta }2\, \bigr) ^{-1}\, \| \psi \| ^2\, .   
$$ 
Estimate (4.9) is also valid for the adjoint operator
$$
\bigl( \bigl( C_2\, \widehat G_-^{\, 1}\, \widehat P^{\, +}_*+\widetilde a\, 
\widehat G_+^{\, 1}\, \widehat P^{\, -}_*\bigr) (\widehat {\mathcal D}
(k-i\varkappa e)+\widehat W_{\{ R\} })^{-1}\bigl( \widehat P^{\, -}+
\widetilde a^{-1}\, \widehat P^{\, +}\bigr) \bigr) ^{\, *}\, ,
$$
hence, for all $\psi \in \widetilde H^1(K;{\mathbb C}^M)$,
$$
\| \, \bigl( \widehat P^{\, -}+\widetilde a^{-1}\, \widehat P^{\, +}\bigr) 
(\widehat {\mathcal D}(k+i\varkappa e)+\widehat W_{\{ R\} })^{-1}\bigl( C_2\, 
\widehat G_-^{\, 1}\, \widehat P^{\, +}_*+\widetilde a\, \widehat G_+^{\, 1}\, 
\widehat P^{\, -}_*\bigr) \psi \, \| ^2 \leqslant  \eqno (4.10)
$$ $$
\leqslant \bigl( \, 1-\frac {\delta }2\, \bigr) ^{-1}\, \| \psi \| ^2\, .
$$ 
Let us now use the interpolation of operators (see \cite{BL,RSII} and also the 
proof of Lemma \ref{l4.1}). From (4.8) and (4.10) it follows that for all $\zeta \in 
{\mathbb C}$ with $0\leqslant {\mathrm {Re}}\ \zeta \leqslant 1$ (and all
$\psi \in \widetilde H^1(K;{\mathbb C}^M)$) the inequality
$$
\| \, \bigl( C_2^{\, \zeta }\, \widehat G_-^{\, \zeta }\, \widehat P^{\, -}+
\widetilde a^{\, 2\zeta -1}\, \widehat G_+^{\, \zeta }\, \widehat P^{\, +}\bigr) 
(\widehat {\mathcal D}(k+i\varkappa e)+\widehat W_{\{ R\} })^{-1}\bigl( C_2^{\, 1-
\zeta }\, \widehat G_-^{\, 1-\zeta }\, \widehat P^{\, +}_*+\widetilde a^{\, 1-
2\zeta }\, \widehat G_+^{\, 1-\zeta }\, \widehat P^{\, -}_*\bigr) \psi \, \| ^2 
\leqslant  
$$ $$
\leqslant \bigl( \, 1-\frac {\delta }2\, \bigr) ^{-1}\, \| \psi \| ^2
$$ 
is fulfilled (we also use the uniform boundedness on the closed set $\{ \zeta \in 
{\mathbb C}:0\leqslant {\mathrm {Re}}\ \zeta \leqslant 1\} $ of the vector function
from the left hand side of the last inequality; it follows easily from (4.8)). For 
$\zeta =\frac 12\, $, the last inequality has the form
$$
\| \, \bigl( C_2^{\, \frac 12}\, \widehat G_-^{\, \frac 12}\, \widehat P^{\, -}+
\widehat G_+^{\, \frac 12}\, \widehat P^{\, +}\bigr) 
(\widehat {\mathcal D}(k+i\varkappa e)+\widehat W_{\{ R\} })^{-1}\bigl( C_2
^{\, \frac 12}\, \widehat G_-^{\, \frac 12}\, \widehat P^{\, +}_*+\widehat G_+
^{\, \frac 12}\, \widehat P^{\, -}_*\bigr) \psi \, \| ^2 
\leqslant  \eqno (4.11)
$$ $$
\leqslant \bigl( \, 1-\frac {\delta }2\, \bigr) ^{-1}\, \| \psi \| ^2\, .
$$
By continuity (see (4.8)), estimate (4.11) is true for all vector functions
$\psi \in \widetilde H^{\, \frac 12}(K;{\mathbb C}^M)$. Since any vector function
$\varphi \in \widetilde H^1(K;{\mathbb C}^M)$ can be represented in the form
$$
\varphi =(\widehat {\mathcal D}(k+i\varkappa e)+\widehat W_{\{ R\} })^{-1}
\bigl( C_2^{\, \frac 12}\, \widehat G_-^{\, \frac 12}\, \widehat P^{\, +}_*+
\widehat G_+^{\, \frac 12}\, \widehat P^{\, -}_*\bigr) \psi \, ,
$$
where $\psi \in \widetilde H^{\, \frac 12}(K;{\mathbb C}^M)$, from (4.11) it follows  
that (for all $\varkappa \geqslant \varkappa _0^{\, \prime }\, $, all $k\in {\mathbb 
R}^d$ with $|(k,\gamma )|=\pi $, and all $\varphi \in \widetilde H^1(K;{\mathbb 
C}^M)$) the inequality  
$$
\| \, \bigl( C_2^{\, -\frac 12}\, \widehat G_-^{\, -\frac 12}\, \widehat P^{\, +}_*+
\widehat G_+^{\, -\frac 12}\, \widehat P^{\, -}_*\bigr) 
(\widehat {\mathcal D}(k+i\varkappa e)+\widehat W_{\{ R\} })\varphi \, \| ^2
\geqslant
$$ $$
\geqslant \bigl( \, 1-\frac {\delta }2\, \bigr) \, \| \bigl( C_2
^{\, \frac 12}\, \widehat G_-^{\, \frac 12}\, \widehat P^{\, -}+\widehat G_+
^{\, \frac 12}\, \widehat P^{\, +}\bigr) \varphi \, \| ^2 
$$
holds. Finally, using the estimate
$$
\| \, \widehat G_-^{\, -\frac 12}\, (\widehat W-\widehat W_{\{ R\} })\varphi \, 
\| ^2 \leqslant  (c^{\, *})^2\, ||| \, \widehat W-\widehat W_{\{ R\} }\, ||| 
_{\gamma ,\, M}^{\, 2}\, \| \, \widehat G_-^{\, \frac 12}\varphi \, \| ^2
\leqslant
$$ $$
\leqslant \frac 1{18}\ \delta ^{\, 2}\, C_2^{\, 2}\, \| \, \widehat G_-^{\, 
\frac 12}\varphi \, \| ^2 \leqslant \frac 1{18}\ \delta ^{\, 2}\, C_2^{\, 2}\, 
\bigl( \, \| \, \widehat G_-^{\, \frac 12}\, \widehat P^{\, -}\varphi \, \| ^2+ 
\| \, \widehat G_+^{\, \frac 12}\, \widehat P^{\, +}\varphi \, \| ^2\, \bigr)
$$
which is a consequence of Lemma \ref{l4.1}, for all $\varkappa \geqslant \varkappa 
_0^{\, \prime }\, $, all vectors $k\in {\mathbb R}^d$ with $|(k,\gamma )|=\pi $,
and all vector functions $\varphi \in \widetilde H^1(K;{\mathbb C}^M)$ we obtain
the inequality
$$
\| \, \bigl( C_2^{\, -\frac 12}\, \widehat G_-^{\, -\frac 12}\, \widehat P^{\, +}_*+
\widehat G_+^{\, -\frac 12}\, \widehat P^{\, -}_*\bigr) 
(\widehat {\mathcal D}(k+i\varkappa e)+\widehat W)\varphi \, \| ^2
\geqslant 
$$ $$
\geqslant \bigl( \, 1-\frac {\delta }3\, \bigr) \, \| \, \bigl( C_2^{\, -\frac 12}\, 
\widehat G_-^{\, -\frac 12}\, \widehat P^{\, +}_*+\widehat G_+^{\, -\frac 12}\, 
\widehat P^{\, -}_*\bigr) (\widehat {\mathcal D}(k+i\varkappa e)+\widehat W_{\{ R\} 
})\varphi \, \| ^2-
$$ $$
-\, \frac 3{\delta }\, \| \, \bigl( C_2^{\, -\frac 12}\, \widehat G_-^{\, -\frac 12}
\, \widehat P^{\, +}_*+\widehat G_+^{\, -\frac 12}\, \widehat P^{\, -}_*\bigr) 
(\widehat W-\widehat W_{\{ R\} })\varphi \, \| ^2 \geqslant
$$ $$
\geqslant \bigl( \, 1-\frac {5\delta }6\, \bigr) \, \| \, \bigl( C_2
^{\, \frac 12}\, \widehat G_-^{\, \frac 12}\, \widehat P^{\, -}+\widehat G_+
^{\, \frac 12}\, \widehat P^{\, +}\bigr) \varphi \, \| ^2 - \, \frac 3{\delta }\, 
C_2^{-1}\, \| \, \widehat G_-^{\, -\frac 12}\, (\widehat W-\widehat W_{\{ R\} })
\varphi \, \| ^2 \geqslant
$$ $$
\geqslant \bigl( \, 1-\frac {5\delta }6\, \bigr) \, \bigl( \, C_2\, \| \,
\widehat G_-^{\, \frac 12}\, \widehat P^{\, -}\varphi \, \| ^2+\| \, \widehat G_+
^{\, \frac 12}\, \widehat P^{\, +}\varphi \, \| ^2 \, \bigr) - \, \frac {\delta }6\, 
C_2\, \bigl( \, \| \, \widehat G_-^{\, \frac 12}\, \widehat P^{\, -}\varphi \, \| 
^2+\| \, \widehat G_+^{\, \frac 12}\, \widehat P^{\, +}\varphi \, \| ^2\, \bigr)
\geqslant
$$ $$
\geqslant (1-\delta )\, \bigl( \, C_2\, \| \, \widehat G_-^{\, \frac 12}\, 
\widehat P^{\, -}\varphi \, \| ^2+\| \, \widehat G_+^{\, \frac 12}\, \widehat 
P^{\, +}\varphi \, \| ^2 \, \bigr) \, .
$$
Theorem \ref{th1.4} is proved.

\section{Proof of Theorem \ref{th1.5}}

In this Section we shall use the modification of the method suggested in \cite{Izv06}.

Let $\# \, {\mathcal O}$ denote the number of elements of a finite set ${\mathcal 
O}$. By $d(K^*)$ denote the diameter of the fundamental domain $K^*$. For a
measurable matrix function $\widehat W:K\to {\mathcal M}_M$ and a number $a\geqslant 
0$, we write 
$$
\widehat W_{[a]}(x)=
\left\{
\begin{array}{ll}
\widehat W(x) & \text {if}\ \| \widehat W(x)\| \geqslant a\, , \\
\widehat 0\ & \text {otherwise}.
\end{array}
\right.
$$
If $\widehat W=(W_{pq})_{p,q=1,\dots ,M}\in L^2(K;{\mathcal M}_M)$, then
$$
\sum\limits_{N\, \in \, \Lambda ^*}\| \widehat W_N\| ^2\leqslant
\sum\limits_{N\, \in \, \Lambda ^*}\sum\limits_{p,\, q\, =\, 1}^M|(W_{pq})
_N| ^2=  \eqno (5.1)
$$ $$
=v^{-1}(K)\int_K\biggl( \, \sum\limits_{p,\, q\, =\, 1}^M|W_{pq}(x)| ^2 
\biggr) \, dx\leqslant
M\, v^{-1}(K)\int_K\| \widehat W(x)\| ^2\, dx\, .
$$

First we shall obtain some auxiliary statements which will be used in the proof
of Theorems \ref{th1.5} and \ref{th1.6}.

As above we suppose that the vector $\gamma \in \Lambda \backslash \{ 0\} $ is fixed, 
$e=|\gamma |^{-1}\gamma $. Choose a number $\Xi \geqslant 1$ (we set $\Xi >1$ in this
Section and $\Xi =1$ in Section 6). Let $\widehat V=\widehat V^{\, (0)}+\widehat 
V^{\, (1)}$, where $\widehat V^{\, (s)}\in L^3_w(K;{\mathcal S}_M^{(s)})$, $s=0,1$. 
For the Fourier coefficients $\widehat V^{\, (s)}_N\, $, $N\in \Lambda ^*$, of the 
matrix functions $\widehat V^{\, (s)}$, $s=0,1$, we have $\widehat V^{\, (s)}_N \in 
{\mathcal S}_M^{(s)}\, $. Hence, the Fourier coefficients $\widehat V^{\, (s)}_N$ (as
operators acting on ${\mathbb C}^M$) commute with all orthogonal projections $\widehat 
P^{\, \pm}_{\widetilde e}\, $, $\widetilde e\in S_1(e)$. The same assertion is true 
for the Fourier coefficients $\widehat V_N\, $, $N\in \Lambda ^*$.

Let us choose (and fix) a number $h\in {\mathbb R}$ for which $h\geqslant 64$ and 
$h>2\pi d(K^*)$. The number $\varkappa _0>0$ is chosen sufficiently large (lower
estimates for the number $\varkappa _0$ are specified in the sequel). To start with, 
we assume that $\varkappa _0>2h$. Let $\varkappa _0\leqslant \varkappa _1\leqslant 
\varkappa \leqslant \Xi \, \varkappa _1$ (if we set $\Xi =1$, then $\varkappa _1
=\varkappa $), $l=l(h,\varkappa _1)\in {\mathbb N}$ is the largest integer with 
$2h^{\, l}\leqslant \varkappa _1$ (hence, $\varkappa _1<2h^{\, l+1}$ and $2h^{\, l}
\leqslant \varkappa <2\, \Xi \, h^{\, l+1}$).

In what follows, the vector $k\in {\mathbb R}^3$ is assumed to satisfy the condition 
$|(k,\gamma )|=\pi $. For any $b\in [0,\varkappa )$, we write
$$
{\mathcal K}(b)=\{ N\in \Lambda ^*:G_N(k;\varkappa )\leqslant b\} 
$$
(the sets ${\mathcal K}(b)$ depend also on $k$ and $\varkappa $). For $b>2\pi d(K^*)$ 
(and $b<\varkappa $) the following estimate holds: 
$$
\# \, {\mathcal K}(b)\leqslant c_2\, v^{-1}(K^*)\, \varkappa b^2\, ,
$$
where $c_2>0$ is a universal constant. If $N\in {\mathcal K}(b)$, then $k_{\perp}+
2\pi N_{\perp} \neq 0$. If $N\in {\mathcal K}(h^{\, l})$, then 
$$
|\varkappa -|k_{\perp}+2\pi N_{\perp}||\leqslant h^{\, l}\leqslant \frac 
{\varkappa}2\, , 
$$
hence, $|k_{\perp}+2\pi N_{\perp}|\geqslant \frac {\varkappa}2>1$. From this it 
follows that for all $N, N^{\, \prime}\in {\mathcal K}(h^{\, l})$, 
$$
|\widetilde e(k+2\pi N)-\widetilde e(k+2\pi N^{\, \prime})|\leqslant \frac {4\pi }{\varkappa }
\, |N_{\perp}-N^{\, \prime}_{\perp}|\, .  \eqno (5.2)
$$

We introduce the notation
$$
{\mathcal K}_1\doteq {\mathcal K}(h)\, ,\ \ \ 
{\mathcal K}_{\mu}\doteq \{ N\in \Lambda ^*:h^{\, \mu -1}<G_N(k;\varkappa )
\leqslant h^{\, \mu }\} \, ,\ \mu =2,\dots ,l\, .
$$
For any vector function $\varphi \in L^2(K;{\mathbb C}^M)$,
$$
\sqrt {\frac {\pi}{|\gamma |}}\ \| \, \widehat P^{\, {\mathcal K}_1}\varphi \, 
\| \leqslant \| \, \widehat G_-^{\, \frac 12}\, \widehat P^{\, {\mathcal K}_1}
\varphi \, \| \leqslant h^{\, \frac 12}\, \| \, \widehat P^{\, {\mathcal K}_1}
\varphi \, \| \, ,  \eqno (5.3)
$$ $$
h^{\, \frac {\mu -1}2}\, \| \, \widehat P^{\, {\mathcal K}_{\mu}}\varphi \, \| 
\leqslant \| \, \widehat G_-^{\, \frac 12}\, \widehat P^{\, {\mathcal K}_{\mu}}\varphi 
\, \| \leqslant h^{\, \frac {\mu}2}\, \| \, \widehat P^{\, {\mathcal K}_{\mu}}
\varphi \, \| \, ,\ \mu =2,\dots ,l\, .  
\eqno (5.4)
$$

Let us denote
$$
\widehat P^{\, (\pm )}_{\mu}\doteq \widehat P^{\, \pm }(k;e)\, \widehat P^{\, 
{\mathcal K}_{\mu}}\, ,\ \mu =1,\dots ,l\, ,\ \ \ \widehat P^{\, (\pm )}\doteq 
\widehat P^{\, \pm }(k;e)\, \widehat P^{\, {\mathcal K}(h^{\, l})}=
\sum\limits_{\mu \, =\, 1}^l\widehat P^{\, (\pm )}_{\mu}\, .
$$
The equality $\widehat P^{\, (+)}+\widehat P^{\, (-)}=\widehat P^{\, {\mathcal K}
(h^{\, l})}$ holds. 

For all $\varkappa >0$ and all $N\in \Lambda ^*$, we get
$$
G^+_N(k;\varkappa )>|k+2\pi N|\geqslant |k_{\parallel}+2\pi N_{\parallel}|\geqslant
\pi |\gamma |^{-1}\, .
$$
Using the inequality $\varkappa <2\, \Xi \, h^{\, l+1}$, for all $N\in \Lambda ^*\, 
\backslash \, {\mathcal K}(h^{\, l})$, we obtain 
$$
G^-_N(k;\varkappa )>h^{\, l}(2\varkappa +h^{\, l})^{-1}G^+_N(k;\varkappa )>
(4\, \Xi \, h+1)^{-1}G^+_N(k;\varkappa )\, ,  \eqno (5.5)
$$ $$
G^-_N(k;\varkappa )>h^{\, l}(\varkappa +h^{\, l})^{-1}|k+2\pi N|>
(2\, \Xi \, h+1)^{-1}|k+2\pi N|\, .  \eqno (5.6)
$$

\begin{lemma} \label{l5.1}
Suppose $\widehat W\in L^2(K;{\mathcal M}_M)$. Then for any finite set ${\mathcal O}
\subset \Lambda ^*$ and any vector function $\varphi \in L^2(K;{\mathbb C}^M)$ we 
have
$$
\| \widehat W\widehat P^{\, {\mathcal O}}\varphi \| \leqslant v^{-\frac 12}(K)\,
(\# \, {\mathcal O})^{\frac 12}\, \| \widehat W\| _{L^2(K;{\mathcal M}_M)}
\, \| \widehat P^{\, {\mathcal O}}\varphi \| \, .
$$
\end{lemma}

\begin{proof}
Indeed,
$$
\| \widehat W\widehat P^{\, {\mathcal O}}\varphi \| \leqslant \| \widehat  
W\| _{L^2(K;{\mathcal M}_M)}\, \| \widehat P^{\, {\mathcal O}}\varphi \| 
_{L^{\infty}(K;{\mathbb C}^M)}\leqslant
\| \widehat W\| _{L^2(K;{\mathcal M}_M)}\, \bigl( \, 
\sum\limits_{N\, \in \, {\mathcal O}}\| \varphi _N\| \, \bigr) \leqslant
$$ $$
\leqslant (\# \, {\mathcal O})^{\frac 12}\, \| \widehat W\| _{L^2(K;{\mathcal 
M}_M)}\, \bigl( \, \sum\limits_{N\, \in \, {\mathcal O}}\| \varphi _N\| ^2\, \bigr)
^{\frac 12}=v^{-\frac 12}(K)\, (\# \, {\mathcal O})^{\frac 12}\, \| \widehat 
W\| _{L^2(K;{\mathcal M}_M)}\, \| \widehat P^{\, {\mathcal O}}\varphi \| \, .
$$
\end{proof}

\begin{lemma} \label{l5.2}
Suppose $\widehat W\in L^3_w(K;{\mathcal M}_M)$. Then there exists a nonincreasing
function $[0,+\infty )\ni t\to f_{\widehat W}(t)\in [0,+\infty )$ such that 
$$
f_{\widehat W}(t)\to \| \widehat W \| ^{(\infty ,\, {\mathrm {loc}})}_{L^3_w(K;
{\mathcal M}_M)}
$$ 
as $t\to +\infty $, and for any finite set ${\mathcal O}\subset \Lambda ^*$ and any
vector function $\varphi \in L^2(K;{\mathbb C}^M)$ the inequality
$$
\| \widehat W\widehat P^{\, {\mathcal O}}\varphi \| \leqslant 4v^{-\frac 13}(K)\,
(\# \, {\mathcal O})^{\frac 13}\, f_{\widehat W}(\# \, {\mathcal O})\, 
\| \widehat P^{\, {\mathcal O}}\varphi \| 
$$
holds.
\end{lemma}

\begin{proof}
For $a>0$, we get
$$
\| \widehat W_{[a]} \| ^2_{L^2(K;{\mathcal M}_M)}=  
\, \sum\limits_{\nu \, =\, 1}^{+\infty }\, \int_{\{ \, x\, \in \, K\, :\, 2^{\, \nu
-1}a\, \leqslant \, \| \widehat W(x)\, \| \, <\, 2^{\, \nu }a\, \} }\| \widehat W(x) 
\| ^2\, dx\leqslant
$$ $$
\leqslant \frac 8a\ \sum\limits_{\nu \, =\, 1}^{+\infty }\, 2^{-\nu }\,
\| \widehat W_{[2^{\, \nu -1}a]}\| ^3_{L^3_w(K;{\mathcal M}_M)}\leqslant
\frac 8a\ \| \widehat W_{[a]}\| ^3_{L^3_w(K;{\mathcal M}_M)}\, .
$$
Hence, for all vector functions $\varphi \in L^2(K;{\mathbb C}^M)$, 
$$
\| \, \widehat W\, \widehat P^{\, {\mathcal O}}\varphi \, \| \leqslant \| \, 
(\widehat W-\widehat W_{[a]})\, \widehat P^{\, {\mathcal O}}\varphi \, \| +\| 
\widehat W_{[a]}\, \widehat P^{\, {\mathcal O}}\varphi \, \| \leqslant  \eqno (5.7)
$$ $$
\leqslant (a+v^{-\frac 12}(K)\, (\# \, {\mathcal O})^{\frac 12}\, \| \widehat 
W_{[a]} \| _{L^2(K;{\mathcal M}_M)})\, \| \, \widehat P^{\, {\mathcal O}}\varphi \, 
\| \leqslant
$$ $$
\leqslant \bigl( \, a+\frac {2\sqrt 2}{\sqrt a}\ v^{-\frac 12}(K)\, (\# \, 
{\mathcal O})^{\frac 12}\, \| \widehat W_{[a]} \| ^{\, \frac 32}
_{L^3(K;{\mathcal M}_M)}\, \bigr) \, \| \widehat P^{\, {\mathcal O}}\varphi \|
$$
(see Lemma \ref{l5.1}). Fix a number $\varepsilon >0$. We can represent the 
fundamental domain $K$ in the form $\bigcup_{j\, =\, 1}^JK^{\, (\varepsilon )}_j$, 
where $J\in {\mathbb N}$, and $K^{\, (\varepsilon )}_j$ are disjoint measurable sets 
of sufficiently small diameters $d(K^{\, (\varepsilon )}_j)$ such that for all $j$ 
the following inequalities are fulfilled:
$$
\| \, \chi _{K^{\, (\varepsilon )}_j}\, \widehat W \, \| ^{(\infty )}_{L^3_w(K;
{\mathcal M}_M)}<\frac {\varepsilon }2+\| \widehat W \| ^{(\infty ,\, {\mathrm 
{loc}})}_{L^3_w(K;{\mathcal M}_M)}
$$
(here $\chi _T$ is the characteristic function of a set $T\subseteq K$). For every
$j=1,\dots ,J$, we choose a number $a_j>0$ such that
$$
\| \, \chi _{K^{\, (\varepsilon )}_j}\, \widehat W_{[a_j]} \, \| _{L^3_w(K;
{\mathcal M}_M)}<\frac {\varepsilon }2+\| \, \chi _{K^{\, (\varepsilon )}_j}\, 
\widehat W \, \| ^{(\infty )}_{L^3_w(K;{\mathcal M}_M)}<\varepsilon +\| \widehat 
W \| ^{(\infty ,\, {\mathrm {loc}})}_{L^3_w(K;{\mathcal M}_M)}\, .
$$
From (5.7), for $a\geqslant \max\limits_ja_j\, $, we derive
$$
\| \, \widehat W\, \widehat P^{\, {\mathcal O}}\varphi \, \| ^2=
\sum\limits_j\, \| \, \chi _{K^{\, (\varepsilon )}_j}\, \widehat W\, \widehat P^{\, 
{\mathcal O}}\varphi \, \| ^2\leqslant  \eqno (5.8)
$$ $$
\leqslant 2\, \sum\limits_j \, \bigl( \, a^2+\frac 8a\ v^{-1}(K)\, (\# \, 
{\mathcal O})\ \| \, \chi _{K^{\, (\varepsilon )}_j}\, \widehat W_{[a_j]} \, \| 
_{L^3_w(K;{\mathcal M}_M)}^{\, 3}\, \bigr) \, \| \, \chi _{K^{\, 
(\varepsilon )}_j}\, \widehat P^{\, {\mathcal O}}\varphi \, \| ^2\leqslant
$$ $$
\leqslant 2\, \bigl( \, a^2+\frac 8a\ v^{-1}(K)\, (\# \, 
{\mathcal O})\ \bigl( \, \varepsilon +\| \widehat W\| _{L^3_w(K;{\mathcal M}_M)}
^{(\infty ,\, {\mathrm {loc}})}\, \bigr) ^3 \bigr) \, \sum\limits_j\| \, \chi 
_{K^{\, (\varepsilon )}_j}\, \widehat P^{\, {\mathcal O}}\varphi \, \| ^2\, .
$$
If
$$
\# \, {\mathcal O}\, \geqslant \, \frac 18\, v(K)\, \bigl( \, \varepsilon +\| 
\widehat W\| _{L^3_w(K;{\mathcal M}_M)}^{(\infty ,\, {\mathrm {loc}})}\, \bigr) 
^{-3}\, \max\limits_j a_j^3\, ,  \eqno (5.9)
$$
then we put
$$
a=2\, v^{-\frac 13}(K)\, (\# \, {\mathcal O})^{\frac 13}\, \bigl( \, \varepsilon +\| 
\widehat W\| _{L^3_w(K;{\mathcal M}_M)}^{(\infty ,\, {\mathrm {loc}})}\, \bigr) \, .
$$
Therefore inequality (5.8) (under condition (5.9)) implies that
$$
\| \, \widehat W\, \widehat P^{\, {\mathcal O}}\varphi \, \| \leqslant 4v^{-\frac 
13}(K)\, (\# \, {\mathcal O})^{\frac 13}\, \bigl( \, \varepsilon +\| \widehat W\| 
_{L^3_w(K;{\mathcal M}_M)}^{(\infty ,\, {\mathrm {loc}})}\, \bigr) \, \| \widehat 
P^{\, {\mathcal O}}\varphi \| \, .  \eqno (5.10)
$$
Since the number $\varepsilon >0$ can be chosen as small as we wish, the inequality
claimed in Lemma \ref{l5.2} is a consequence of (5.10).
\end{proof}

\begin{lemma} \label{l5.3}
There is a number $c_3=c_3(h)>0$ such that for any $\varepsilon >0$ and any matrix
function $\widehat V\in L^3_w(K;{\mathcal M}_M)$, for which
$$
\| \widehat V \| _{L^3_w(K;{\mathcal M}_M)}^{(\infty ,\, {\mathrm {loc}})}
\leqslant c_3\, \varepsilon \, ,  \eqno (5.11)
$$
there exists a number $\varkappa _0^{\, \prime}\, (\varepsilon )=\varkappa _0^{\, 
\prime}\, (\varepsilon ; \Lambda ,\gamma ,h,\widehat V)>2h$ such that for all 
$\varkappa _1\geqslant \varkappa _0^{\, \prime}\, (\varepsilon )$, all $\varkappa \in
[\varkappa _1 \, ,\Xi \, \varkappa _1]$, all vectors $k\in {\mathbb R}^3$ with
$|(k,\gamma )|=\pi $, and all vector functions $\varphi \in L^2(K;{\mathbb C}^M)$ 
the inequality
$$
\| \, \widehat G^{-\frac 12}_+\, \widehat V\, \widehat P^{\, {\mathcal K}(h^{\,
l})}\varphi \, \| \leqslant \varepsilon \, \| \, \widehat G_-^{\, \frac 12}\,
\widehat P^{\, {\mathcal K}(h^{\, l})}\varphi \, \|  \eqno (5.12)
$$
holds.
\end{lemma}

\begin{proof}
Taking into account the estimates $\# \, {\mathcal K}_{\mu}\leqslant c_2\, v^{-1}
(K^*)\, \varkappa h^{\, 2\mu }$, $\mu =1,\dots ,l$, and
$$
h^{\, \frac 12}\, \| \, \widehat P^{\, {\mathcal K}_1}\varphi \, \| \leqslant 
\sqrt {\frac {|\gamma |h}{\pi}}\ \| \, \widehat G_-^{\, \frac 12}\, \widehat P
^{\, {\mathcal K}_1}\varphi \, \| \, , \ \ 
h^{\, \frac {\mu}2}\, \| \, \widehat P^{\, {\mathcal K}_{\mu}}\varphi \, \| 
\leqslant h^{\, \frac 12}\, \| \, \widehat G_-^{\, \frac 12}\, \widehat P^{\, 
{\mathcal K}_{\mu}}\varphi \, \| \, ,\ \mu=2,\dots ,l\, ,
$$ 
and using Lemma \ref{l5.2}, for all vector functions $\varphi \in L^2(K;{\mathbb 
C}^M)$ we obtain
$$
\| \, \widehat G^{-\frac 12}_+\, \widehat V\, \widehat P^{\, {\mathcal K}(h^{\, l})}
\varphi \, \| \leqslant \varkappa ^{-\frac 12}\, \sum\limits_{\mu \, =\, 1}^l\, \| 
\, \widehat V\, \widehat P^{\, {\mathcal K}_{\mu}}\varphi \, \| \leqslant  \eqno (5.13)
$$ $$
\leqslant 4\, \varkappa ^{-\frac 12}\, v^{-\frac 13}(K)\, \sum\limits_{\mu \, =\, 1}
^l\, (c_2\, v^{-1}(K^*)\, \varkappa h^{\, 2\mu })^{\frac 13}\, f_{\widehat V}(c_2
\, v^{-1}(K^*)\, \varkappa h^{\, 2\mu })\, \| \, \widehat P^{\, {\mathcal K}_{\mu}}
\varphi \| \leqslant
$$ $$
\leqslant 4\, c_2^{\, \frac 13}\, f_{\widehat V}(c_2\, v^{-1}(K^*)\, \varkappa 
h^{\, 2})\, \varkappa ^{-\frac 16}\, \sum\limits_{\mu \, =\, 1}^l\, h^{\, \frac 
{\mu}6}\, (h^{\, \frac {\mu}2}\, \| \, \widehat P^{\, {\mathcal K}_{\mu}}\varphi \, 
\| ) \leqslant 
$$ $$
\leqslant 4\, c_2^{\, \frac 13}\, h^{\, \frac 12}\, \biggl(
\, \sqrt {\frac {|\gamma |}{\pi }}\, \varkappa ^{-\frac 16}\, h^{\, \frac 16}+
\varkappa ^{-\frac 16}\, \sum\limits_{\mu \, =\, 2}^l\, h^{\, \frac {\mu}6}\,
\biggr) \, f_{\widehat V}(c_2\, v^{-1}(K^*)\, \varkappa h^{\, 2})\, \| \, \widehat 
G_-^{\, \frac 12}\, \widehat P^{\, {\mathcal K}(h^{\, l})}\varphi \, \| \, .
$$
Since
$$
\varkappa ^{-\frac 16}\, \sum\limits_{\mu \, =\, 2}^l\, h^{\, \frac {\mu}6}=
(\varkappa ^{-1}h^{\, l})^{\frac 16}\, \sum\limits_{\nu \, =\, 0}^{l-2}h^{\, 
-\frac {\nu}6}<2^{\, \frac 56}<2
$$
and 
$$
f_{\widehat V}(c_2\, v^{-1}(K^*)\, \varkappa h^{\, 2})\to \| \widehat V \| 
^{(\infty ,\, {\mathrm {loc}})}_{L^3_w(K;{\mathcal M}_M)}
$$ 
as $\varkappa \to +\infty $, estimate (5.12) follows from (5.13) for
$$
c_3=\frac 18\ c_2^{-\frac 13}\, h^{-\frac 12}
$$
and under the choice of a sufficiently large number $\varkappa ^{\, \prime }_0
(\varepsilon )$.
\end{proof}

In the sequel, we fix a number $\varepsilon ^{\, \prime }>0$ and put $\varepsilon =
\varepsilon ^{\, \prime }/3$, $\varkappa _0\geqslant \varkappa ^{\, \prime }_0
(\varepsilon )$. We assume that condition (5.11) is fulfilled for the function 
$\widehat V\in L^3_w(K;{\mathcal M}_M)$. From (5.12) it follows that
$$
\| \, \widehat G^{-\frac 12}_+\, \widehat P^{\, (-)}\, \widehat V\widehat P^{\, (-)}
\varphi \, \| \leqslant \varepsilon \, \| \, \widehat G_-^{\, \frac 12}\, \widehat 
P^{\, (-)}\varphi \, \| \, ,  \eqno (5.14)
$$ $$
\| \, \widehat G^{-\frac 12}_+\, \widehat P^{\, (-)}\, \widehat V\, \widehat P^{\, 
(+)}\varphi \, \| \leqslant \varepsilon \, \| \, \widehat G_-^{\, \frac 12}\, 
\widehat P^{\, (+)}\varphi \, \| \leqslant \varepsilon \, \| \, \widehat G^{\, 
\frac 12}_+\widehat P^{\, (+)}\varphi \, \| \, .  \eqno (5.15)
$$

We shall also assume that $\varkappa _0\geqslant \varkappa _0^{\, \prime}\, 
((4\, \Xi \, h+1)^{-\frac 12}\, \varepsilon )$. Then (5.5) and (5.12) yield
$$
\| \, \widehat G_-^{-\frac 12}\, \widehat P^{\, \Lambda ^*\backslash {\mathcal K}
(h^{\, l})}\, \widehat V\, \widehat P^{\, (-)}\varphi \, \| \leqslant  \eqno (5.16)
$$ $$
\leqslant (4\, \Xi \, h+1)^{\frac 12}\, \| \, \widehat G^{-\frac 12}_+\, \widehat 
V\, \widehat P^{\, (-)}\varphi \, \| \leqslant \varepsilon \, \| \, \widehat G_-
^{\, \frac 12}\, \widehat P^{\, (-)}\varphi \, \| \, ,
$$ $$
\| \, \widehat G_-^{-\frac 12}\, \widehat P^{\, \Lambda ^*\backslash {\mathcal K}
(h^{\, l})}\, \widehat V\, \widehat P^{\, (+)}\varphi \, \| \leqslant  \eqno (5.17)
$$ $$
\leqslant (4\, \Xi \, h+1)^{\frac 12}\, \| \, \widehat G^{-\frac 12}_+\, \widehat 
V\, \widehat P^{\, (+)}\varphi \, \| \leqslant \varepsilon \, \| \, \widehat G_-
^{\, \frac 12}\, \widehat P^{\, (+)}\varphi \, \| \leqslant \varepsilon \, \| \, 
\widehat G_+^{\, \frac 12}\, \widehat P^{\, (+)}\varphi \, \| \, .
$$
If we put $\varphi =\widehat G_-^{-\frac 12}\, \psi $, where $\psi \in \widetilde 
H^{\, \frac 12}(K;{\mathbb C}^M)$, then (5.12) implies that
$$
\| \, \widehat G^{-\frac 12}_+\, \widehat V\, \widehat P^{\, {\mathcal K}(h^{\, l})}\,
\widehat G_-^{-\frac 12}\psi \, \| \leqslant \varepsilon \, \| \, P^{\, {\mathcal 
K}(h^{\, l})}\psi \, \| \leqslant \varepsilon \, \| \psi \| \, .
$$
By continuity, the operator $\widehat G^{-\frac 12}_+\, \widehat V \, \widehat P^{\, 
{\mathcal K}(h^{\, l})}\, \widehat G_-^{-\frac 12}$ having the domain 
$$
D(\widehat G^{-\frac 12}_+\, \widehat V \, \widehat P^{\, {\mathcal K}(h^{\, l})}\, 
\widehat G_-^{-\frac 12})=\widetilde H^{\, \frac 12}(K;{\mathbb C}^M)
$$ 
extends to the bounded operator $\widehat G^{-\frac 
12}_+\, \widehat V \, \widehat G_-^{-\frac 12}\, \widehat P^{\, {\mathcal K}(h^{\, 
l})}$ acting on the space $L^2(K;{\mathbb C}^M)$. Therefore, for the adjoint 
operator that takes a vector function $\psi \in \widetilde 
H^{\, \frac 12}(K;{\mathbb C}^M)$ to the vector function $\widehat G_-^{-\frac 12}\, 
\widehat P^{\, {\mathcal K}(h^{\, l})}\, \widehat V\, \widehat G^{-\frac 12}_+\psi $, 
the estimate
$$
\| \, \widehat G_-^{-\frac 12}\, \widehat P^{\, {\mathcal K}(h^{\, l})}\, \widehat V\,
\widehat G^{-\frac 12}_+\psi \, \| \leqslant \varepsilon \, \| \psi \| \, ,\ \ \psi 
\in \widetilde H^{\, \frac 12}(K;{\mathbb C}^M)\, ,  \eqno (5.18)
$$
is fulfilled as well. If we put $\psi =\widehat G^{\, \frac 12}_+\, \varphi $, where 
$\varphi \in \widetilde H^1(K;{\mathbb C}^M)$, then (5.18) yields
$$
\| \, \widehat G_-^{-\frac 12}\, \widehat P^{\, {\mathcal K}(h^{\, l})}\, \widehat 
V\varphi \, \| \leqslant \varepsilon \, \| \, \widehat G^{\, \frac 12}_+\varphi \, 
\| \, .  \eqno (5.19)
$$
In particular, for any vector function $\varphi \in L^2(K;{\mathbb C}^M)$, 
$$
\| \, \widehat G_-^{-\frac 12}\, \widehat P^{\, (+)}\, \widehat V\, \widehat P^{\, (+)}
\varphi \, \| \leqslant \varepsilon \, \| \, \widehat G^{\, \frac 12}_+\, \widehat 
P^{\, (+)}\varphi \, \| \, .  \eqno (5.20)
$$

Since we assume that $\varkappa _0\geqslant \varkappa _0^{\, \prime}\, ((4\, \Xi \,  
h+1)^{-\frac 12}\, \varepsilon )$, inequality (5.19) is also true under the change
$\varepsilon $ to $(4\, \Xi \, h+1)^{-\frac 12}\, \varepsilon $. Therefore (see 
(5.5)) for any vector function $\varphi \in \widetilde H^1(K;{\mathbb C}^M)$ we get
$$
\| \, \widehat G_-^{-\frac 12}\, \widehat P^{\, (+)}\, \widehat V\, \widehat 
P^{\, \Lambda ^*\backslash {\mathcal K}(h^{\, l})}\varphi \, \| \leqslant  \eqno (5.21)
$$ $$
\leqslant (4\, \Xi \, h+1)^{-\frac 12}\, \varepsilon \, \| \, \widehat G^{\, 
\frac 12}_+\, \widehat P^{\, \Lambda ^*\backslash {\mathcal K}(h^{\, l})}\varphi \, 
\| \leqslant \varepsilon \, \| \, \widehat G_-^{\, \frac 12}\, \widehat P^{\, \Lambda 
^*\backslash {\mathcal K}(h^{\, l})}\varphi \, \| \, ,
$$ $$
\| \, \widehat G^{-\frac 12}_+\, \widehat P^{\, (-)}\, \widehat V\, \widehat P^{\, 
\Lambda ^*\backslash {\mathcal K}(h^{\, l})}\varphi \, \| \leqslant  
\| \, \widehat G^{-\frac 12}_-\, \widehat P^{\, (-)}\, \widehat V\, \widehat P^{\, 
\Lambda ^*\backslash {\mathcal K}(h^{\, l})}\varphi \, \| \leqslant  \eqno (5.22)
$$ $$
\leqslant (4\, \Xi \, h+1)^{-\frac 12}\, \varepsilon \, \| \, \widehat G^{\frac 12}_+
\, \widehat P^{\, \Lambda ^*\backslash {\mathcal K}(h^{\, l})}\varphi \, \| \leqslant 
\varepsilon \, \| \, \widehat G^{\frac 12}_-\, \widehat P^{\, \Lambda ^*\backslash 
{\mathcal K}(h^{\, l})}\varphi \, \| \, .
$$

Now, suppose that the condition
$$
\| \widehat V \| ^{(\infty ,\, {\mathrm {loc}})}_{L^3_w(K;{\mathcal M}_M)}
<C^{\, -1}\, (2\, \Xi \, h+1)^{-1}\, \varepsilon  \eqno (5.23)
$$
is satisfied along with condition (5.11), where $C=C(3)>0$ is the constant from 
(0.6). From (0.6), (1.1), and (5.6) it follows that under the choice of a 
sufficiently large number $\varkappa _0^{\, \prime \prime }>2h$ (dependent on 
$\varepsilon $, $\Lambda $, $\gamma $, $\widehat V$, $\Xi $, $h$) and for all 
$\varkappa _1\geqslant \varkappa _0 \geqslant \varkappa _0^{\, \prime \prime }\, $, 
all $\varkappa \in [\varkappa _1\, , \Xi \varkappa _1]$, all vectors $k\in {\mathbb 
R}^3$ with $|(k,\gamma )|=\pi $, and all vector functions $\psi \in L^2(K;{\mathbb 
C}^M)$ the inequality
$$
\| \, \widehat V\, \widehat G_-^{-1}\, \widehat P^{\, \Lambda ^*\backslash {\mathcal
K}(h^{\, l})}\psi \, \| \leqslant  \varepsilon \, \| \, \widehat P^{\, \Lambda ^*
\backslash {\mathcal K}(h^{\, l})}\psi \, \|  \eqno (5.24)
$$
holds (and $\widehat G_-^{-1}\psi \in \widetilde H^1(K;{\mathbb C}^M)$). Furthermore,
$$
(\widehat V\, \widehat G_-^{-1}\, \widehat P^{\, \Lambda ^*\backslash {\mathcal
K}(h^{\, l})})^{\, *}\, \widehat P^{\, \Lambda ^*\backslash {\mathcal K}(h^{\, 
l})}\, \psi = \widehat G_-^{-1}\, \widehat P^{\, \Lambda ^*\backslash {\mathcal
K}(h^{\, l})}\, \widehat V\, \widehat P^{\, \Lambda ^*\backslash {\mathcal K}(h^{\, 
l})}\, \psi
$$
for the adjoint operator $(\widehat V\, \widehat G_-^{-1}\, \widehat P^{\, \Lambda 
^*\backslash {\mathcal K}(h^{\, l})})^{\, *}$ and all vector functions $\psi \in
\widetilde H^1(K;{\mathbb C}^M)$, whence
$$
\| \, \widehat G_-^{-1}\, \widehat P^{\, \Lambda ^*\backslash {\mathcal
K}(h^{\, l})}\, \widehat V\, \widehat P^{\, \Lambda ^*\backslash {\mathcal K}(h^{\, 
l})}\, \psi \, \| \leqslant \varepsilon \, \| \, \widehat P^{\, \Lambda ^*\backslash 
{\mathcal K}(h^{\, l})}\, \psi \, \| \, .  \eqno (5.25)
$$
Using the interpolation of operators (see \cite{BL,RSII} and also the proof of Lemma 
\ref{l4.1}), from (5.24) and (5.25), for all vector functions $\psi \in \widetilde 
H^{\, \frac 12}(K;{\mathbb C}^M)$ we derive
$$
\| \, \widehat G_-^{-\frac 12}\, \widehat P^{\, \Lambda ^*\backslash {\mathcal K}
(h^{\, l})}\, \widehat V\, \widehat G_-^{-\frac 12}\, \widehat P^{\, \Lambda ^*
\backslash {\mathcal K}(h^{\, l})}\, \psi \, \| \leqslant \varepsilon \, \| \, 
\widehat P^{\, \Lambda ^*\backslash {\mathcal K}(h^{\, l})}\, \psi \, \| \, .
$$
If we put $\varphi =\widehat G_-^{-\frac 12}\, \psi \in \widetilde H^1(K;{\mathbb 
C}^M)$, then the previous inequality yields 
$$
\| \, \widehat G_-^{-\frac 12}\, \widehat P^{\, \Lambda ^*\backslash {\mathcal K}
(h^{\, l})}\, \widehat V\, \widehat P^{\, \Lambda ^*\backslash {\mathcal K}(h^{\, 
l})}\, \varphi \, \| \leqslant \varepsilon \, \| \, \widehat G_-^{\, \frac 12}\,
\widehat P^{\, \Lambda ^*\backslash {\mathcal K}(h^{\, l})}\, \varphi \, \| \, .
\eqno (5.26)
$$

The following theorem is a key point in the proof of Theorem \ref{th1.5}.

\begin{theorem} \label{th5.1}
Let $\widehat V^{\, (s)}\in L^3\ln ^{\, 1+\delta } L(K;{\mathcal S}^{(s)}_M)$, 
$\delta >0$, $s=0,1$, let $\widehat V=\widehat V^{\, (0)}+\widehat V^{\, (1)}$, and 
let $\Xi >1$. Then $\mathrm ($for a given number $\varepsilon >0$$\mathrm )$ there 
exists a number $\widetilde \varkappa _0>2h$ such that for any $\varkappa _1\geqslant 
\widetilde \varkappa _0$ there is a number $\varkappa \in [\varkappa _1\, , \Xi \,
\varkappa _1]$ such that for all vectors $k\in {\mathbb R}^3$ with $|(k,\gamma )|
=\pi $, and all vector functions $\varphi \in L^2(K;{\mathbb C}^M)$ we have
$$
\| \widehat G_-^{-\frac 12}\widehat P^{\, (+)}\widehat V\widehat P^{\, (-)}
\varphi \| \leqslant \varepsilon \, \| \widehat G^{\, \frac 12}_-\widehat P^{\, (-)}
\varphi \| \, .  \eqno(5.27)
$$
\end{theorem}

\begin{proof}
For all $\mu , \nu =1,\dots ,l$ and all vector functions $\psi \in L^2(K;{\mathbb 
C}^M)$, we shall obtain upper estimates for the norms
$$
\| \widehat P^{\, (+)}_{\mu}\widehat V\widehat P^{\, (-)}_{\nu}\psi \| \, .
$$
Given $\mu , \nu \in \{ 1,\dots ,l\} $ and $n\in \Lambda ^*$, by $S_{\mu \nu }(n)$
denote the number of vectors $N\in {\mathcal K}_{\mu}\, $ such that $N-n\in 
{\mathcal K}_{\nu}\, $. If either $2\pi |n_{\perp}|>2\varkappa +h^{\, \mu }
+h^{\, \nu }$ or $2\pi |n_{\parallel}|>h^{\, \mu }+h^{\, \nu }$, then $S_{\mu
\nu }(n)=0$. Since $h>2\pi d(K^*)$, we have $h^{\, l}+2\pi d(K^*)<2h^{\, l}\leqslant
\varkappa _1 \leqslant \varkappa $, hence, $h^{\, l} <\varkappa -2\pi d(K^*)$. Under 
these conditions, there is a universal constant $c_4>0$ (see \cite{TMF95}) such that
for all $\mu , \nu =1,\dots ,l$ and all vectors $n\in \Lambda ^*$ with $\pi 
|n_{\perp}|\leqslant \varkappa +h^{\, \max \, \{ \mu ,\nu \} }$ (and $\pi 
|n_{\parallel}|\leqslant h^{\, \max \, \{ \mu ,\nu \} }$), the estimate
$$
S_{\mu \nu }(n)\leqslant \frac {c_4}{v(K^*)}\ \frac {h^{\, \mu +\nu +\min \, \{ \mu 
,\nu \} }\varkappa ^{\, \frac 32}}{(\pi |n_{\perp}|+h^{\, \max \, \{ \mu ,\nu \} })
\, \sqrt {\varkappa +2h^{\, \max \, \{ \mu ,\nu \} }-\pi |n_{\perp}|}}
$$
holds. Therefore (also see (1.4) and (5.1)), for all matrix functions $\widehat 
W\in L^2(K;{\mathcal M}_M)$, which have the Fourier coefficients $\widehat 
W_N\, $, $N\in \Lambda ^*$, that commute with all orthogonal projections $\widehat 
P^{\, \pm }_{\widetilde e}\, $, $\widetilde e\in S_1(e)$, and for all vector functions 
$\psi \in L^2(K;{\mathbb C}^M)$ we get
$$
\| \widehat P^{\, (+)}_{\mu}\widehat W\widehat P^{\, (-)}_{\nu}\psi \| ^2 =
\eqno (5.28)
$$ $$
=\, v(K) \sum\limits_{N\, \in \, {\mathcal K}_{\mu}}\, \bigl\| \sum\limits
_{\substack{n\, \in \, \Lambda ^*\, : \\ N-n\, \in \, {\mathcal K}_{\nu}}}
\widehat P^{\, +}_{\widetilde e(k+2\pi N)}\widehat W_n\widehat P^{\, -}_{\widetilde 
e(k+2\pi (N-n))}\, \psi _{N-n}\, \| ^2 \leqslant
$$ $$
\leqslant \frac 14\, v(K) \sum\limits_{N\, \in \, {\mathcal K}_{\mu}}\, \bigl( 
\sum\limits _{\substack{n\, \in \, \Lambda ^*\, : \\ N-n\, \in \, {\mathcal K}_{\nu}}}
|\, \widetilde e(k+2\pi N)-\widetilde e(k+2\pi (N-n))\, |\, 
\| \widehat W_n\| \, \| (\widehat P^{\, -}\psi )_{N-n}\| \, \bigr) ^2 
\leqslant
$$ $$
\leqslant v(K)\, \varkappa ^{-2}\sum\limits_{N\, \in \, {\mathcal K}_{\mu}}\, \bigl( 
\sum\limits_{\substack{n\, \in \, \Lambda ^*\, : \\ N-n\, \in \, 
{\mathcal K}_{\nu}}}2\pi |n_{\perp}|\, \| \widehat W_n\| \,
\| (\widehat P^{\, -}\psi )_{N-n}\| \, \bigr) ^2 \leqslant
$$ $$
\leqslant v(K)\, \varkappa ^{-2}\sum\limits_{N\, \in \, {\mathcal K}_{\mu}}\, \bigl( 
\sum\limits_{\substack{n\, \in \, \Lambda ^*\, : \\ N-n\, \in \, 
{\mathcal K}_{\nu}}}(2\pi |n_{\perp}|)^2 \| \widehat W_n\| ^2 \, 
\bigr) \bigl( \sum\limits_{\substack{n\, \in \, \Lambda ^*\, : \\ N-n\, \in \, 
{\mathcal K}_{\nu}}}\| (\widehat P^{\, -}\psi )_{N-n}\| ^2\, \bigr) =
$$ $$
=\varkappa ^{-2} \, \bigl( \sum\limits_{n\, \in \, \Lambda ^*}\, \bigl( 
\sum\limits_{\substack{N\, \in \, {\mathcal K}_{\mu}\, : \\ N-n\, \in \, 
{\mathcal K}_{\nu}}}1\, \bigr) \, (2\pi |n_{\perp}|)^2 \| \widehat W_n\| 
^2 \, \bigr) \, \| \widehat P^{\, (-)}_{\nu}\psi \| ^2 =
$$ $$
=\varkappa ^{-2} \, \bigl( \sum\limits_{n\, \in \, \Lambda ^*}S_{\mu \nu }(n)\,
(2\pi |n_{\perp}|)^2 \| \widehat W_n\| ^2 \, \bigr) \, \| \widehat P^{\, 
(-)}_{\nu}\psi \| ^2 \leqslant
$$ $$
\leqslant c_4\, v(K)\, h^{\, \mu +\nu +\min \, \{ \mu ,\nu \} }\, \times 
$$ $$
\times \ \biggl( \, \sum\limits_{\substack{n\, \in \, \Lambda ^*\, : \\ \pi 
|n_{\perp}|\, \leqslant \, \varkappa +h^{\, \max \, \{ \mu ,\nu \} }\, , \\ \pi 
|n_{\parallel}|\, \leqslant \, h^{\, \max \, \{ \mu ,\nu \} } }}\frac {(2\pi 
|n_{\perp}|)^2 \| \widehat W_n\| ^2}{(\pi |n_{\perp}|+h^{\, \max \, \{ \mu ,\nu \} 
})\, \sqrt {\varkappa (\varkappa +2h^{\, \max \, \{ \mu ,\nu \} }-\pi |n_{\perp}|)}}\ 
\biggr) \, \| \widehat P^{\, (-)}_{\nu}\psi \| ^2\, .
$$
For any $a\geqslant 0$, the Fourier coefficients $(\widehat V_{[a]})_N\, $, $N\in
\Lambda ^*$, of the matrix function $\widehat V_{[a]}\in L^3(K;{\mathcal M}_M)
\subset L^2(K;{\mathcal M}_M)$ commute with all orthogonal projections $\widehat 
P^{\, \pm }_{\widetilde e}\, $, $\widetilde e\in S_1(e)$. In particular, from this
and (5.28), using the inequalities $2\varkappa +2h^{\, \max \, \{ \mu ,\nu \} }
\leqslant 2\varkappa +2h^{\, l}<3\varkappa $ and $\varkappa \leqslant \Xi \varkappa 
_1\, $, we obtain
$$
\| \widehat P^{\, (+)}_{\mu}\widehat V_{[a]}\widehat P^{\, (-)}_{\nu}\psi \| 
\leqslant \eqno (5.29)
$$ $$
\leqslant h^{\, \frac 12\, (\mu +\nu +\min \, \{ \mu ,\nu \} )}\, \biggl( 6c_4
\, \Xi \, \sqrt {\varkappa _1}\, v(K)\, \sum\limits_{\substack{n\, \in \, \Lambda 
^*\, : \\ \pi |n_{\perp}|\, \leqslant \, \varkappa +h^{\, \max \, \{ \mu ,\nu \} 
}}}\frac {\| (\widehat V_{[a]})_n\| ^2}{\sqrt {\varkappa +2h^{\, \max \, \{ \mu ,\nu 
\} }-\pi |n_{\perp}|}}\ \biggr) ^{\frac 12}\, \| \widehat P^{\, (-)}_{\nu}\psi \| \, .
$$

Let $a\geqslant 2$. Denote
$$
Y_{\delta}(\widehat V;a)\doteq \int_{\{ x\, \in \, K\, :\, \| \widehat 
V(x)\| \, \geqslant \, a\} }\| \widehat V(x)\| ^3\ln^{\, 1+\delta 
}\| \widehat V(x)\| \, dx \, ;
$$
$Y_{\delta}(\widehat V;a)\downarrow 0$ as $a\to +\infty $. The following
inequality is valid:
$$
\| \widehat V_{[a]}\| ^2_{L^2(K;{\mathcal M}_M)}\leqslant  \eqno (5.30)
$$ $$
\leqslant (a\ln ^{\, 1+\delta }a)^{-1}\int_{\{ x\, \in \, K\, :\, \| \widehat 
V(x)\| \, \geqslant \, a\} }\| \widehat V(x)\| ^3\ln^{\, 1+\delta 
}\| \widehat V(x)\| \, dx=(a\ln ^{\, 1+\delta }a)^{-1}\, Y_{\delta}(\widehat V;a)\, .
$$

Choose numbers $r_j\in (0,1]$, $j\in {\mathbb N}$, such that $r_1=1$, $r_j\downarrow 
0$, $a_j\doteq h^{\, \frac 13\, j}r_j\uparrow +\infty $ (then $a_j\geqslant 2$), 
$r_j^{-1}Y_{\delta}(\widehat V;a_j)\downarrow 0$ as $j\to +\infty $, and
$$
\sum\limits_{j\, =\, 1}^{+\infty }\, (\ln a_j)^{-1-\delta }<+\infty \, .
$$

Let us define the functions
$$
{\mathcal G}^{\, (l)}(h;\xi )=
\left\{
\begin{array}{lll}
(h-\xi )^{-\frac 12} & \text {if}\ \xi <h\, , \\
(h^{\, k}-\xi )^{-\frac 12} & \text {if}\ h^{\, k-1}<\xi <h^{\, k}\, ,\
k=2,\dots ,l\, , \\
0 & \text {if}\ \xi >h^{\, l}\, . 
\end{array}
\right.
$$
For all $\Delta >0$, we write
$$
T_{l,\, h}(\Delta )=\max\limits_{\tau \, \in \, {\mathbb R}}\, \int_{\tau }
^{\tau +\Delta }{\mathcal G}^{\, (l)}(h;\xi )\, d\xi \, .
$$
If $0<\Delta \leqslant h^{\, 2}-h$, then
$$
T_{l,\, h}(\Delta )=\int_0^{\Delta }\frac {d\xi }{\sqrt {\xi }}=2\sqrt {\Delta }\, .
$$
If either $h^{\, k}-h<\Delta \leqslant h^{\, k+1}-h$, $k=2,\dots ,l-1$, or $h^{\, l}-h
<\Delta $ (for $k=l$), then
$$
T_{l,\, h}(\Delta )=\biggl( \, \int_0^{h^{\, 2}-h}+\int_0^{h^{\, 3}-h^{\, 2}}+\dots
+\int_0^{h^{\, k}-h^{\, k-1}}+\int_0^{\Delta -(h^{\, k}-h)}\ \biggr) \, \frac 
{d\xi }{\sqrt {\xi }}<
$$ $$
<\, 2\, \bigl( \bigl( 1-\frac 1{\sqrt h}\, \bigr) ^{-1}h^{\, \frac k2}+\sqrt {\Delta 
}\, \bigr) < 5\, \sqrt {\Delta }\, .
$$

Assign the number $j=j(\mu ,\nu )=\mu +\nu +\min \, \{ \mu ,\nu \} \in \{ 1,\dots ,
3l\} $ to each ordered pair $(\mu ,\nu )$ of numbers $\mu , \nu \in \{ 1,\dots ,l\} $. 
Let ${\mathcal L}(j)$, where $j=1,\dots ,3l$, be the set of ordered pairs $(\mu ,\nu 
)$ with $j(\mu ,\nu )=j$. We have $\# \, {\mathcal L}(j)<j$. For every $j\in \{ 1,
\dots ,3l\} $ (if ${\mathcal L}(j)\neq \emptyset $), let $\mu _1^{(j)},\dots ,\mu 
_{k(l,j)}^{(j)}$ be the different numbers $\max \, \{ \mu ,\nu \} $ with
$(\mu ,\nu )\in {\mathcal L}(j)$ arranged in the increasing order. We define the
functions
$$
{\mathcal G}^{\, (l)}_j(h;\xi )=
\left\{
\begin{array}{lll}
(h^{\, \mu _1^{(j)}}-\xi )^{-\frac 12} & \text {if}\ \xi <h^{\, \mu _1^{(j)}}
\, , \\
(h^{\, \mu _k^{(j)}}-\xi )^{-\frac 12} & \text {if}\ h^{\, \mu _{k-1}^{(j)}}
<\xi <h^{\, \mu _k^{(j)}}\, ,\ k=2,\dots ,k(l,j)\, , \\
0 & \text {if}\ \xi >h^{\, \mu _{k(l,j)}^{(j)}}\, . 
\end{array}
\right.
$$
For all $\xi \in {\mathbb R}\backslash \{ h,h^{\, 2},\dots ,h^{\, l}\} $, we have
${\mathcal G}^{\, (l)}_j(h;\xi )\leqslant {\mathcal G}^{\, (l)}(h;\xi )$. Hence,
for all $\Delta >0$,
$$
\max\limits_{\tau \, \in \, {\mathbb R}}\, \int_{\tau }
^{\tau +\Delta }{\mathcal G}^{\, (l)}_j(h;\xi )\, d\xi \leqslant T_{l,\, h}(\Delta )
<5\, \sqrt {\Delta }\, .  \eqno (5.31)
$$

Using (5.1), (5.30), and (5.31), we obtain the following estimates (the prime above
the summation sign means that we omit the summands with $Y_{\delta}(\widehat 
V;a_j)=0$):
$$
{\sum\limits_{j\, =\, 1}^{3l}}^{\, \prime}\, a_jY_{\delta}^{-1}(\widehat V;a_j)\, 
\frac 1{(\Xi -1)\varkappa _1}\, \times \eqno (5.32)
$$ $$
\times \ \int_{\varkappa _1}^{\, \Xi \, \varkappa _1}\, \biggl( \,
\max\limits_{(\mu ,\nu )\, \in \, {\mathcal L}(j)}\,   \sum\limits_{\substack{n\, 
\in \, \Lambda ^*\, : \\ \pi |n_{\perp}|\, \leqslant \, \varkappa +h^{\, \max \,
\{ \mu ,\nu \} }}}\frac {\| (\widehat V_{[a_j]})_n\| ^2}{\sqrt {\varkappa +2h^{\, 
\max \, \{ \mu ,\nu \} }   -\pi |n_{\perp}|}}\ \biggr) \, d\varkappa \leqslant
$$ $$
\leqslant {\sum\limits_{j\, =\, 1}^{3l}}^{\, \prime}\, a_jY_{\delta}^{-1}(\widehat 
V;a_j)\, \frac 1{(\Xi -1)\varkappa _1}\, \int_{\varkappa _1}^{\, \Xi \, \varkappa _1}\, 
\bigl( \, \sum\limits_{n\, \in \, \Lambda ^*}{\mathcal G}^{\, (l)}_j(h;\pi |n_{\perp}|
-\varkappa )\, \| (\widehat V_{[a_j]})_n\| ^2\, \bigr) \, d\varkappa \leqslant
$$ $$
\leqslant \, \frac 5{\sqrt {(\Xi -1)\varkappa _1}}\ {\sum\limits_{j\, =\, 1}^{3l}}
^{\, \prime}\, a_jY_{\delta}^{-1}(\widehat V;a_j)\, \sum\limits_{n\, \in \, \Lambda 
^*}\| (\widehat V_{[a_j]})_n\| ^2 \leqslant
$$ $$
\leqslant \, \frac {5M}{\sqrt {(\Xi -1)\varkappa _1}}\ v^{-1}(K)\ {\sum\limits_{j\, 
=\, 1}^{3l}}^{\, \prime}\, a_jY_{\delta}^{-1}(\widehat V;a_j)\, \| (\widehat V
_{[a_j]})_n\| ^2_{L^2(K;{\mathcal M}_M)} \leqslant
$$ $$
\leqslant \, \frac {5M}{\sqrt {(\Xi -1)\varkappa _1}}\ v^{-1}(K)\, \sum\limits
_{j\, =\, 1}^{+\infty }\, (\ln a_j)^{-1-\delta }\, .
$$ 

On the other hand, 
$$
\sum\limits_{j=1}^{3l}\, 2^{-j}\, \sum\limits_{(\mu ,\nu )\, \in \, {\mathcal L}
(j)}\, \frac 1{(\Xi -1)\varkappa _1}\, \times  \eqno (5.33)
$$ $$
\times \ \int_{\varkappa _1}^{\, \Xi \, \varkappa _1}\, \biggl( \, \sum\limits
_{\substack{n\, \in \, \Lambda ^*\, : \\ \frac 12\, \varkappa _1\, <\, \pi 
|n_{\perp}|\, \leqslant \, \varkappa +h^{\, \max \, \{ \mu ,\nu \} }}}\frac {\| 
\widehat V_n\| ^2}{\sqrt {\varkappa +2h^{\, \max \, \{ \mu ,\nu \} }-\pi 
|n_{\perp}|}}\ \biggr) \, d\varkappa \leqslant
$$ $$
\leqslant \sum\limits_{j=1}^{3l}\, 2^{-j}\, \sum\limits_{(\mu ,\nu )\, \in \, 
{\mathcal L}(j)}\, \frac 1{(\Xi -1)\varkappa _1}\, \times
$$ $$
\times \ \sum\limits
_{\substack{n\, \in \, \Lambda ^*\, : \\ \varkappa _1\, <\, 2\pi |n_{\perp}|}}
\ \biggl( \, \int_{\, [\varkappa _1,\, \Xi \, \varkappa _1]\, \cap \, [\pi |n_{\perp}|-
h^{\, \max \, \{ \mu ,\nu \} },+\infty )}\ \frac {d\varkappa }{\sqrt {\varkappa +2h^{\, 
\max \, \{ \mu ,\nu \} }-\pi |n_{\perp}|}}\, \biggr) \, \| \widehat V_n\| ^2
\leqslant
$$ $$
\leqslant \biggl( \, \sum\limits_{j=1}^{3l}\, j2^{-j}\, \biggr) \ \frac 1{(\Xi -1)
\varkappa _1}\ \biggl( \, \int_{\varkappa _1}^{\, \Xi \, \varkappa _1}\, \frac
{d\varkappa }{\sqrt {\varkappa -\varkappa _1}}\, \biggr) \, \sum\limits
_{\substack{n\, \in \, \Lambda ^*\, : \\ \varkappa _1\, <\, 2\pi |n_{\perp}|}}
\, \| \widehat V_n\| ^2 \leqslant
$$ $$
\leqslant \frac 4{\sqrt {(\Xi -1)\varkappa _1}}\ \sum\limits_{\substack{n\, \in \, 
\Lambda ^*\, : \\ \varkappa _1\, <\, 2\pi |n_{\perp}|}}\, \| \widehat V_n\| ^2\, .
$$
Let us denote
$$
c_5=c_5(M,\Lambda ,\Xi ,h;\widehat V,\{ r_j\} )\doteq \frac 2{\sqrt {\Xi -1}}\
\bigl( \, 2+\frac 52\, Mv^{-1}(K)\, \sum\limits_{j\, =\, 1}^{+\infty }\, (\ln a_j)
^{-1-\delta }\, \bigr) \, .
$$
Then (5.32) and (5.33) imply that there is a number $\varkappa \in [\varkappa _1\, ,
\Xi \, \varkappa _1]$ such that for all $\mu ,\nu =1,\dots ,l$,
$$
\sum\limits_{\substack{n\, \in \, \Lambda ^*\, : \\ 
\pi |n_{\perp}|\, \leqslant \, \varkappa +h^{\, \max \, \{ \mu ,\nu \} }}}\frac 
{\| (\widehat V_{[a_j]})_n\| ^2}{\sqrt {\varkappa +2h^{\, \max \, \{ \mu ,\nu \} }
-\pi |n_{\perp}|}}\, \leqslant \, c_5\, \varkappa _1^{-\frac 12}\, a_j^{-1}Y
_{\delta}(\widehat V;a_j)\, ,  \eqno (5.34)
$$ $$
\sum\limits_{\substack{n\, \in \, \Lambda ^*\, : \\ 
\frac 12\, \varkappa _1\, <\, \pi |n_{\perp}|\, \leqslant \, \varkappa +h^{\, \max \, 
\{ \mu ,\nu \} }}}\frac {\| \widehat V_n\| ^2}{\sqrt {\varkappa +2h^{\, \max \, \{ 
\mu ,\nu \} }-\pi |n_{\perp}|}}\, \leqslant \, 2^{\, j}\, c_5\, \varkappa _1^{-\frac 
12}\, \sum\limits_{\substack{n\, \in \, \Lambda ^*\, : \\ \varkappa _1\, < \, 
2\pi |n_{\perp}|}}\| \widehat {\mathcal V}_n\| ^2\, ,  \eqno (5.35)
$$ 
where $a_j=h^{\, \frac 13j}r_j\, $, $j=j(\mu ,\nu )$. For every number $j\in \{ 1,
\dots ,3l\} $ (if the number $\varkappa _1$ is fixed), from (5.28) and (5.35) it 
follows that for all ordered pairs $(\mu ,\nu )\in {\mathcal L}(j)$, all vector 
functions $\psi \in L^2(K;{\mathbb C}^M)$, and for the number $\varkappa \in 
[\varkappa _1\, ,\, \Xi \, \varkappa _1]$ chosen as above, the estimate
$$
\| \widehat P^{\, (+)}_{\mu}\widehat V\widehat P^{\, (-)}_{\nu}\psi \| 
\leqslant  \eqno (5.36)
$$ $$
\leqslant (2c_4\, v(K))^{\, \frac 12}\, h^{\, \frac 12j}\, \bigl( \bigl( \, 3\, 
\Xi \, \sqrt {\varkappa _1}\sum\limits_{\substack{n\, \in \, \Lambda ^*\, : \\ 
\frac 12\, \varkappa _1\, <\, \pi |n_{\perp}|\, \leqslant \, \varkappa +h^{\, \max \, 
\{ \mu ,\nu \} }}}\frac {\| \widehat V_n\| ^2}{\sqrt {\varkappa +2h^{\, \max \, \{ 
\mu ,\nu \} }-\pi |n_{\perp}|}}\ \bigr) ^{\frac 12}+
$$ $$
+\, \bigl( \, \sqrt {\varkappa _1}\sum\limits_{\substack{n\, \in \, \Lambda ^*\, : \\ 
2\pi |n_{\perp}|\, \leqslant \, \varkappa _1}}\frac {2\pi |n_{\perp}|\ \| \widehat 
V_n\| ^2}{\varkappa _1 \sqrt {\varkappa +2h^{\, \max \, \{ \mu ,\nu \} }-\pi 
|n_{\perp}|}}\ \bigr) ^{\frac 12}\, \bigr) \, \| \widehat P^{\, (-)}_{\nu}\psi \| 
\leqslant
$$ $$
\leqslant (2c_4\, v(K))^{\, \frac 12}\, h^{\, \frac 12j}\, \bigl( \bigl( \, 
3\cdot 2^{\, j}\, \Xi\, c_5\sum\limits_{\substack{n\, \in \, \Lambda ^*\, : \\ 
\varkappa _1\, < \, 2\pi |n_{\perp}|}}\| \widehat V_n\| ^2\, \bigr) ^{\frac 12}+
\bigl( \, \sqrt 2 \sum\limits_{\substack{n\, \in \, \Lambda ^*\, : \\ 2\pi 
|n_{\perp}|\, \leqslant \, \varkappa _1}}\frac {2\pi |n_{\perp}|}{\varkappa _1}\   
\| \widehat V_n\| ^2\, \bigr) ^{\frac 12}\, \bigr) \, \| \widehat P^{\, (-)}
_{\nu}\psi \| 
$$
holds. Let $\varepsilon ^{\, \prime \prime}\doteq \frac 13\, \varepsilon \, h^{-1}
\min \, \{ 1,\pi |\gamma |^{-1}\} $. We choose a number $j_0=j_0(\varepsilon ^{\, 
\prime \prime})\in {\mathbb N}$ (also dependent on $v(K)$, $c_5\, $, $\Xi $, $\widehat 
V$, $h$, $\{ r_j\} $) such that $r_j\leqslant \frac 12\, \varepsilon ^{\, \prime 
\prime}$ and
$$
(6\, \Xi \, c_4\, c_5\, v(K))^{\frac 12}\, (r_j^{-1}Y_{\delta}(\widehat V;a_j))
^{\frac 12}\leqslant \frac 12\, \varepsilon ^{\, \prime \prime} 
$$
for all $j>j_0\, $. Since
$$
\sum\limits_{\substack{n\, \in \, \Lambda ^*\, : \\ \varkappa _1\, < \, 2\pi 
|n_{\perp}|}}\, \| \widehat V_n\| ^2\to 0\, ,\ \ \sum\limits_{\substack{n\, \in 
\, \Lambda ^*\, : \\ 2\pi |n_{\perp}|\, \leqslant \, \varkappa _1}}\, \frac {2\pi 
|n_{\perp}|}{\varkappa _1}\ \| \widehat V_n\| ^2 \to 0 \ \ \ {\text {as}}\ \, 
\varkappa _1\to +\infty \, ,
$$
from (5.36) it follows that there is a number 
$\widetilde \varkappa _0=\widetilde \varkappa _0(M,\Lambda ,\, \Xi ,h,\widehat V,
\{ r_j\} ;j_0,\varepsilon ^{\, \prime \prime})>2h$ such that for all $\varkappa 
_1\geqslant \widetilde \varkappa _0\, $, for the numbers $\varkappa \in [\varkappa 
_1\, ,\, \Xi \, \varkappa _1]$ chosen as above, for all numbers $\mu ,\nu =1,\dots ,l$ 
with $j(\mu ,\nu )\leqslant j_0\, $ (where $l=l(h,\varkappa _1)\in {\mathbb N}$), 
and all vector functions $\psi \in L^2(K;{\mathbb C}^M)$ we have
$$
\| \widehat P^{\, (+)}_{\mu}\widehat V\widehat P^{\, (-)}_{\nu}\psi \| \leqslant 
\varepsilon ^{\, \prime \prime}\, h^{\, \frac 13\, j(\mu ,\nu )}\, \| \widehat P^{\, 
(-)}_{\nu}\psi \| \, .  \eqno (5.37)
$$
At the same time, if $\varkappa _1\geqslant \widetilde \varkappa _0>2h$ and the 
number $[\varkappa _1\, ,\, \Xi \, \varkappa _1]\ni \varkappa $ is chosen as above, 
then for all numbers $\mu ,\nu =1,\dots ,l$ with $j=j(\mu ,\nu )>j_0\, $, and all
vector functions $\psi \in L^2(K;{\mathbb C}^M)$, taking into account estimates (5.29),
(5.34), and the definition of the number $j_0\, $, we derive
$$
\| \widehat P^{\, (+)}_{\mu}\widehat V\widehat P^{\, (-)}_{\nu}\psi \| \leqslant 
\| \widehat P^{\, (+)}_{\mu}(\widehat V-\widehat V_{[a_j]})\widehat P^{\, 
(-)}_{\nu}\psi \| +\| \widehat P^{\, (+)}_{\mu}\widehat V_{[a_j]}\widehat P^{\, 
(-)}_{\nu}\psi \| \leqslant
$$ $$
\leqslant \bigl( \, a_j+ h^{\, \frac 12\, j}\, \bigl( \, 6\, \Xi \, c_4\, \sqrt 
{\varkappa _1}\, v(K)\, \sum\limits_{\substack{n\, \in \, \Lambda ^*\, : \\ \pi 
|n_{\perp}|\, \leqslant \, \varkappa +h^{\, \max \, \{ \mu ,\nu \} }}}\frac {\| 
(\widehat V_{[a_j]})_n\| ^2}{\sqrt {\varkappa +2h^{\, \max \, \{ \mu ,\nu \} }-\pi 
|n_{\perp}|}}\ \bigr) ^{\frac 12}\, \bigr) \, \| \widehat P^{\, (-)}_{\nu}\psi \| 
\leqslant
$$ $$
\leqslant h^{\, \frac 13\, j}\, \bigl( r_j+(6\, \Xi \, c_4\, c_5\, v(K))^{\, \frac 
12}\, (r_j^{-1}Y_{\delta}(\widehat V;a_j))^{\, \frac 12}\, \bigr) \, \| \widehat P^{\, 
(-)}_{\nu}\psi \| \leqslant \varepsilon ^{\, \prime \prime}\, h^{\, \frac 13\, j}\, \| 
\widehat P^{\, (-)}_{\nu}\psi \| \, ,
$$
that is, estimate (5.37) is valid for all $\mu ,\nu =1,\dots ,l$.

From (5.3), (5.4), and (5.37) (for all $\mu ,\nu =1,\dots ,l$ and all vector
functions $\psi \in L^2(K;{\mathbb C}^M)$) we deduce the following estimate:
$$
\| \widehat G_-^{-\frac 12}\widehat P^{\, (+)}_{\mu}\widehat V\widehat G_-^{-\frac 12}
\widehat P^{\, (-)}_{\nu}\psi \| \leqslant \varepsilon ^{\, \prime \prime}\, (h\, 
\max \{ 1,\, \frac {|\gamma }{\pi }\, \} )\, h^{\, \frac 13\, j(\mu ,\nu )}\, h^{\, 
-\frac {\mu +\nu }2}\, \| \widehat P^{\, (-)}_{\nu}\psi \| =
$$ $$
=\, \frac {\varepsilon }3\ h^{\, \frac 13\, \min \, \{ \mu ,\nu \} }\, h^{\, -\frac 
{\mu +\nu }6}\, \| \widehat P^{\, (-)}_{\nu}\psi \| =\frac {\varepsilon }3\ h^{\, 
-\frac 16\, |\mu -\nu |}\, \| \widehat P^{\, (-)}_{\nu}\psi \| \, .
$$
Whence
$$
\| \widehat G_-^{-\frac 12}\widehat P^{\, (+)}\widehat V\widehat G_-^{-\frac 12}
\widehat P^{\, (-)}\psi \| ^2=\sum\limits_{\mu \, =\, 1}^l \, \| \widehat G_-
^{-\frac 12}\widehat P^{\, (+)}_{\mu}\widehat V\widehat G_-^{-\frac 12}\, \biggl( 
\, \sum\limits_{\nu \, =\, 1}^l\widehat P^{\, (-)}_{\nu}\biggr) \, \psi \| ^2\leqslant
$$ $$
\leqslant \sum\limits_{\mu \, =\, 1}^l\, \biggl( \, \sum\limits_{\nu \, =\, 1}^l \, 
\| \widehat G_-^{-\frac 12}\widehat P^{\, (+)}_{\mu }\widehat V\widehat G_-^{-\frac 
12}\widehat P^{\, (-)}_{\nu }\psi \| \biggr) ^2\leqslant 
\biggl( \, \frac {\varepsilon }3\, \biggr) ^2\, \sum\limits_{\mu \, =\, 1}^l\, 
\biggl( \, \sum\limits_{\nu \, =\, 1}^l h^{\, -\frac 16\, |\mu -\nu |}\, \| \widehat 
P^{\, (-)}_{\nu}\psi \| \biggr) ^2\leqslant    
$$ $$
\leqslant \biggl( \, \frac {\varepsilon }3\, \biggr) ^2\, \sum\limits_{\mu \, =\, 1}
^l\, \biggl( \, \sum\limits_{\nu \, =\, 1}^l h^{\, -\frac 16\, |\mu -\nu |}\, \biggr) 
\, \biggl( \, \sum\limits_{\nu \, =\, 1}^l h^{\, -\frac 16\, |\mu -\nu |}\, \| 
\widehat P^{\, (-)}_{\nu}\psi \| ^2 \biggr) \leqslant 
$$ $$
\leqslant \biggl( \, \frac {1+h^{\, -\frac 16}}{1-h^{\, -\frac 16}}\, \biggr) \,  
\biggl( \, \frac {\varepsilon }3\, \biggr) ^2\, \sum\limits_{\nu \, =\, 1}^l\,
\biggl( \, \sum\limits_{\mu \, =\, 1}^l h^{\, -\frac 16\, |\mu -\nu |}\, \biggr) \,
\| \widehat P^{\, (-)}_{\nu}\psi \| ^2\leqslant
$$ $$
\leqslant \biggl( \, \frac {1+h^{\, -\frac 16}}{1-h^{\, -\frac 16}}\, \biggr) ^2\,  
\biggl( \, \frac {\varepsilon }3\, \biggr) ^2\, \sum\limits_{\nu \, =\, 1}^l\, 
\| \widehat P^{\, (-)}_{\nu}\psi \| ^2\leqslant (\varepsilon )^2\, \| \widehat P^{\, 
(-)}\psi \| ^2\, .
$$
To complete the proof, it remains to put $\psi =\widehat G_-^{\, \frac 12}\widehat 
P^{\, (-)}\varphi $, $\varphi \in L^2(K;{\mathbb C}^M)$. Theorem \ref{th5.1} is 
proved.
\end{proof}

Now, let us use estimates (5.14) -- (5.17), (5.20) -- (5.22), and (5.26), (5.27)
(conditions (5.11) and (5.23) are fulfilled because 
$$
\| \widehat V \| ^{(\infty ,\, 
{\mathrm {loc}})}_{L^3_w(K;{\mathcal M}_M)}=\| \widehat V \| ^{(\infty )}_{L^3_w(K;
{\mathcal M}_M)}=0\, )\, .
$$
We choose the number $\varkappa _0\doteq \max \, \{ \varkappa _0^{\, \prime}
(\varepsilon ), \varkappa _0^{\, \prime}\, ((4\, \Xi \, h+1)^{-\frac 12}\,
\varepsilon ), \varkappa _0^{\, \prime \prime}, \widetilde \varkappa _0\} >2h$. Then,
for all $\varkappa _1\geqslant \varkappa _0\, $, for the number $\varkappa \in 
[\varkappa _1\, ,\, \Xi \, \varkappa _1]$ chosen in Theorem \ref{th5.1}, for all 
vectors $k\in {\mathbb R}^3$ with $|(k,\gamma )|=\pi $, and all vector functions $\varphi \in 
\widetilde H^1(K;{\mathbb C}^M)$ we get 
$$
\| \widehat G_-^{-\frac 12}\widehat P^{\, +}_*\widehat V\varphi \| ^2+\| \widehat 
G^{-\frac 12}_+\widehat P^{\, -}_*\widehat V\varphi \| ^2\leqslant  \eqno (5.38)
$$ $$
\leqslant 3\, \bigl( \, \| \widehat G_-^{-\frac 12}P^{\, (+)}\widehat VP^{\, (+)}
\varphi \| ^2+\| \widehat G_-^{-\frac 12}P^{\, (+)}\widehat VP^{\, (-)}\varphi \| ^2+
\| \widehat G_-^{-\frac 12}P^{\, (+)}\widehat VP^{\, \Lambda ^*\backslash {\mathcal
K}(h^{\, l})}\varphi \| ^2\bigr) +
$$ $$
+\, 3\, \bigl( \, \| \widehat G_+^{-\frac 12}P^{\, (-)}\widehat VP^{\, (+)}\varphi 
\| ^2+\| \widehat G_+^{-\frac 12}P^{\, (-)}\widehat VP^{\, (-)}\varphi \| ^2+
\| \widehat G_+^{-\frac 12}P^{\, (-)}\widehat VP^{\, \Lambda ^*\backslash {\mathcal
K}(h^{\, l})}\varphi \| ^2\bigr) +
$$ $$
3\, \bigl( \| \widehat G_-^{-\frac 12}P^{\, \Lambda ^*\backslash {\mathcal
K}(h^{\, l})}\widehat VP^{\, (+)}\varphi \| ^2+\| \widehat G_-^{-\frac 12}P^{\, 
\Lambda ^*\backslash {\mathcal K}(h^{\, l})}\widehat VP^{\, (-)}\varphi \| ^2+
\| \widehat G_-^{-\frac 12}P^{\, \Lambda ^*\backslash {\mathcal K}(h^{\, l})}
\widehat VP^{\, \Lambda ^*\backslash {\mathcal K}(h^{\, l})}\varphi \| ^2\bigr) 
\leqslant
$$ $$
\leqslant 9\, (\varepsilon )^2\, \bigl( \, \| \widehat G^{\, \frac 12}_+P^{\, (+)}
\varphi \| ^2+\| \widehat G_-^{\, \frac 12}P^{\, (-)}\varphi \| ^2+\| \widehat G_-
^{\, \frac 12}P^{\, \Lambda ^*\backslash {\mathcal K}(h^{\, l})}\varphi \| ^2\, 
\bigr) \leqslant
$$ $$
\leqslant (\varepsilon ^{\, \prime })^2\, \bigl( \, \| \widehat G^{\, \frac 12}_+
\widehat P^{\, +}\varphi \| ^2+\| \widehat G_-^{\, \frac 12}\widehat P^{\, -}
\varphi \| ^2\, \bigr) 
$$
(where $l=l(h,\varkappa _1)$). This completes the proof of Theorem \ref{th1.5}.

\section{Proof of Theorem \ref{th1.6}}

In what follows, we use the assumptions and the notation from Section 5 in the
case where $\Xi =1$ (and $\varkappa _1=\varkappa $). We assume that the vectors
$k\in {\mathbb R}^3$ satisfy the condition $|(k,\gamma )|=\pi $. The number $h$ 
is chosen (and fixed) in Section 5 (in the end of the proof of Theorem \ref{th1.6},
we shall put $h=64$). Fix a number $\varepsilon ^{\, \prime }>0$ and put 
$\varepsilon =\frac 1{2\sqrt 6}\, \varepsilon ^{\, \prime }$. Choose a number 
$\widetilde c_1=\widetilde c_1(h)>0$ such that $\widetilde c_1\leqslant \frac 
1{8\sqrt 6}\, c_3$ and $\widetilde c_1<\frac 1{8\sqrt 6}\, C^{-1}\, (2h^2+1)^{-1}
<C^{-1}$, where $c_3=c_3(h)$ is the constant from Lemma \ref{l5.3} and $C=C(3)$ is
the constant from inequality (0.6). If $\widehat V^{\, (s)}\in L^3_w(K;{\mathcal S}
^{(s)}_M)$ and
$$
\| \widehat V^{\, (s)}\| ^{(\infty ,\, {\mathrm {loc}})}_{L^3_w(K;{\mathcal M}_M)}
\leqslant \widetilde c_1\, \varepsilon ^{\, \prime }\, ,\ s=0,1\, ,
$$
then inequalities (5.11) and (5.23) (for $\Xi =h$) hold for the function $\widehat V
=\widehat V^{\, (0)}+\widehat V^{\, (1)}$.

Assume that (for $\Xi =1$) $\varkappa =\varkappa _1\geqslant 2h^2$ (then $l=l(h,
\varkappa )\geqslant 2$ and $2h^{\, l-1}<h^{\, l}\leqslant \frac 12\, \varkappa $). 
We define the functions
$$ 
\Theta (h,\varkappa ;t)=\left\{
\begin{array}{lll}
1 & {\text {if}} \ \, t\leqslant h^{\, l-1}\, , \\ [0.2cm]
2-h^{-l+1}t & {\text {if}} \ \, h^{\, l-1}<t\leqslant 2h^{\, l-1}\, , \\ [0.2cm]
0 & {\text {if}} \ \, t>2h^{\, l-1}\, ,
\end{array}
\right.
$$
and the operators $\widehat \Theta =\widehat \Theta (h,k;\varkappa )$ that take
vector functions $\psi \in L^2(K;{\mathbb C}^M)$ to the vector functions
$$
\widehat \Theta \psi =\sum\limits_{N\, \in \, \Lambda ^*}\, \Theta (h,\varkappa ;
G^{\, -}_N(k;\varkappa ))\, \psi _N\, e^{\, 2\pi i\, (N,x)}.
$$

The following theorem is a key point in the proof of Theorem \ref{th1.6}.

\begin{theorem} \label{th6.1}
Let $d=3$, $\gamma \in \Lambda \backslash \{ 0\} $, $\sigma \in (0,2]$. Then for 
any $\Lambda $-periodic matrix function 
$$
\widehat V^{\, (s)}=\sum\limits_{q=1}^{Q_s}\widehat V^{\, (s)}_q\, ,\ s=0,1\, ,
$$ 
with $\widehat V^{\, (s)}_q\in L^3_w(K;{\mathcal S}^{(s)}_M)$, $\beta _{\gamma ,\, 
\sigma }(0;\widehat V^{\, (s)}_q)<+\infty $, $q=1,\dots ,Q_s\, $, for which the
essential supports ${\mathrm {supp}}\, \widehat V^{\, (s)}_q$ do not intersect for
different $q$ $\mathrm ($for $s=0$ and $s=1$, separately$\mathrm )$, and for any 
$\delta >0$ there is a number $\widetilde \varkappa ^{\, \prime }_0(\delta )=
\widetilde \varkappa ^{\, \prime }_0(M,\Lambda ,|\gamma |,h,\sigma ,\widehat V^{\, 
(0)},\widehat V^{\, (1)};\delta )\geqslant 2h^2$ such that for all $\varkappa 
\geqslant \widetilde \varkappa ^{\, \prime }_0(\delta )$, all vectors $k\in 
{\mathbb R}^3$ with $|(k,\gamma )|=\pi $, and all vector functions
$\varphi \in L^2(K;{\mathbb C}^M)$ the inequalities
$$
\| \, \widehat G_-^{\, -\frac 12}\, \widehat P^{\, +}\, \widehat \Theta \, 
\widehat V^{\, (s)}\, \widehat P^{\, -}\, \widehat \Theta \, \varphi \, \|
\leqslant c_6\, (\, \delta +\max\limits_{q\, =\, 1,\dots ,Q_s}\, \beta _{\gamma ,\, 
\sigma }(\widehat V^{\, (s)}_q) )\, \| \, \widehat G_-^{\, \frac 12}\, 
\widehat P^{\, (-)}\varphi \, \| \, ,\ \, s=0,1\, ,
$$
hold, where $c_6=c_6\, (h,\sigma )>0$ $\mathrm ($see $\mathrm ($6.14$\mathrm 
)$$\mathrm )$.
\end{theorem}

Theorem \ref{th6.1} is proved in the end of this Section.
\vskip 0.2cm

{\it Proof} \, of Theorem \ref{th1.6}. First let us obtain estimates which are
similar to estimates (5.14) -- (5.17), (5.20) -- (5.22), and (5.26) used in the proof
of Theorem \ref{th1.5}. We assume that $\varkappa \geqslant \varkappa ^{\, \prime 
\prime \prime }_0\geqslant 2h^2\, $. A few additional lower bounds on the number 
$\varkappa ^{\, \prime \prime \prime }_0$ will be given below. For the vectors 
$k\in {\mathbb R}^3$ we suppose that $|(k,\gamma )|=\pi $, and for the number 
$l=l(h,\varkappa )\in {\mathbb N}\backslash \{ 1\} $ we have $2h^{\, l}\leqslant 
\varkappa <2h^{\, l+1}$ (the number $h$ satisfies the conditions from Section 5: 
$h\geqslant 64$ and $h>2\pi d(K^*)$). Let $\varkappa ^{\, \prime \prime \prime }_0
\geqslant \varkappa ^{\, \prime }_0(\varepsilon )$. By Lemma \ref{l5.3}, estimate
(5.12) holds. Hence, from (5.14), (5.15), and (5.20) it follows that for all
$\varphi \in L^2(K;{\mathbb C}^M)$ the following estimates are fulfilled:
$$
\| \, \widehat G_+^{\, -\frac 12}\, \widehat P^{\, (-)}\, \widehat V\, 
\widehat P^{\, \mp }\, \widehat \Theta \, \varphi \, \| \leqslant \varepsilon \,
\| \, \widehat G_{\mp }^{\, \frac 12}\, \widehat P^{\, (\mp )}\varphi \, \| \, ,
\eqno (6.1)
$$ $$
\| \, \widehat G_-^{\, -\frac 12}\, \widehat P^{\, (+)}\, \widehat V\, 
\widehat P^{\, +}\, \widehat \Theta \, \varphi \, \| \leqslant \varepsilon \,
\| \, \widehat G_+^{\, \frac 12}\, \widehat P^{\, (+)}\varphi \, \| \, .
\eqno (6.2)
$$
From (5.5) (setting $\Xi =h$ and replacing $l$ by $l-1$), for all $N\in \Lambda 
^*\backslash \, {\mathcal K}(h^{\, l-1})$ we obtain
$$
G^{\, -}_N(k;\varkappa )>(4h^2+1)^{-1}\, G^{\, +}_N(k;\varkappa )\, .
$$
Therefore, from Lemma \ref{l5.3}, for $\varkappa ^{\, \prime \prime \prime }_0
\geqslant \varkappa ^{\, \prime }((4h^2+1)^{-\frac 12}\, \varepsilon )$, (by analogy
with estimates (5.16) and (5.17)) we get
$$
\| \, \widehat G_-^{\, -\frac 12}\, \widehat P^{\, \Lambda ^*\backslash \, {\mathcal 
K}(h^{\, l-1})}\, \widehat V\, \widehat P^{\, \mp }\, \widehat \Theta \, \varphi \, 
\| \leqslant \varepsilon \, \| \, \widehat G_{\mp }^{\, \frac 12}\, \widehat 
P^{\, (\mp )} \varphi \, \| \, ,\ \, \varphi \in L^2(K;{\mathbb C}^M)\, ,
\eqno (6.3)
$$
and from (5.19), where  $\varepsilon$ is replaced by $(4h^2+1)^{-\frac 12}\, 
\varepsilon \, $, (by analogy with estimates (5.16) and (5.17)) we deduce the
inequalities
$$
\| \, \widehat G_{\mp }^{\, -\frac 12}\, P^{\, (\pm )}\, \widehat V\, \widehat 
P^{\, \Lambda ^*\backslash \, {\mathcal K}(h^{\, l-1})}\, \varphi \, \| \leqslant 
\varepsilon \, \| \, \widehat G_-^{\, \frac 12}\, \widehat P^{\, \Lambda ^*
\backslash \, {\mathcal K}(h^{\, l-1})}\, \varphi \, \| \, ,\ \, \varphi \in 
\widetilde H^1(K;{\mathbb C}^M)\, .  \eqno (6.4)
$$
Now, let $\varkappa ^{\, \prime \prime \prime }_0\geqslant \varkappa ^{\, \prime 
\prime }_0\, $, where the number $\varkappa ^{\, \prime \prime }_0$ is chosen for 
$\Xi =h$. Changing $l$ to $l-1$ in estimate (5.24), we have
$$
\| \, \widehat G_-^{\, -\frac 12}\, P^{\, \Lambda ^*\backslash \, {\mathcal K}(h^{\, 
l-1})}\, \widehat V\, \widehat P^{\, \Lambda ^*\backslash \, {\mathcal K}(h^{\, 
l-1})}\, \varphi \, \| \leqslant \varepsilon \, \| \, \widehat G_-^{\, \frac 12}\, 
\widehat P^{\, \Lambda ^*\backslash \, {\mathcal K}(h^{\, l-1})}\, \varphi \, \| 
\, ,\ \, \varphi \in \widetilde H^1(K;{\mathbb C}^M)\, .  \eqno (6.5)
$$
Under the conditions of Theorem \ref{th6.1}, we put $\delta =\frac 14\, \varepsilon 
\, c_6^{-1}$ and assume that the number $\varkappa ^{\, \prime \prime \prime }_0$ 
satisfies the last lower estimate: $\varkappa ^{\, \prime \prime \prime }_0\geqslant 
\widetilde \varkappa ^{\, \prime }_0(\delta )$. Choose a constant
$$
\widetilde c_1^{\, \prime }=\widetilde c_1^{\, \prime }(h,\sigma )=
\delta (\varepsilon ^{\, \prime})^{-1}=\frac 1{8\sqrt 6}\ c_6^{-1}\, .
$$
From Theorem \ref{th6.1} (for $\varkappa \geqslant \varkappa ^{\, \prime \prime 
\prime }_0$) it follows that
$$
\| \, \widehat G_-^{\, -\frac 12}\, \widehat P^{\, +}\, \widehat \Theta \, 
\widehat V\, \widehat P^{\, -}\, \widehat \Theta \, \varphi \, \| \leqslant 
\varepsilon \, \| \, \widehat G_-^{\, \frac 12}\, \widehat P^{\, (-)} \varphi \, 
\| \, ,\ \, \varphi \in L^2(K;{\mathbb C}^M)\, .  \eqno (6.6)
$$
Finally, from (6.1) -- (6.6), for all $\varkappa \geqslant \varkappa ^{\, \prime 
\prime \prime }_0\, $, all $k\in {\mathbb R}^3$ with $|(k,\gamma )|=\pi $, and all 
$\varphi \in \widetilde H^1(K;{\mathbb C}^M)$ (by analogy with (5.38)) we obtain 
estimate (1.15):
$$
\| \widehat G_-^{-\frac 12}\widehat P^{\, +}_*\widehat V\varphi \| ^2+\| \widehat 
G^{-\frac 12}_+\widehat P^{\, -}_*\widehat V\varphi \| ^2\leqslant  
$$ $$
\leqslant 3\, \bigl( \, \| \widehat G_-^{-\frac 12}P^{\, (+)}\widehat VP^{\, +}
\widehat \Theta \varphi \| ^2+2\, \| \widehat G_-^{-\frac 12}P^{\, +}\widehat \Theta 
\, \widehat VP^{\, -}\widehat \Theta \varphi \| ^2+\| \widehat G_-^{-\frac 12}P^{\, 
(+)}\widehat V(\widehat I-\widehat \Theta )\varphi \| ^2\bigr) +
$$ $$
+\, 3\, \bigl( \, \| \widehat G_+^{-\frac 12}P^{\, (-)}\widehat VP^{\, +}\widehat 
\Theta \varphi \| ^2+\| \widehat G_+^{-\frac 12}P^{\, (-)}\widehat VP^{\, -}
\widehat \Theta \varphi \| ^2+\| \widehat G_+^{-\frac 12}P^{\, (-)}\widehat V
(\widehat I-\widehat \Theta )\varphi \| ^2\bigr) +
$$ $$
3\, \bigl( \| \widehat G_-^{-\frac 12}P^{\, \Lambda ^*\backslash \, {\mathcal
K}(h^{\, l})}\widehat VP^{\, +}\widehat \Theta \varphi \| ^2+2\, \| \widehat G_-
^{-\frac 12}P^{\, \Lambda ^*\backslash \, {\mathcal K}(h^{\, l-1})}\widehat VP^{\, -}
\widehat \Theta \varphi \| ^2+\| \widehat G_-^{-\frac 12}P^{\, \Lambda ^*\backslash 
\, {\mathcal K}(h^{\, l})}\widehat V(\widehat I-\widehat \Theta )\varphi \| ^2\bigr) 
\leqslant
$$ $$
\leqslant 3\, (\varepsilon )^2\, \bigl( \, 3\, \| \widehat G^{\, \frac 12}_+P^{\, (+)}
\varphi \| ^2+5\, \| \widehat G_-^{\, \frac 12}P^{\, (-)}\varphi \| ^2+3\, \| 
\widehat G_-^{\, \frac 12}P^{\, \Lambda ^*\backslash \, {\mathcal K}(h^{\, l-1})}
\varphi \| ^2\, \bigr) \leqslant
$$ $$
\leqslant 3\, (\varepsilon )^2\, \bigl( \, 6\, \| \widehat G^{\, \frac 12}_+P^{\, (+)}
\varphi \| ^2+8\, \| \widehat G_-^{\, \frac 12}P^{\, (-)}\varphi \| ^2+\, \| 
\widehat G_-^{\, \frac 12}P^{\, \Lambda ^*\backslash \, {\mathcal K}(h^{\, l})}
\varphi \| ^2\, \bigr) \leqslant
$$ $$
\leqslant 24\, (\varepsilon )^2\, \bigl( \, \| \widehat G^{\, \frac 12}_+P^{\, (+)}
\varphi \| ^2+\| \widehat G_-^{\, \frac 12}P^{\, (-)}\varphi \| ^2+\| \widehat G_-
^{\, \frac 12}P^{\, \Lambda ^*\backslash \, {\mathcal K}(h^{\, l})}\varphi \| ^2\, 
\bigr) \leqslant
$$ $$
\leqslant (\varepsilon ^{\, \prime })^2\, \bigl( \, \| \widehat G^{\, \frac 12}_+
\widehat P^{\, +}\varphi \| ^2+\| \widehat G_-^{\, \frac 12}\widehat P^{\, -}
\varphi \| ^2\, \bigr) \, .
$$
It remains to remove the technical lower bound $h>2\pi d(K^*)$ for the number  
$h$. Under the linear transformations $x\to \lambda x$, $x\in {\mathbb R}^3$, where
$\lambda >0$, we have $\gamma \to \lambda \gamma $ (the vector $\widetilde e\in 
S_1(\gamma )$ does not change), $K^*\to \lambda ^{-1}K^*$, $d(K^*)\to \lambda ^{-1}
d(K^*)$, and
$$
\varepsilon ^{\, \prime }\to \lambda \, \varepsilon ^{\, \prime },\ \ 
\| \widehat W\| ^{(\infty ,\, {\mathrm {loc}})}_{L^3_w(K;{\mathcal M}_M)}\to
\lambda \, \| \widehat W\| ^{(\infty ,\, {\mathrm {loc}})}_{L^3_w(K;{\mathcal M}_M)}
\, ,\ \ \beta _{\gamma ,\, \sigma }(\widehat W)\to \lambda \, \beta _{\gamma ,\, 
\sigma }(\widehat W)\, .
$$
Therefore conditions 3 and 4 from Theorem \ref{th1.6}, and estimate (1.15) do not
change under such transformations. Choosing the number $\lambda >0$ such that $h=64>
\lambda ^{-1}\cdot 2\pi \, d(K^*)$, we conclude that we can take the universal
constant $\widetilde c_1=\widetilde c_1(64)$ and the constant $\widetilde c_1^{\, 
\prime }=\widetilde c_1^{\, \prime }(64,\sigma )$ dependent only on $\sigma $. 
Theorem \ref{th1.6} is proved.

\begin{theorem} \label{th6.2}
Let $d=3$, $\gamma \in \Lambda \backslash \{ 0\} $, $\sigma \in (0,2]$. 
Suppose that $\widehat W\in L^2(K;{\mathcal M}_M)$, $\beta _{\gamma ,\, 
\sigma }(0;\widehat W)<+\infty $, and for a.e. $x\in {\mathbb R}^3$ the matrices
$\widehat W(x)$ commute with all orthogonal projections $\widehat P^{\, \pm }
_{\widetilde e}\, $, $\widetilde e\in S_1(\gamma )$ $\mathrm ($in particular,
we may consider the functions $\widehat W=\widehat W^{\, (0)}+\widehat W^{\, (1)}$,
$\widehat W^{\, (s)}\in L^2(K;{\mathcal S}^{\, (s)}_M)$, $s=0,1$$\mathrm )$.
Then for any $\delta >0$ there is a number $\varkappa ^{\, \sharp}>2h$ such that
for all $\varkappa \geqslant \varkappa ^{\, \sharp}$, all vectors $k\in {\mathbb 
R}^3$ with $|(k,\gamma )|=\pi $, and all vector functions $\psi \in L^2(K;{\mathbb 
C}^M)$ the inequality
$$
\| \, \widehat G_-^{\, -\frac 12}\, \widehat P^{\, (+)}\,  
\widehat W\, \widehat G_-^{\, -\frac 12}\, \widehat P^{\, (-)}\, \psi \, \|
\leqslant \frac 12\, c_6\, (\, \delta +\beta _{\gamma ,\, 
\sigma }(\widehat W))\, \| \, \psi \, \| \, ,  \eqno (6.7)
$$
holds, where $c_6=c_6\, (h,\sigma )>0$ is the constant from Theorem \ref{th6.1} 
$\mathrm ($see $\mathrm ($6.14$\mathrm )$$\mathrm )$.
\end{theorem}

{\it Proof}. To start with, we assume that $\sigma \in (0,\frac 14\, ]$. Let 
$\varkappa >2h$ and let $k\in {\mathbb R}^3$ with $|(k,\gamma )|=\pi $. We shall
derive upper bounds for the norms
$$
\| \, \widehat P^{\, (+)}_{\mu }\, \widehat W\, \widehat P^{\, (-)}_{\nu }\psi \,
\| \, ,\ \mu ,\nu =1,\dots ,l\, ,\ \psi \in L^2(K;{\mathbb C}^M)\, . 
$$
Since the Fourier coefficients $\widehat W_N$, $N\in \Lambda ^*$, commute with all
orthogonal projections $\widehat P^{\, \pm }_{\widetilde e}\, $, $\widetilde e\in 
S_1(\gamma )$, we can use estimate (5.28). For $\mu ,\, \nu \in \{ 1,\dots ,l \} $ 
and $R\in [2\pi d(K^*),2\varkappa ]$, the function $\widehat W$ can be represented
in the form
$$
\widehat W(x)=\widehat W^{\, [0]}_{\mu ,\, \nu }(R;x)+\widehat W^{\, [1]}_{\mu ,\, 
\nu }(R;x)\, ,\ x\in {\mathbb R}^3\, ,
$$
where
$$
\widehat W^{\, [0]}_{\mu ,\, \nu }(R;x)=\sum\limits_{\substack{N\, \in \, \Lambda 
^*\, : \\  2\pi |N_{\perp}|\, \leqslant \, R}}\widehat W_N\, e^{\, 2\pi i\, (N,x)}
\, ,\ \, \widehat W^{\, [1]}_{\mu ,\, \nu }(R;x)=\sum\limits_{\substack{N\, \in \, 
\Lambda ^*\, : \\  R\, <\, 2\pi |N_{\perp}|}}\widehat W_N\, e^{\, 2\pi i\, (N,x)}\, .
$$
By (5.28), we get
$$
\| \, \widehat P^{\, (+)}_{\mu }\, \widehat W^{\, [0]}_{\mu ,\, \nu }(R;.)\, 
\widehat P^{\, (-)}_{\nu }\psi \, \| ^2\leqslant  \eqno (6.8)
$$ $$
\leqslant v(K)\, \varkappa ^{-2}\sum\limits_{N\, \in \, {\mathcal K}_{\mu}}\, \bigl( 
\sum\limits_{\substack{n\, \in \, \Lambda ^*\, : \\ 2\pi |n_{\perp}|\, \leqslant \, R
\, , \\ N-n\, \in \, {\mathcal K}_{\nu}}}2\pi |n_{\perp}|\, \| \widehat W_n\| \,
\| (\widehat P^{\, -}\psi )_{N-n}\| \, \bigr) ^2 \leqslant
$$ $$
\leqslant v(K)\, \varkappa ^{-2}\, \bigl( \sum\limits_{\substack{n\, \in \, \Lambda 
^*\, : \\ 2\pi |n_{\perp}|\, \leqslant \, R \, , \\ \pi |n_{\parallel}|\, \leqslant \, 
h^{\, \max \, \{ \mu ,\nu \} }}}2\pi |n_{\perp}|\, \| \widehat W_n\| \, \bigr) ^2
\, \bigl( \, \sum\limits_{N\, \in \, \Lambda ^*}\| (\widehat P^{\, (-)}_{\nu }\psi )
_N\| ^2\, \bigr) \leqslant
$$ $$
\leqslant R^{\, 2}\, \varkappa ^{-2}\, \bigl( \sum\limits_{\substack{n\, \in \, 
\Lambda ^*\, : \\ 2\pi |n_{\perp}|\, \leqslant \, R \, , \\ \pi |n_{||}|\, \leqslant 
\, h^{\, \max \, \{ \mu ,\nu \} }}}1 \, \bigr) \, \bigl( \, \sum\limits_{n\, \in \, 
\Lambda ^*}\| \widehat W_n\| ^2 \bigr) \, \| \widehat P^{\, (-)}_{\nu }\psi \| 
^2 \leqslant
$$ $$
\leqslant \, 3\pi ^{-2}M\, R^{\, 4}\, h^{\, \max \, \{ \mu ,\nu \} }\, 
\varkappa ^{-2}\, \| \widehat W\| ^2_{L^2(K;{\mathcal M}_M)}\, \| \widehat P^{\, (-)}
_{\nu }\psi \| ^2\, .
$$
The following estimate is also a consequence of (5.28):
$$
\| \, \widehat P^{\, (+)}_{\mu }\, \widehat W^{\, [1]}_{\mu ,\, \nu }(R;.)\, 
\widehat P^{\, (-)}_{\nu }\psi \, \| ^2\leqslant  \eqno (6.9)
$$ $$
\leqslant c_4\, v^{-1}(K)\, h^{\, \mu +\nu +\min \, \{ \mu ,\nu \} }\, \varkappa 
^{-\frac 12}\, \beta ^2_{\gamma ,\, \sigma }(R;\widehat W)\, \times 
$$ $$
\times \ \biggl( \, \sum\limits_{\substack{n\, \in \, \Lambda ^*\, : \\ R\, <\, 2\pi 
|n_{\perp}|\, \leqslant \, 2\varkappa +2h^{\, \max \, \{ \mu ,\nu \} }\, , \\ \pi 
|n_{\parallel}|\, \leqslant \, h^{\, \max \, \{ \mu ,\nu \} } }}\frac {(2\pi 
|n_{\perp}|)^{-2(1-\sigma )}\, (2\pi |n|)^{-2\sigma }}{(\pi |n_{\perp}|+h^{\, \max \, 
\{ \mu ,\nu \} })\, \sqrt {\varkappa +2h^{\, \max \, \{ \mu ,\nu \} }-
\pi |n_{\perp}|}}\ \biggr) \, \| \widehat P^{\, (-)}_{\nu}\psi \| ^2\, .
$$
Under the condition $R<2h^{\, \max \, \{ \mu ,\nu \} }$, we have
$$
\sum\limits_{\substack{n\, \in \, \Lambda ^*\, : \\ R\, <\, 2\pi 
|n_{\perp}|\, \leqslant \, 2h^{\, \max \, \{ \mu ,\nu \} }\, , \\ \pi 
|n_{\parallel}|\, \leqslant \, h^{\, \max \, \{ \mu ,\nu \} } }}(2\pi |n_{\perp}|)
^{-2(1-\sigma )}\, (2\pi |n|)^{-2\sigma }\leqslant  
$$ $$
\leqslant \, \frac 4{(2\pi )^3\, v(K^*)}\ \biggl( 4\pi \, (\, h^{\, \max \, \{ \mu 
,\nu \} }+\pi d(K^*)\, )\, +
$$ $$
+\, \int_R^{\, 2h^{\, \max \, \{ \mu ,\nu \} }\, +\, 2\pi d(K^*)}\, \frac {2\pi \xi 
\, d\xi }{\xi ^{\, 2(1-\sigma )}}\ \int_{-2h^{\, \max \, \{ \mu ,\nu \} }\, -\, 
2\pi d(K^*)}^{\, 2h^{\, \max \, \{ \mu ,\nu \} }\, +\, 2\pi d(K^*)}\frac {d\eta }
{(\xi ^2+\eta ^2)^{\sigma }}\, \biggr) \leqslant
$$ $$
\leqslant \pi ^{-2}\, v^{-1}(K^*)\ \biggl( 3\, h^{\, \max \, \{ \mu ,\nu \} }+ 
\int_R^{\, 3h^{\, \max \, \{ \mu ,\nu \} }}\frac {d\xi }{\xi ^{\, 1-2\sigma }}\ 
\int_{-3h^{\, \max \, \{ \mu ,\nu \} }}^{\, 3h^{\, \max \, \{ \mu ,\nu \} }}\frac 
{d\eta }{(\xi ^2+\eta ^2)^{\sigma }}\, \biggr) <
$$ $$
< 8\, \pi ^{-2}\, \sigma ^{-1}\, v^{-1}(K^*)\, h^{\, \max \, \{ \mu ,\nu 
\} }
$$ 
and
$$
\sum\limits_{\substack{n\, \in \, \Lambda ^*\, : \\ h^{\, \max \, \{ \mu ,\nu \} }\, 
<\, \pi |n_{\perp}|\, \leqslant \, \varkappa +h^{\, \max \, \{ \mu ,\nu \} }\, , \\ 
\pi |n_{\parallel}|\, \leqslant \, h^{\, \max \, \{ \mu ,\nu \} } }}\frac 
{(2\pi |n_{\perp}|)^{2\sigma }\, (2\pi |n|)^{-2\sigma }}{(2\pi |n_{\perp}|)^3\,
\sqrt {\varkappa +2h^{\, \max \, \{ \mu ,\nu \} }-\pi |n_{\perp}|}}\ \leqslant  
$$ $$
\leqslant \, \sum\limits_{\substack{n\, \in \, \Lambda ^*\, : \\ h^{\, \max \, 
\{ \mu ,\nu \} }\, <\, \pi |n_{\perp}|\, \leqslant \, \varkappa +h^{\, \max \, 
\{ \mu ,\nu \} }\, , \\ \pi |n_{\parallel}|\, \leqslant \, h^{\, \max \, \{ 
\mu ,\nu \} } }}\frac 1{(2\pi |n_{\perp}|)^3}\ \frac 1{\sqrt {\varkappa +2h^{\, 
\max \, \{ \mu ,\nu \} }-\pi |n_{\perp}|}}\ \leqslant  
$$ $$
\leqslant \, \frac 4{(2\pi )^3\, v(K^*)}\, \left( \frac 32 \right) ^{\frac 52}\, 
( h^{\, \max \, \{ \mu ,\nu \} }+\pi d(K^*) )\, \times
$$ $$
\times \, \int_{2h^{\, \max \, \{ \mu ,\nu 
\} }\, -\, 2\pi d(K^*)}^{\, 2\varkappa \, +\, 2h^{\, \max \, \{ \mu ,\nu \} }\, +\,
2\pi d(K^*)}\ \frac {2\pi \xi \, d\xi }{\xi ^{\, 3}\, \sqrt {\varkappa +2h^{\, \max 
\, \{ \mu ,\nu \} }-\frac 12\, \xi }} \leqslant
$$ $$
\leqslant \, \frac {\sqrt 3}{\pi ^2}\, \left( \frac 32 \right) ^4 v^{-1}(K^*)\,
h^{\, \max \, \{ \mu ,\nu \} }\, \times
$$ $$
\times \, \biggl( \varkappa ^{-\frac 12}\ \biggl| ^{\, \xi \, =\, \varkappa }
_{\, \xi \, =\, h^{\, \max \, \{ \mu ,\nu \} }}-\frac 1{\xi }\ +\ 2\varkappa ^{-2}\
\biggr| ^{\, \xi \, =\, 2\varkappa \, +\, 3h^{\, \max \, \{ \mu ,\nu \} }}_{\, \xi 
\, =\, \varkappa }-\sqrt {2\varkappa +4h^{\, \max \, \{ \mu ,\nu \} }-\xi }\, 
\biggr) <
$$ $$
< \, 34\, \pi ^{-2}\, v^{-1}(K^*)\, \varkappa ^{-\frac 12}\, .
$$
Hence, (6.9) yields
$$
\| \, \widehat P^{\, (+)}_{\mu }\, \widehat W^{\, [1]}_{\mu ,\, \nu }(R;.)\, 
\widehat P^{\, (-)}_{\nu }\psi \, \| ^2\leqslant  \eqno (6.10)
$$ $$
\leqslant \, 25\, \pi ^{-2}\, c_4\, \sigma ^{-1}\, h^{\, \mu +
\nu +\min \, \{ \mu ,\nu \} }\, \varkappa ^{-1}\, \beta ^{\, 2}_{\gamma ,\, \sigma }
(R;\widehat W)\, \| \widehat P^{\, (-)}_{\nu}\psi \| ^2 \, .
$$
From (6.8) and (6.10) (also see (5.3) and (5.4)), for all $\varkappa >2h$, all
vectors $k\in {\mathbb R}^3$ with $|(k,\gamma )|=\pi $, all vector functions
$\psi \in L^2(K;{\mathbb C}^M)$, and all numbers $R\in [2\pi d(K^*),2\varkappa ]$
we obtain
$$
\| \, \widehat G_-^{-\frac 12}\, \widehat P^{\, (+)}\, \widehat W\, \widehat G_-
^{-\frac 12}\, \widehat P^{\, (-)}\psi \, \| \leqslant \sum\limits_{\mu ,\, \nu \, =\, 
1}^l\| \, \widehat G_-^{-\frac 12}\, \widehat P^{\, (+)}_{\mu }\, \widehat W\, 
\widehat G_-^{-\frac 12}\, \widehat P^{\, (-)}_{\nu }\psi \, \| 
\leqslant  \eqno (6.11)
$$ $$
\leqslant \, \sqrt {\frac {|\gamma |}{\pi }}\ \sum\limits_{\nu \, =\, 1}^l
\, \| \, \widehat P^{\, (+)}_1\, (\widehat W^{\, [0]}_{1,\, \nu }(R;.)+\widehat W^{\, 
[1]}_{1,\, \nu }(R;.))\, \widehat G_-^{-\frac 12}\, \widehat P^{\, (-)}_{\nu }\psi 
\, \| \, +
$$ $$
+ \, \sum\limits_{\mu \, =\, 2}^l\, \sum\limits_{\nu \, =\, 1}^lh^{\, -\frac {\mu 
-1}2}\, \| \, \widehat P^{\, (+)}_{\mu }\, (\widehat W^{\, [0]}_{\mu ,\, \nu }(R;.)+
\widehat W^{\, [1]}_{\mu ,\, \nu }(R;.))\, \widehat G_-^{-\frac 12}\, \widehat 
P^{\, (-)}_{\nu }\psi \, \| \leqslant  
$$ $$
\leqslant \, \sqrt {\frac {|\gamma |}{\pi }}\ \,  
\sum\limits_{\nu \, =\, 1}^l\, \biggl( \frac {\sqrt 
{3M}}{\pi }\ R^{\, 2}\, h^{\, \frac 12\, \nu }\, \varkappa ^{-1}\, 
\| \widehat W\| _{L^2(K;{\mathcal M}_M)}\, +
$$ $$
+\, \frac 5{\pi {\sqrt \sigma }}\ c_4 ^{\, 
\frac 12}\, h^{\, 1+\frac 12\, \nu }\, \varkappa ^{-\frac 12}\, \beta _{\gamma ,\, 
\sigma }(R;\widehat W) \biggr) \, \| \, \widehat G_-^{-\frac 12}\, \widehat P^{\, 
(-)}_{\nu}\psi \, \| \, +
$$ $$
+ \,  \sum\limits_{\mu \, =\, 2}^l\,
\sum\limits_{\nu \, =\, 1}^lh^{\, -\frac {\mu -1}2}\, \biggl( \frac {\sqrt 
{3M}}{\pi }\ R^{\, 2}\, h^{\, \frac 12\, \max \, \{ \mu ,\nu \} }\, \varkappa ^{-1}\, 
\| \widehat W\| _{L^2(K;{\mathcal M}_M)}\, +
$$ $$
+\, \frac 5{\pi {\sqrt \sigma }}\ c_4 ^{\, \frac 12}\, h^{\, \frac 12\, (\mu +\nu 
+\min \, \{ \mu ,\nu \} )}\, \varkappa ^{-\frac 12}\, \beta _{\gamma ,\, \sigma }
(R;\widehat W) \biggr) \, \| \, \widehat G_-^{-\frac 12}\, \widehat P^{\, (-)}
_{\nu}\psi \, \| \leqslant
$$ $$
\leqslant \, \frac {\sqrt h}{\pi }\ \biggl( \, \frac {|\gamma |}{\pi }+2(l-1)\,
\sqrt {\frac {|\gamma |}{\pi }}\ \biggr)\, \bigl( \sqrt {3M}\, R^{\, 2}\, 
\varkappa ^{-1}\, \| \widehat W\| _{L^2(K;{\mathcal M}_M)}\, +
$$ $$
+\, 5\, c_4^{\, \frac 12}\, \sigma ^{-\frac 12}\, h\, \varkappa ^{-\frac 12}\, 
\beta _{\gamma ,\, \sigma }(R;\widehat W) \bigr) \, \| \, \widehat P^{\, 
(-)}_{\nu}\psi \, \| \, +
$$ $$
+ \, \frac h{\pi }\ \sum\limits_{\mu ,\, \nu \, =\, 2}^l\, \bigl( \sqrt {3M}\, R^{\, 
2}\, h^{\, -\frac 12\, \min \, \{ \mu ,\nu \} }\, \varkappa ^{-1}\, \| \widehat W\| 
_{L^2(K;{\mathcal M}_M)}\, +
$$ $$
+\, 5\, c_4^{\, \frac 12}\, \sigma ^{-\frac 12}\, h^{\, \frac 12\, \min \, \{ \mu 
,\nu \} }\, \varkappa ^{-\frac 12}\, \beta _{\gamma ,\, \sigma }(R;\widehat W) 
\biggr) \, \| \, \widehat P^{\, (-)}_{\nu}\psi \, \| \, .
$$
On the other hand, $l\leqslant (\ln ^{-1} h)\, \ln \, \frac {\varkappa }2$ and
$$
\frac h{\varkappa }\ \sum\limits_{\mu ,\, \nu \, =\, 1}^l\, h^{\, -\frac 
12\, \min \, \{ \mu ,\nu \} }\, \leqslant \, h^{\, \frac 12}\, l^{\, 2}\, \varkappa 
^{-1}\leqslant (h^{\, \frac 12}\, \ln ^{-1} h)\, \varkappa ^{-1}\ln \, \frac 
{\varkappa }2\ ,  \eqno (6.12)
$$
\vskip 0.1cm
$$
\frac h{\sqrt {\varkappa }}\ \sum\limits_{\mu ,\, \nu \, =\, 2}^l\, h^{\, \frac 
12\, \min \, \{ \mu ,\nu \} }\, =\, \frac h{\sqrt {\varkappa }}\ \, h^{\, \frac 12\, 
l}\, \sum\limits_{\mu _1,\, \nu _1\, =\, 0}^{l-2}\, h^{\, -\frac 12\, \max \, \{ 
\mu _1,\nu _1\} }\leqslant
$$ $$
\leqslant \, \frac h{\sqrt 2}\ \sum\limits_{\mu _1\, =\, 0}^{+\infty }(2\mu _1+1)
\, h^{\, -\frac 12\, \mu _1}\, =\, \frac h{\sqrt 2}\ \, (1+h^{\, -\frac 12}\, )
(1-h^{\, -\frac 12}\, )^{-2} <2h
$$
(because $h>16$). Therefore (for $\sigma \in (0,\frac 14\, ]$),
$$
\| \, \widehat G_-^{-\frac 12}\, \widehat P^{\, (+)}\, \widehat W\, \widehat 
G_-^{-\frac 12}\, \widehat P^{\, (-)}\psi \, \| \leqslant  \eqno (6.13)
$$ $$
\leqslant \, \bigl( \, 10\, \pi ^{-1}\, c_4^{\, \frac 12}\, \sigma ^{-\frac 12}\, 
h\, \beta _{\gamma ,\, \sigma }(R;\widehat W) +\mathfrak F\, (\sigma ,h,R,\widehat 
W;\varkappa )\, \bigr) \, \| \, \widehat P^{\, (-)}\psi \, \| \, ,
$$
where
$$
\mathfrak F\, (\sigma ,h,R,\widehat W;\varkappa )=\pi ^{-1}\, \sqrt h\, (\, \ln
^{-1}h\, )\, (\, \ln \frac {\varkappa }2\, )\, \biggl( 1+\sqrt {\frac {|\gamma |}
{\pi }}\ \biggr) ^2\, \times
$$ $$
\times \, \bigl( \sqrt {3M}\, R^{\, 2}\, \varkappa ^{-1} 
\, \| \widehat W\| _{L^2(K;{\mathcal M}_M)} +
5\, c_4^{\, \frac 12}\, \sigma ^{-\frac 12}\, h\, \varkappa ^{\, -\frac 12}\,         
\beta _{\gamma ,\, \sigma }(R;\widehat W) \bigr) \, .
$$
Now let us suppose that $\sigma \in (0,2]$. Denote 
$$
c_6=40\, {\sqrt 2}\, \pi ^{-1}\, c_4^{\, \frac 12}\, h\,  \sigma ^{-\frac 12}
\eqno (6.14)
$$
and $\sigma ^{\, \prime }\doteq \min \, \{ \frac 14\, ,\sigma \} $. Since
$(\sigma ^{\, \prime })^{-\frac 12}\leqslant 2\sqrt 2\, \sigma ^{-\frac 12}$ and
$\beta _{\gamma ,\, \sigma ^{\, \prime }}(R;\widehat W)\leqslant \beta 
_{\gamma ,\, \sigma }(R;\widehat W)$ (for all $R\geqslant 0$), inequality
(6.13) (in which we replace $\sigma $ by $\sigma ^{\, \prime }$) implies the
estimate
$$
\| \, \widehat G_-^{-\frac 12}\, \widehat P^{\, (+)}\, \widehat W\, \widehat 
G_-^{-\frac 12}\, \widehat P^{\, (-)}\psi \, \| \leqslant  
\bigl( \frac 12\, c_6\, \beta _{\gamma ,\, \sigma }(R;\widehat W) +\mathfrak F\, 
(\sigma ,h,R,\widehat W;\varkappa )\, \bigr) \, \| \, \widehat P^{\, (-)}\psi \, 
\| \, .  \eqno (6.15)
$$
Finally, choose (and fix) a number $R\geqslant 2\pi d(K^*)$ such that
$$
\beta_{\gamma ,\, \sigma }(R;\widehat W)\leqslant \frac {\delta }2+\beta_{\gamma ,\, 
\sigma }(\widehat W)\, .
$$
Then, from (6.15) it follows that there is a number $\varkappa ^{\, \sharp }=
\varkappa ^{\, \sharp }(M,\Lambda ,|\gamma |,h,\sigma ;\widehat W,\delta )>2h$, for
which $\varkappa ^{\, \sharp }\geqslant \frac 12\, R$, such that for all
$\varkappa \geqslant \varkappa ^{\, \sharp }\, $, all $k\in {\mathbb R}^3$ with 
$|(k,\gamma )|=\pi $, and all $\psi \in L^2(K;{\mathbb C}^M)$ the inequality (6.7)
holds. This completes the proof of Theorem \ref{th6.2}.
\vskip 0.2cm

Let $l^{\, \prime }\in {\mathbb N}\backslash \{ 1\} $ and let $\varkappa \geqslant
2h^{\, l^{\, \prime }+1}$. Then $l=l(h,\varkappa )\geqslant l^{\, \prime }+1
\geqslant 3$. Denote
$$
l_1\doteq l-l^{\, \prime }\in \{ 1,\dots ,l-2\}
$$
and define the functions
$$ 
\Theta _{l^{\, \prime }}(h,\varkappa ;t)=\left\{
\begin{array}{llll}
1 & {\text {if}} \ \, h^{\, l_1+1}<t\leqslant h^{\, l-1}\, , \\ [0.2cm]
2-h^{-l+1}t & {\text {if}} \ \, h^{\, l-1}<t\leqslant 2h^{\, l-1}\, , \\ [0.2cm]
-1+2h^{-l_1-1}t & {\text {if}} \ \, \frac 12\, h^{\, l_1+1}<t\leqslant 
h^{\, l_1+1}\, , \\ [0.2cm]
0 & {\text {if}}\ \, t\leqslant \frac 12\, h^{\, l_1+1}\ \, 
{\text {or}}\, \ t>2h^{\, l-1}\, ,
\end{array}
\right.
$$
and the operators $\widehat \Theta _{l^{\, \prime }}=\widehat \Theta _{l^{\, 
\prime }}(h,k;\varkappa )$ that take a vector function $\psi \in L^2(K;{\mathbb C}
^M)$ to the vector function 
$$
\widehat \Theta _{l^{\, \prime }}\, \psi =\sum\limits_{N\, \in \, \Lambda ^*}\, 
\Theta _{l^{\, \prime }}(h,\varkappa ;G^{\, -}_N(k;\varkappa ))\, \psi _N\, 
e^{\, 2\pi i\, (N,x)}.
$$

\begin{theorem} \label{th6.3}
Let $d=3$, $\gamma \in \Lambda \backslash \{ 0\} $, $\sigma \in (0,2]$. 
Suppose that $\widehat V=\widehat V^{\, (0)}+\widehat V^{\, (1)}$, where 
$\widehat V^{\, (s)}\in L^2(K;{\mathcal S}^{(s)}_M)$, $s=0,1$, and $\beta 
_{\gamma ,\, \sigma }(0;\widehat V)<+\infty $. Then for any $\widetilde
\varepsilon >0$ there are numbers $l^{\, \prime }=l^{\, \prime }(\Lambda ,|\gamma |,
\sigma , \widehat V;\widetilde \varepsilon )\in {\mathbb N}\backslash \{ 1\} $ and
$\varkappa ^{\, \sim }_0=\varkappa ^{\, \sim }_0(M,\Lambda ,|\gamma |,h,
\sigma , \widehat V;\widetilde \varepsilon )\geqslant 2h^{\, l^{\, \prime }+1}$
such that for all $\varkappa \geqslant \varkappa ^{\, \sim }_0$, all vectors 
$k\in {\mathbb R}^3$ with $|(k,\gamma )|=\pi $, and all vector functions 
$\psi \in L^2(K;{\mathbb C}^M)$ the inequality
$$
\| \, \widehat G_-^{\, -\frac 12} \widehat P^{\, +} \widehat \Theta \, 
\widehat V \widehat G_-^{\, -\frac 12} \widehat P^{\, -} \widehat \Theta  
\psi \, -\, \widehat G_-^{\, -\frac 12} \widehat P^{\, +} \widehat \Theta _{\, l^{\, 
\prime }} \widehat V \widehat G_-^{\, -\frac 12} \widehat P^{\, -} \widehat 
\Theta _{\, l^{\, \prime }} \psi \, \| \leqslant \widetilde \varepsilon \, \| \, 
\widehat P^{\, -} \widehat \Theta \, \psi \, \|  \eqno (6.16)
$$
holds.
\end{theorem}

{\it Proof}. Taking into account the inclusions
$$
\widehat \Theta _{\, l^{\, \prime }} \varphi \in {\mathcal H}\, ( \bigcup\limits
_{\mu \, =\, l_1+1}^l{\mathcal K}_l\, )\, ,\ \ (\widehat \Theta - \widehat \Theta _{\, 
l^{\, \prime }} )\varphi \in {\mathcal H}\, (\,  \bigcup\limits_{\mu \, =\, 1}^{l_1+1}
{\mathcal K}_l\, )\, ,\ \ \varphi \in L^2(K;{\mathbb C}^M)\, ,
$$
and estimates (5.3), (5.4), from (6.8) and (6.10), where we put $R=2\pi d(K^*)$, 
it follows that for all $\varkappa \geqslant 2h^{\, l^{\, \prime }+1}$, all
vectors $k\in {\mathbb R}^3$ with $|(k,\gamma )|=\pi $, and all vector functions 
$\psi \in L^2(K;{\mathbb C}^M)$ the following inequalities are valid (see (6.11)):
$$
\| \, \widehat G_-^{\, -\frac 12} \widehat P^{\, +} \widehat \Theta \, 
\widehat V \widehat G_-^{\, -\frac 12} \widehat P^{\, -} \widehat \Theta  
\psi \, -\, \widehat G_-^{\, -\frac 12} \widehat P^{\, +} \widehat \Theta _{\, l^{\, 
\prime }} \widehat V \widehat G_-^{\, -\frac 12} \widehat P^{\, -} \widehat 
\Theta _{\, l^{\, \prime }} \psi \, \| \leqslant  \eqno (6.17)
$$ $$
\leqslant \| \, \widehat G_-^{\, -\frac 12} \widehat P^{\, +} (\widehat \Theta 
-\widehat \Theta _{\, l^{\, \prime }})\, \widehat V \widehat G_-^{\, -\frac 12} 
\widehat P^{\, -} \widehat \Theta \psi \, \| + \| \, \widehat G_-^{\, -\frac 12} 
\widehat P^{\, +} \widehat \Theta _{\, l^{\, \prime }} \widehat V \widehat G_-^{\, 
-\frac 12} \widehat P^{\, -} (\widehat \Theta -\widehat \Theta _{\, l^{\, \prime }}) 
\psi \, \| \leqslant 
$$ $$
\leqslant \sum\limits_{\mu \, =\, 1}^{l_1+1}\ \sum\limits_{\nu \, =\, 1}^l\
\| \, \widehat G_-^{\, -\frac 12} \widehat P^{\, (+)}_{\mu }\, \widehat V \widehat 
G_-^{\, -\frac 12} \widehat P^{\, (-)}_{\nu } \widehat \Theta \psi \, \| +
$$ $$
+\sum\limits_{\mu \, =\, l_1+1}^l\ \sum\limits_{\nu \, =\, 1}^{l_1+1}\
\| \, \widehat G_-^{\, -\frac 12} \widehat P^{\, (+)}_{\mu }\, \widehat V \widehat 
G_-^{\, -\frac 12}\widehat P^{\, (-)}_{\nu } (\widehat \Theta -\widehat \Theta _{\, 
l^{\, \prime }})\, \psi \, \| \leqslant
$$ $$
\leqslant \, \frac h{\pi }\ \max \ \bigl\{ 1, \frac {|\gamma |}{\pi }\, \bigr\} \ 
\sqrt {3M}\, (2\pi d(K^*))^{\, 2}\, \| \widehat V\| _{L^2(K;{\mathcal M}_M)}\,
\varkappa ^{-1}\, \times
$$ $$
\times \, \biggl( \biggl( \ \sum\limits_{\mu \, =\, 1}^{l_1+1}\ \sum\limits_{\nu 
\, =\, 1}^l\, +\sum\limits_{\mu \, =\, l_1+1}^l\ \sum\limits_{\nu \, =\, 1}^{l_1+1}\ 
\biggr) \ h^{\, -\frac 12\, \min \, \{ \mu ,\nu \} }\biggr) \, \| \, \widehat P^{\, 
-} \widehat \Theta \psi \, \| +
$$ $$
+\, \frac h{\pi }\ \max \ \bigl\{ 1, \frac {|\gamma |}{\pi }\, \bigr\} \ 5\, 
c_4^{\, \frac 12}\, \sigma ^{-\frac 12 }\, \beta _{\gamma ,\, \sigma }
(0;\widehat V)\, \varkappa ^{-\frac 12}\, \times
$$ $$
\times \, \biggl( \biggl( \ \sum\limits_{\mu \, =\, 1}^{l_1+1}\ \sum\limits_{\nu 
\, =\, 1}^l\, +\sum\limits_{\mu \, =\, l_1+1}^l\ \sum\limits_{\nu \, =\, 1}^{l_1+1}\ 
\biggr) \ h^{\, \frac 12\, \min \, \{ \mu ,\nu \} }\biggr) \, \| \, \widehat P^{\, 
-} \widehat \Theta \psi \, \| \, .
$$
At the same time, the estimates (6.12) and
$$
\frac h{\sqrt {\varkappa }}\ \biggl( \ \sum\limits_{\mu \, =\, 1}^{l_1+1}\ 
\sum\limits_{\nu \, =\, 1}^l\, +\sum\limits_{\mu \, =\, l_1+1}^l\ \sum\limits_{\nu 
\, =\, 1}^{l_1+1}\ \biggr) \ h^{\, \frac 12\, \min \, \{ \mu ,\nu \} }\, <
\, \frac {2h}{\sqrt {\varkappa }}\ \sum\limits_{\mu \, =\, 1}^{l_1+1}\ 
\sum\limits_{\nu \, =\, 1}^l\ h^{\, \frac 12\, \min \, \{ \mu ,\nu \} }\, <
$$ $$
<\, 2h\, (h^{\, \frac 12\, l}\varkappa ^{-\frac 12})\, \sum\limits_{\mu _1\, =\, 0}
^{+\infty }\, (\, l^{\, \prime }+2\mu _1)\, h^{\, -\frac 12\, (\, l^{\, \prime }+
\mu _1-1)} < 2\sqrt 2\, (l^{\, \prime }+2)\, h^{\, \frac 12\, (3-l^{\, \prime })}
$$
hold (because $h>4$). Therefore, inequality (6.17) implies that there are numbers 
$l^{\, \prime }=l^{\, \prime }(\Lambda ,|\gamma |,\sigma , \widehat V;\widetilde 
\varepsilon )\in {\mathbb N}\backslash \{ 1\} $ and $\varkappa ^{\, \sim }_0=\varkappa 
^{\, \sim }_0(M,\Lambda ,|\gamma |,h,\sigma , \widehat V;\widetilde \varepsilon )
\geqslant 2h^{\, l^{\, \prime }+1}$ such that inequality (6.16) is fulfilled for
all $\varkappa \geqslant \varkappa ^{\, \sim }_0$, all $k\in {\mathbb R}^3$ with 
$|(k,\gamma )|=\pi $, and all $\psi \in L^2(K;{\mathbb C}^M)$. Theorem \ref{th6.3} 
is proved.
\vskip 0.2cm

{\it Proof} of Theorem \ref{th6.1}. From Theorem \ref{th6.3} (under the change
$\psi =\widehat G_-^{\, \frac 12}\widehat P^{\, (-)}\varphi $, $\varphi \in
L^2(K;{\mathbb C}^M)$), we see that it suffices to prove that for any numbers 
$\delta >0$ and $l^{\, \prime }\in {\mathbb N}\backslash \{ 1\} $ there is a number 
$\widetilde \varkappa _0^{\, \prime }(\delta ,l^{\, \prime })=\widetilde 
\varkappa _0^{\, \prime }(M,\Lambda ,|\gamma |,h,\sigma ,\widehat V^{\, (0)},
\widehat V^{\, (1)};\delta ,l^{\, \prime })\geqslant 2h^{\, l^{\, \prime }+1}$
such that for all $\varkappa \geqslant \widetilde \varkappa _0^{\, \prime }
(\delta ,l^{\, \prime })$, all $k\in {\mathbb R}^3$ with $|(k,\gamma )|=\pi $, and
all $\psi \in L^2(K;{\mathbb C}^M)$ the inequalities
$$
\| \, \widehat G_-^{\, -\frac 12}\, \widehat P^{\, +}\, \widehat \Theta _{\, l^{\, 
\prime }}\, \widehat V^{\, (s)}\, \widehat G_-^{\, -\frac 12}\, \widehat P^{\, -}\, 
\widehat \Theta _{\, l^{\, \prime }}\, \psi \, \| \leqslant  \eqno (6.18)
$$ $$
\leqslant \, c_6\, (\, \delta +\max\limits_{q\, =\, 1,\dots ,Q_s}\, \beta _{\gamma ,\, 
\sigma }(\widehat V^{\, (s)}_q) )\, \| \, \widehat P^{\, (-)}\psi \, \| \, ,\ 
\, s=0,1\, ,  
$$
hold. Fix numbers $\delta >0$ and $l^{\, \prime }\in {\mathbb N}\backslash \{ 1\} $.
For $s=0,1$ and $q=1,\dots ,Q_s$, suppose ${\mathcal F}^{\, (s)}_q\in C^{\infty }
({\mathbb R}^3;{\mathbb R})$ are $\Lambda $-periodic functions such that $0\leqslant 
{\mathcal F}^{\, (s)}_q(x)\leqslant 1$ for all $x\in {\mathbb R}^3$, ${\mathcal F}
^{\, (s)}_q(x)=1$ for $x\in {\mathrm {supp}}\, {\mathcal F}^{\, (s)}_q$, and 
${\mathrm {supp}}\, {\mathcal F}^{\, (s)}_{q_1}\cap {\mathrm {supp}}\, {\mathcal F}
^{\, (s)}_{q_2}=\emptyset $ for $q_2\neq q_2\, $, $s=0,1$. Let us denote $\widetilde 
\delta =16\, h^{\, -l^{\, \prime }-1}$ (then $\widetilde \delta \varkappa <32\, h^{\, 
l_1}$, where $l_1=l-l^{\, \prime }$). We define the functions
$$
{\mathcal F}^{\, (s,1)}_q(x)=\sum\limits_{\substack{N\, \in \, \Lambda 
^*\, : \\  2\pi |N|\, \leqslant \, \widetilde \delta \varkappa }}({\mathcal F}^{\, 
(s)}_q)_N\, e^{\, 2\pi i\, (N,x)}\, ,\ \ {\mathcal F}^{\, (s,2)}_q(x)=\sum\limits
_{\substack{N\, \in \, \Lambda ^*\, : \\  2\pi |N|\, > \, \widetilde \delta \varkappa 
}}({\mathcal F}^{\, (s)}_q)_N\, e^{\, 2\pi i\, (N,x)}\, ,\ \ x\in {\mathbb R}^3\, .
$$
For all $\beta \geqslant 0$ (and all $q=1,\dots ,Q_s\, $, $s=0,1$),
$$
\varkappa ^{\, \beta }\, \| {\mathcal F}^{\, (s,2)}_q \| _{L^{\infty}({\mathbb R}
^3;{\mathbb R})}\to 0  \eqno (6.19)
$$
as $\varkappa \to +\infty $.

In what follows, we shall use the brief notation
$$
\widehat {\mathcal A}^{\, \pm }_{\, l^{\, \prime }}\doteq \widehat G_-^{\, -\frac
12}\, \widehat P^{\, \pm }\, \widehat \Theta _{\, l^{\, \prime }}\, .  \eqno (6.20)
$$

\begin{lemma} \label{l6.1}
There are constants $c_7(h,l^{\, \prime };{\mathcal F}^{\, (s)}_q)>0$ such that
for all $\varkappa \geqslant 2h^{\, l^{\, \prime }+1}$, all $k\in {\mathbb R}^3$,
and all $\varphi \in L^2(K;{\mathbb C}^M)$ we have
$$
\| \, (\, \widehat {\mathcal A}^{\, \pm }_{\, l^{\, \prime }}\, {\mathcal F}^{\, 
(s,1)}_q-{\mathcal F}^{\, (s,1)}_q\, \widehat {\mathcal A}^{\, \pm }_{\, l^{\, 
\prime }}\, )\, \varphi \, \| \leqslant c_7(h,l^{\, \prime };{\mathcal F}^{\, (s)}_q)
\, \varkappa ^{\, -\frac 32}\, \| \varphi \| \, ,  \eqno (6.21)
$$
$q=1,\dots ,Q_s\, $, $s=0,1$.
\end{lemma}

\begin{proof}
The choice of the number $\widetilde \delta $ implies that
$$
(\, {\mathcal F}^{\, (s,1)}_q\, \widehat {\mathcal A}^{\, \pm }_{\, l^{\, 
\prime }}\, \varphi \, )_N=(\, \widehat {\mathcal A}^{\, \pm }_{\, l^{\, 
\prime }}\, {\mathcal F}^{\, (s,1)}_q\, \varphi \, )_N=0
$$
for $N\in (\Lambda ^*\backslash \, {\mathcal K}(h^{\, l}))\cup {\mathcal K}(h^{\, 
l_1})$,
$$
(\, \widehat {\mathcal A}^{\, \pm }_{\, l^{\, 
\prime }}\, {\mathcal F}^{\, (s,1)}_q\, \varphi \, )_N=  \eqno (6.22)
$$ $$
=\sum\limits_{\substack{n\, \in \, \Lambda 
^*\, : \\  2\pi |n|\, \leqslant \, \widetilde \delta \varkappa \, , \\
N-n\, \in \, {\mathcal K}(h^{\, l})\backslash \, {\mathcal K}(h^{\, l_1})}}
(G^{\, -}_N(k;\varkappa ))^{-\frac 12}\, \widehat P^{\, \pm }_{\widetilde e(k+
2\pi N)}\, \Theta _{\, l^{\, \prime }}(h,\varkappa ;G^{\, -}_N(k;\varkappa ))\,
({\mathcal F}^{\, (s)}_q)_n\, \varphi _{N-n}\, ,  
$$ $$
(\, {\mathcal F}^{\, (s,1)}_q\, \widehat {\mathcal A}^{\, \pm }_{\, l^{\, 
\prime }}\, \varphi \, )_N=  \eqno (6.23)
$$ $$
=\sum\limits_{\substack{n\, \in \, \Lambda 
^*\, : \\  2\pi |n|\, \leqslant \, \widetilde \delta \varkappa \, , \\
N-n\, \in \, {\mathcal K}(h^{\, l})\backslash \, {\mathcal K}(h^{\, l_1})}}
(G^{\, -}_{N-n}(k;\varkappa ))^{-\frac 12}\, \widehat P^{\, \pm }_{\widetilde e(k+
2\pi (N-n))}\, \Theta _{\, l^{\, \prime }}(h,\varkappa ;G^{\, -}_{N-n}(k;\varkappa 
))\, ({\mathcal F}^{\, (s)}_q)_n\, \varphi _{N-n}  
$$
for $N\in {\mathcal K}(h^{\, l})\backslash \, {\mathcal K}(h^{\, l_1})$. Furthermore,
$$
(G^{\, -}_N(k;\varkappa ))^{-\frac 12}\, \widehat P^{\, \pm }_{\widetilde e(k+
2\pi N)}\, \Theta _{\, l^{\, \prime }}(h,\varkappa ;G^{\, -}_N(k;\varkappa ))\, -
$$ $$
-\, (G^{\, -}_{N-n}(k;\varkappa ))^{-\frac 12}\, \widehat P^{\, \pm }_{\widetilde e(k+
2\pi (N-n))}\, \Theta _{\, l^{\, \prime }}(h,\varkappa ;G^{\, -}_{N-n}(k;\varkappa 
))=  \eqno (6.24)
$$ $$
=((G^{\, -}_N(k;\varkappa ))^{-\frac 12}-(G^{\, -}_{N-n}(k;\varkappa ))^{-\frac 
12}\, )\, \widehat P^{\, \pm }_{\widetilde e(k+2\pi N)}\, \Theta _{\, l^{\, \prime }}
(h,\varkappa ;G^{\, -}_N(k;\varkappa ))\, +
$$ $$
+\, (G^{\, -}_{N-n}(k;\varkappa ))^{-\frac 12}\, (\, \widehat P^{\, \pm }_{\widetilde 
e(k+2\pi N)}-\widehat P^{\, \pm }_{\widetilde e(k+2\pi (N-n))}\, )\, \Theta _{\, l^{\, 
\prime }}(h,\varkappa ;G^{\, -}_N(k;\varkappa ))\, +
$$ $$
+\,(G^{\, -}_{N-n}(k;\varkappa ))^{-\frac 12}\, \widehat P^{\, \pm }_{\widetilde e(k+
2\pi (N-n))}\, (\, \Theta _{\, l^{\, \prime }}(h,\varkappa ;G^{\, -}_N(k;\varkappa ))-
\Theta _{\, l^{\, \prime }}(h,\varkappa ;G^{\, -}_{N-n}(k;\varkappa ))\, )\, .
$$
For $N, N-n\in {\mathcal K}(h^{\, l})\backslash \, {\mathcal K}(h^{\, l_1})$,
we have 
$$
|\, (G^{\, -}_N(k;\varkappa ))^{-\frac 12}-(G^{\, -}_{N-n}(k;\varkappa ))^{-\frac 
12}\, | \leqslant \, \frac 12\ h^{\, -\frac 32\, l_1}\cdot 2\pi |n|\, ,
$$ $$
\| \, \widehat P^{\, \pm }_{\widetilde 
e(k+2\pi N)}-\widehat P^{\, \pm }_{\widetilde e(k+2\pi (N-n))}\, \| \leqslant
\frac {2\pi |n|}{\varkappa }\ ,
$$ $$
|\, \Theta _{\, l^{\, \prime }}(h,\varkappa ;G^{\, -}_N(k;\varkappa ))-
\Theta _{\, l^{\, \prime }}(h,\varkappa ;G^{\, -}_{N-n}(k;\varkappa ))\, | \leqslant
2h^{-l_1-1}\cdot 2\pi |n|
$$
(see (1.4), (5.2), and the definition of the functions $\Theta _{\, l^{\, \prime }}
(h,\varkappa ;.)$). Therefore (6.22) -- (6.24) yield
$$
\| \, ((\, \widehat {\mathcal A}^{\, \pm }_{\, l^{\, \prime }}\, {\mathcal F}^{\, 
(s,1)}_q-{\mathcal F}^{\, (s,1)}_q\, \widehat {\mathcal A}^{\, \pm }_{\, l^{\, 
\prime }}\, )\, \varphi \, )_N\, \| \leqslant 3\sqrt 2\, h^{\, \frac 32\, (l^{\,
\prime }+1)}\, \varkappa ^{-\frac 32}\, \sum\limits_{\substack{n\, \in \, \Lambda 
^*\, : \\  2\pi |n|\, \leqslant \, \widetilde \delta \varkappa }}
2\pi |n| \cdot |({\mathcal F}^{\, (s)}_q)_n|\cdot \| \varphi _{N-n}\| \, .
$$ 
From this we obtain that estimates (6.21) hold with constants
$$ 
c_7(h,l^{\, \prime };{\mathcal F}^{\, (s)}_q)= 3{\sqrt 2}\ h^{\, \frac 32\, (l^{\,
\prime }+1)}\, \sum\limits_{n\, \in \, \Lambda ^*}2\pi |n|\cdot |({\mathcal F}^{\, 
(s)}_q)_n|\, .
$$
Lemma \ref{l6.1} is proved.
\end{proof}

If $\widehat W\in L^3_w(K;{\mathcal M}_M)$, then (see (0.6) and (1.1)) there is a
constant $c_8=c_8(\Lambda ,|\gamma |;\widehat W)>0$ such that for all $k\in {\mathbb 
R}^3$ with $|(k,\gamma )|=\pi $, and all $\varphi \in \widetilde H^1(K;{\mathbb C}
^M)$ the following inequality is satisfied:
$$
\| \widehat W\varphi \| \leqslant c_8\, \| \sum\limits_{j=1}^3\widehat \alpha _j\,
\bigl( k_j-i\, \frac {\partial }{\partial x_j}\, \bigr) \varphi \, \| \, .
$$
Therefore, for all vector functions $\varphi \in {\mathcal H}({\mathcal K}(h^{\,
l}))$,
$$
\| \widehat W\varphi \| \leqslant \frac 32\, c_8\varkappa \, \| \varphi \| \, .  
\eqno (6.25)
$$

Now let us obtain inequality (6.18) (for sufficiently large numbers $\varkappa 
\geqslant \widetilde \varkappa ^{\, \prime }_0(\delta ,l^{\, \prime })\geqslant 
2h^{\, l^{\, \prime }+1}$). From (6.19), (6.25), and Lemma \ref{l6.1} it follows
that there exists a number $\widetilde \varkappa ^{\, \prime }_0(\delta ,l^{\, 
\prime })\geqslant 2h^{\, l^{\, \prime }+1}$ (dependent also on $\Lambda $, $|\gamma 
|$, $h$, $\sigma $, on the numbers $Q_s\, $, and on the functions $\widehat V^{\, 
(s)}_q$, ${\mathcal F}^{\, (s)}_q$) such that for all $\varkappa \geqslant \widetilde 
\varkappa ^{\, \prime }_0(\delta ,l^{\, \prime })$, all $k\in {\mathbb R}^3$ with 
$|(k,\gamma )|=\pi $, and all $\psi \in L^2(K;{\mathbb C}^M)$ we have
$$
\| \, \widehat G_-^{\, -\frac 12}\, \widehat P^{\, +}\, \widehat \Theta _{\, l^{\, 
\prime }}\, \widehat V^{\, (s)}\, \widehat G_-^{\, -\frac 12}\, \widehat P^{\, -}\, 
\widehat \Theta _{\, l^{\, \prime }}\, \psi \, \| ^2=
\| \, \widehat {\mathcal A}^{\, +}_{\, l^{\, \prime }}\, \sum\limits_{q\, =\, 1}
^{Q_s}\, ({\mathcal F}^{\, (s,1)}_q+{\mathcal F}^{\, (s,2)}_q)\, \widehat V^{\, (s)}\,
\widehat {\mathcal A}^{\, -}_{\, l^{\, \prime }}\, \psi \, \| ^2 \leqslant
$$ $$
\leqslant \, \frac 87\ \| \, \widehat {\mathcal A}^{\, +}_{\, l^{\, \prime }}\, 
\sum\limits_{q\, =\, 1}^{Q_s}{\mathcal F}^{\, (s,1)}_q\, \widehat V^{\, (s)}\,
\widehat {\mathcal A}^{\, -}_{\, l^{\, \prime }}\, \psi \, \| ^2+
\frac 1{14}\ \delta ^2c_6^2\, \| \psi \| ^2\, \leqslant
$$ $$
\leqslant \, \frac 97\ \| \, \sum\limits_{q\, =\, 1}^{Q_s}{\mathcal F}^{\, (s,1)}_q\,
\widehat {\mathcal A}^{\, +}_{\, l^{\, \prime }}\,  \widehat V^{\, (s)}\,
\widehat {\mathcal A}^{\, -}_{\, l^{\, \prime }}\, \psi \, \| ^2+
\frac 2{14}\ \delta ^2c_6^2\, \| \psi \| ^2\, \leqslant
$$ $$
\leqslant \, \frac {10}7\ \| \, \sum\limits_{q\, =\, 1}^{Q_s}{\mathcal F}^{\, (s)}_q\,
\widehat {\mathcal A}^{\, +}_{\, l^{\, \prime }}\,  \widehat V^{\, (s)}\,
\widehat {\mathcal A}^{\, -}_{\, l^{\, \prime }}\, \psi \, \| ^2+
\frac 3{14}\ \delta ^2c_6^2\, \| \psi \| ^2\, =
$$ $$
=\, \frac {10}7\, \sum\limits_{q\, =\, 1}^{Q_s} \, \| \, {\mathcal F}^{\, (s)}_q\,
\widehat {\mathcal A}^{\, +}_{\, l^{\, \prime }}\,  \widehat V^{\, (s)}\,
\widehat {\mathcal A}^{\, -}_{\, l^{\, \prime }}\, \psi \, \| ^2+
\frac 3{14}\ \delta ^2c_6^2\, \| \psi \| ^2\, \leqslant
$$ $$
\leqslant \, \frac {11}7\, \sum\limits_{q\, =\, 1}^{Q_s} \, \| \, {\mathcal F}^{\, 
(s,1)}_q\, \widehat {\mathcal A}^{\, +}_{\, l^{\, \prime }}\,  \widehat V^{\, (s)}\,
\widehat {\mathcal A}^{\, -}_{\, l^{\, \prime }}\, \psi \, \| ^2+
\frac 4{14}\ \delta ^2c_6^2\, \| \psi \| ^2\, \leqslant
$$ $$
\leqslant \, \frac {12}7\, \sum\limits_{q\, =\, 1}^{Q_s} \, \| \, \widehat {\mathcal A}
^{\, +}_{\, l^{\, \prime }}\, {\mathcal F}^{\, (s,1)}_q\, \widehat V^{\, (s)}\,
\widehat {\mathcal A}^{\, -}_{\, l^{\, \prime }}\, \psi \, \| ^2+
\frac 5{14}\ \delta ^2c_6^2\, \| \psi \| ^2\, \leqslant
$$ $$
\leqslant \, \frac {13}7\, \sum\limits_{q\, =\, 1}^{Q_s} \, \| \, \widehat {\mathcal A}
^{\, +}_{\, l^{\, \prime }}\, \widehat V^{\, (s)}\, \widehat {\mathcal A}^{\, -}
_{\, l^{\, \prime }}\, {\mathcal F}^{\, (s,1)}_q\, \psi \, \| ^2+
\frac 6{14}\ \delta ^2c_6^2\, \| \psi \| ^2\, \leqslant
$$ $$
\leqslant \, 2\, \sum\limits_{q\, =\, 1}^{Q_s} \, \| \, \widehat {\mathcal A}
^{\, +}_{\, l^{\, \prime }}\, \widehat V^{\, (s)}\, \widehat {\mathcal A}^{\, -}
_{\, l^{\, \prime }}\, {\mathcal F}^{\, (s)}_q\, \psi \, \| ^2+
\frac 12\ \delta ^2c_6^2\, \| \psi \| ^2
$$
(we use the notation (6.20)). Finally, (perhaps picking a larger number $\widetilde 
\varkappa ^{\, \prime }_0(\delta ,l^{\, \prime })$ which now may also depend on
$M$) these estimates and Theorem \ref{th6.2} imply that inequalities (6.18) hold:
$$
\| \, \widehat G_-^{\, -\frac 12}\, \widehat P^{\, +}\, \widehat \Theta _{\,
l^{\, \prime }}\, \widehat V^{\, (s)}\, \widehat G_-^{\, -\frac 12}\,
\widehat P^{\, -}\, \widehat \Theta _{\, l^{\, \prime }}\, \psi \, \| ^2
\leqslant 
$$ $$
\leqslant \, \frac 12\ c_6^2\, \sum\limits_{q\, =\, 1}^{Q_s} \, (\, \delta 
+\beta _{\gamma ,\, \sigma }(\widehat V^{\, (s)}_q))^2\, \| \, {\mathcal F}^{\, 
(s)}_q\, \psi \, \| ^2+ \frac 12\, \delta ^2c_6^2\, \| \psi \| ^2\, \leqslant
$$ $$
\leqslant \, \frac 12\ c_6^2\, (\, \delta +\max\limits_{q\, =\, 1,\dots ,Q_s}\, 
\beta _{\gamma ,\, \sigma }(\widehat V^{\, (s)}_q))^2\, \sum\limits_{q\, =\, 1}^{Q_s}
\| \, {\mathcal F}^{\, (s)}_q\, \psi \, \| ^2+ \frac 12\, \delta ^2c_6^2\, \| \psi 
\| ^2\, \leqslant 
$$ $$
\leqslant \, c_6^2\, (\, \delta +\max\limits_{q\, =\, 1,\dots ,Q_s}\, 
\beta _{\gamma ,\, \sigma }(\widehat V^{\, (s)}_q))^2\, \| \psi \| ^2\, ,\ \
s=0,1\, . 
$$
This completes the proof of Theorem \ref{th6.1}.

\end{document}